%% file: main.tex
\documentclass[11pt]{article}
\usepackage[a4paper,left=1in,right=1in,top=1in,bottom=1in]{geometry}
\usepackage{setspace}
\usepackage{graphicx} %
\usepackage{amsfonts,amssymb,graphics,epsfig,verbatim,bm,latexsym,amsmath,amsthm,url,amsbsy,multirow,rotating,enumerate,pgf,tikz,subfigure,color,xcolor,lscape, mathtools, float}
\DeclareRobustCommand\dashed{\tikz[baseline=-0.6ex]\draw[very thick,dashed] (0,0)--(0.34,0);}
\usetikzlibrary{fit}
\usepackage{tikz-cd}
\usepackage{bbm}
\usepackage{natbib}
\usepackage{hyperref}
\usepackage[section]{placeins}
\usetikzlibrary{positioning}
\usetikzlibrary{calc} %

\usepackage{enumitem}
\usepackage{cleveref}
\usepackage{booktabs}
\usepackage{authblk}
\usepackage{dsfont}
\usepackage{physics}
\usepackage{comment} 
\usepackage{siunitx} 

\usepackage{fullpage}

\input{commands}

\DeclareMathOperator{\ind}{\perp\!\!\!\!\perp} 

\DeclareMathOperator{\unif}{Unif}
\newcommand{\pre}{\operatorname{pre}}

\usepackage{listings}
\newtheorem{theorem}{Theorem}
\newtheorem{lemma}{Lemma}

\newtheorem{proposition}{Proposition}

\newtheorem*{remark}{Remark}

\newtheorem*{remarks}{Remarks}
\newtheorem{example}{Example}

\title{Distributional Instrumental Variable Method}
\author[1]{Anastasiia Holovchak\thanks{Address for correspondence: Anastasiia Holovchak, Seminar f\"ur Statistik, Department of Mathematics, ETH Z\"urich, R\"amistrasse 101, 8092 Z\"urich, Switzerland. \href{mailto:anastasiia.holovchak@stat.math.ethz.ch}{anastasiia.holovchak@stat.math.ethz.ch}}}
\author[1]{Sorawit Saengkyongam}
\author[2]{Nicolai Meinshausen}
\author[2]{Xinwei Shen}

\affil[1]{Seminar f\"ur Statistik, ETH Z\"urich}
\affil[2]{Department of Statistics, University of Washington}
\date{\today}

\begin{document}

\maketitle

\begin{abstract}
The instrumental variable (IV) approach is commonly used to infer causal effects in the presence of unmeasured confounding. 
Existing methods typically aim to estimate the mean causal effects, whereas a few other methods focus on quantile treatment effects. The aim of this work is to estimate the entire interventional distribution, which yields the classical causal estimands as functionals. We propose a method called Distributional Instrumental Variable (DIV), which uses generative modelling in a nonlinear IV setting. We establish identifiability of the interventional distribution under general assumptions and demonstrate an `under-identified' case, where DIV can identify the causal effects while two-step least squares fails to. 
Our empirical results show that the DIV method performs well for a broad range of simulated data, exhibiting advantages over existing IV approaches in terms of the identifiability and estimation error of the mean or quantile treatment effects. Furthermore, we apply DIV to an economic data set to examine the causal relation between institutional quality and economic development and our results align well with the original study. We also apply DIV to a single-cell data set, where we study the generalizability and stability in predicting gene expression under unseen interventions. The software implementations of DIV are available in \texttt{R} and \texttt{Python}.
\end{abstract}

\section{Introduction}

Understanding causal effects is crucial in many fields, for example to assess the influence of an economic policy~\citep{anderson2000,acemoglu2001} or a medical treatment~\citep{mcclellan1994,wang2005,muennig2011}, or to make predictions under unseen interventions~\citep{buhlmann2020invariance,christiansen2022}. In the absence of randomised experiments, it is often challenging to infer these effects due to unmeasured confounding factors that simultaneously affect both the treatment and the outcome. A central strategy to address such endogeneity is the instrumental variable (IV) approach, first proposed by~\citet{wright1928tariff} and subsequently developed in econometrics and statistics~\citep{hansen1982,newey2003,angrist1996}. An instrument is a variable that influences the treatment but has no direct effect on the outcome, thereby enabling consistent estimation of causal parameters. Classical estimators such as two-stage least squares (2SLS) and generalised method of moments (GMM) have been widely applied across economics, epidemiology, and the social sciences~\citep{angrist1991,card1999,staiger1997,chernozhukov2005}, making IVs a standard tool for causal inference.

Formally, we are interested in the causal relationship between treatment variables $X \in \mathbb{R}^d$ and response variables $Y \in \mathbb{R}^p$. There may exist a set of exogenous observable covariates $W \in \mathbb{R}^l$ that have an effect on $X$ or $Y$, or both. We also allow for the existence of unobserved variables $H \in \mathbb{R}^m$ (also known as hidden confounders) that can affect both $X$ and $Y$.
In the IV framework, we require observing an extra set of instrumental variables $Z \in \mathbb{R}^q$ that satisfy the following assumptions \citep{pearl2009}:
\begin{enumerate}[label=(\textit{A\arabic*})]
	\item relevance: the instrument $Z$ is not independent of the treatment variable $X$. \label{ass:a1} 
    \item exclusion restriction: $Z$ is independent of all error terms that have an influence on $Y$ that is not mediated by $X$. \label{ass:a3} 
\end{enumerate}
The precise formulation depends on the specific methodological framework.\footnote{For instance, for linear 2SLS, the relevance assumption \ref{ass:a1} corresponds to the full-rank condition, i.e.\ \( \operatorname{Cov}(Z, X) \) has a full column rank, implying  $q \geq d$. The exclusion restriction \ref{ass:a3} translates to \( \operatorname{Cov}(Z, Y - X^\top \beta) = 0 \).} If assumptions \ref{ass:a1}-\ref{ass:a3} are fulfilled, the instrument $Z$ is called a valid instrument, although the latter assumption is not testable with  observational data since $H$ is unobserved. 

All instrumental variable methods exploit exogenous heterogeneity in the data, introduced through an instrument, to identify the causal effect of $X$ on $Y$ despite hidden confounding. For example, two-stage least squares (2SLS) uses the instrument to first predict the variation in $X$ that is free from confounding and then employs this variation to estimate the causal effect of $X$ on $Y$. See Section~\ref{subsec:related_work} for a detailed review. Our proposal also shares a similar spirit, but we try to exploit as much information from the observed data as possible, which, in a statistical sense, means to estimate the full observed distribution of $X$, $Y$, $Z$, and $W$ (if any). We argue that in this way, we can make full use of the information existing in the observed data, particularly from the instruments, offering the best chance to identify the causal effect.

As for the estimand, traditional IV methods aim to infer the mean effect on the outcome of interest under interventions. Recent advances in causal inference have introduced methods for estimating quantile treatment effects \citep{chernozhukov2005, imbens2009noadditivity}.
In many applications, however, it is useful to understand the entire distribution of the outcome under interventions rather than summarizing them into single quantities. %
To this end, this paper considers the full interventional distribution as our estimand, that is, the distribution of $Y$ under a do-intervention \citep{pearl2009} on $X$, i.e.\ do$(X:=x)$, for any $x$, or, equivalently, the distribution of the potential outcome $Y(x)$ (see, e.g.\ \citet{rubin2005}). Classical causal estimands, such as the average or quantile treatment effects, are functionals of our estimand. 

Conventionally, estimation of a distribution is a challenging problem. In a regression context, to estimate the conditional distribution of a response $Y$ given predictors $X$, various distributional regression approaches have been developed through estimating the cumulative distribution function~\citep{foresi1995conditional,hothorn2014conditional}, density function~\citep{dunson2007bayesian}, quantile function~\citep{meinshausen2006quantile}, etc. 
Alternatively, one can model the target distribution in a generative form, i.e.\ $Y=g(X,\varepsilon)$, where $\varepsilon$ follows a standard Gaussian distribution and $g$ is a (generally) nonlinear function. Each generative model $g(x,\varepsilon)$ induces a conditional distribution of $Y|X=x$, and when the function class of $g$ is expressive enough, there exists a function that induces any target conditional distribution. \citet{shen2024engression} recently proposed a method called engression to find such a function based on proper scoring rules; see Section~\ref{sec:engression} for a review. Engression has been shown to be a very simple yet flexible method that can effectively estimate many real-world data distributions. 

We note that generative modelling naturally aligns with the structural causal model (SCM) framework~\citep{pearl2009}, which represents causal mechanisms through explicit functional relationships between variables. Building on this connection, we employ generative modelling to capture the underlying causal process and develop an estimation method based on engression. This leads to a novel contribution to instrumental variable methodology, enabling estimation of the entire interventional distribution of the outcome.

The remainder of the paper is organised as follows. Section~\ref{sec:related_work_prelim} reviews related literature on instrumental variable methods and summarises key preliminaries, including recent advances in distribution estimation via engression. Section \ref{sec:setting} describes the setup and shows some illustrative identifiability theory as motivations for our method. In particular, we point out that estimating the observed distribution is sufficient (and sometimes necessary) for identifying our causal estimand.
Section \ref{sec:DIV_method} proposes the Distributional Instrumental Variable (DIV) method.
We allow for general model classes for both the treatment and outcome models, avoiding the common restrictive assumption of additive noises.
Section \ref{sec:identifiability} presents the theoretical results on identifiability of the interventional distribution for three different model classes. In Section \ref{sec:experiments}, we conduct extensive simulation studies to examine the performance of DIV against benchmark methods in estimating the interventional means and quantiles, while Section~\ref{sec:realdata} shows two real-world applications of DIV to an economic data set and a single-cell data set. The empirical results support our theoretical findings and suggest that the DIV method has a wide range of potential applications.

We provide a  software implementation of the DIV method in the \texttt{R} package \texttt{DistributionIV}. The package allows for estimation of the interventional mean and quantiles, as well as sampling from the fitted interventional distribution. More details on the software, including an illustrative example, can be found in Appendix \ref{sec:software}. We also provide a basic \texttt{Python} package \texttt{DistributionIV}.

\section{Related Work and Preliminaries}\label{sec:related_work_prelim}

\subsection{Related work} \label{subsec:related_work}

In the linear case, two-stage least squares (2SLS) \citep{wright1928tariff, theil1958economic} and control function (CF) approach \citep{heckman1976cf, newey1999control} are two commonly used methods for causal effect estimation in presence of hidden confounders and a valid instrument, both leading to the same estimation results. In case of nonlinear models, these two methods produce different estimates (see, for example, \citet{guo2016} for a systematic comparison). The idea of 2SLS is to use the part of treatment that is independent of the hidden confounders for modelling the outcome. In contrast, the CF approach relies on `splitting' the hidden confounders into two parts - one that is correlated with the treatment $X$, and the other that is uncorrelated with $X$ (the latter part is then used as an independent covariate when modelling the outcome).

To allow for flexible modelling assumptions, several recent methods have incorporated deep neural networks into instrumental variable (IV) frameworks. For instance, DeepIV \citep{hartford2017deepIV} adopts a two-stage procedure, where the conditional distribution of the treatment is first estimated using a neural network, and the counterfactual prediction function is then learned in a second stage. DeepGMM \citep{bennett2020deepGMM} formulates the IV estimation problem through moment conditions and uses neural networks to solve them. HSIC-X \citep{saengkyongam2022exploiting} leverages a stronger independence restriction rather than traditional moment conditions.

For the methods stated above, the aim is to estimate the interventional mean $\mathbb{E}_Y^{\mathrm{do}(X \coloneqq x)}$ and it is assumed that the noise term in $Y$ (consisting of both, the independent noise term $\varepsilon_Y$ and the hidden confounder $H$) is additive.

\citet{imbens2009noadditivity} provide identification and estimation results for the interventional quantiles and interventional mean in a nonseparable model that does not require additive noise, assuming that the treatment variable X is scalar and continuous. Their approach uses the conditional CDF $V \coloneqq F_{X|Z}$ as a control variable but relies on a strong condition that the conditional support of $V$ given $X$ does not depend on $X$, and is therefore limited to continuous treatments.

\citet{torgovitsky2015} investigates identification in a related nonseparable model that allows for multivariate instrument $Z$ and treatment $X$ (with univariate outcome $Y$). The author shows that point identification can be achieved even when the instrument has small or discrete support by assuming scalar unobserved heterogeneity and strict monotonicity of both the outcome and first-stage functions with respect to their unobservables. These monotonicity assumptions serve as an alternative to the common support conditions imposed in the former work.

\citet{sanchez2020kneib} propose a flexible IV distributional regression method based on GAMLSS and the control function approach, but it still requires specifying an outcome distribution family, whereas DIV learns the interventional law nonparametrically. Recently, \citet{chernozhukov2024} introduced a copula-based distributional regression restricted to binary instruments, and \citet{kook2024} developed DIVE for interventional CDFs under a binary treatment (with absolutely continuous outcomes), whereas our approach accommodates general instruments and treatments and targets broader distributional effects without imposing a copula model.

\subsection{Engression for distribution estimation}\label{sec:engression}

We review the engression method by \citet{shen2024engression} in a distributional regression setting for estimating the conditional distribution of the response $Y$ given predictors $X$. As introduced earlier, we model the conditional distribution in a generative form. To find the desired generative model that induces the target distribution, we need a quantitative metric to assess each generative model according to observed data. 

The energy score used in engression is a commonly used proper scoring rule \citep{gneiting2007scoring}. For a random variable $Y$ that follows a distribution $P$ and an observation $y$,  the energy score is defined as 
\begin{equation*}
	S(P,y)=\frac{1}{2}\bbE\|Y-Y'\| - \bbE\|Y-y\|,
\end{equation*} 
where $Y$ and $Y'$ are two independent draws from $P$. A higher score value means a better distributional fit. The energy score is a strictly proper scoring rule~\citep{szekely2003statistics}, that is, given $P^*$ and for any $P$, we have
\begin{equation*}
	\bbE_{P^*}[S(P^*,Y)] \ge \bbE_{P^*}[S(P,Y)],
\end{equation*}
where the equality holds if and only if $P=P^*$. 

Now given the predictors $X=x$, we consider the energy score to assess conditional distributions. For a generative model $g(x,\varepsilon)$ with $\varepsilon$ following the standard Gaussian, denote by $P^g_{Y|X=x}$ its induced distribution, i.e.\ $g(x,\varepsilon)\sim P^g_{Y|X=x}$. The above inequality of the expected energy score becomes
\begin{equation*}
	\bbE_{P^*_{Y|X=x}}[S(P^*_{Y|X=x}, Y)] \ge \bbE_{P^*_{Y|X=x}}[S(P^g_{Y|X=x}, Y)],
\end{equation*}
where the equality holds if and only if $P^g_{Y|X=x}=P^*_{Y|X=x}$. This motivates \citet{shen2024engression} to design a loss function based on the right-hand side of the above inequality, leading to the definition of engression as the generative model that minimises the negative expected energy score:
\begin{equation*}
	g^*\in\argmin_g\left[\bbE\|Y-g(X,\varepsilon)\| - \frac12\bbE\|g(X,\varepsilon)-g(X,\varepsilon')\|\right],
\end{equation*}
where $\varepsilon$ and $\varepsilon'$ are two i.i.d.\ draws from the standard Gaussian. \citet[Proposition 1]{shen2024engression} show that $g^*(x,\varepsilon)\sim P^*_{Y|X=x}$ for a fixed $x$, when such a $g^*$ exists.

While generative models do not directly provide closed-form expressions for the conditional density or quantile functions, they offer a  way of sampling from the conditional distribution. This allows for Monte Carlo estimation of various functionals of the distribution, such as conditional means, variances, or quantiles. We will describe this concretely in the causal inference context in Section~\ref{subsec:estimation}.

\subsection{Notation}

For a vector $x \in \mathbb{R}^d$, let $\|x\|$ be the Euclidean norm. For a random variable $X$ and $\alpha \in [0,1]$, we denote $Q_\alpha(X) := \inf \{x: \mathbb{P}(X \leq x) \geq \alpha\}$ the $\alpha$-quantile of $X$, or simply $Q_\alpha^X$. For two random variables, i.e.\ $X$ and $X^\prime$, following the same distribution, we write $X \eqdist X^\prime$. For a random variable $X$ following a probability distribution $P$, we simply write $X \sim P$. The support of a random vector $B \in \Omega \subseteq \mathbb{R}^q$ (for some $q \in \mathbb{N}$) is defined as the set of all $b \in \Omega$ for which every open neighbourhood of $b$ (in $\Omega$) has positive probability. A random variable $X$ is called bounded if there exists a constant $R < \infty$ such that $\mathbb{P}(|X| \leq R)=1$. Equalities of conditional laws are understood in the sense of regular conditional distributions: statements of the form $(X \! \mid \! Z=z) \eqdist (X^\prime \! \mid \! Z=z)$ hold for $P_Z$-a.e.\ $z$. If $Z$ is discrete, $P_Z$-a.e.\ means for all $z$ with $\mathbb{P}(Z=z)>0$.

\section{Setting and motivating theory} \label{sec:setting}

In this section, we introduce our setting of general structural causal models (SCM), where our target of interest is the interventional distribution of the outcome $Y$ under a $\mathrm{do}$-intervention on the treatment $X$. To motivate the proposed method, we present some illustrative identifiability results.

\subsection{SCM and estimand}

We assume the observed data of $(X,Y,Z)$ is generated according to an underlying structural causal model
\begin{equation} \label{eq:IV_model}
\begin{aligned}
    X &\coloneqq g(Z, \eta_X) \\
    Y &\coloneqq f(X, \eta_Y),
\end{aligned}
\end{equation}
where $Z \in \bbR^q$, $X \in \bbR^d$, $Y \in \bbR^p$, $Z$ is exogenous and independent of noise variables $(\eta_X,\eta_Y)$, while $\eta_X \in \bbR^d$ and $\eta_Y \in \bbR^p$ are generally correlated due to unobserved confounding between $X$ and $Y$, and functions $g$ and $f$ are generally nonlinear to allow for more complex relationships both between the instrument $Z$ and treatment $X$, and between the treatment $X$ and outcome $Y$. Throughout the paper, we assume all noise variables $\eta_X, \eta_Y$ to be absolutely continuous with respect to the Lebesgue measure, unless indicated otherwise. The SCM \eqref{eq:IV_model} induces the observational distribution $P_{(X,Y,Z)}$ over the observed variables $(X,Y,Z)$.

Our estimand is the interventional distribution $\Pint$ for all $x$ in the support of $X$. Note that the conventional estimands are functionals of our distributional estimand. For example, for some $x_1, x_0 \in \supp(X)$, the average treatment effect is defined as a contrast of its means: $\mathbb{E}_Y^{\mathrm{do}(X \coloneqq x_1)} - \mathbb{E}_Y^{\mathrm{do}(X \coloneqq x_0)}$; the $\alpha$-quantile treatment effect is a contrast of its $\alpha$-quantiles $Q_\alpha(P_Y^{\mathrm{do}(X \coloneqq x_1)})-Q_\alpha(P_Y^{\mathrm{do}(X \coloneqq x_0)})$.

\subsection{Identifiability}
We present an identification result for $\Pint$ to motivate our method introduced in Section~\ref{sec:DIV_method}. For the ease of illustration, here we consider a simplified case where all observed variables $(X,Y,Z)$ are univariate. More general and comprehensive identification results will be given in Section~\ref{sec:identifiability}.
\begin{proposition}\label{prop:univariate_identify}
 Consider the model in \eqref{eq:IV_model} and suppose the following assumptions hold:
 \begin{enumerate}[label=(\textit{\roman*})]
     \item For all $z \in \supp(Z)$, it holds that $g(z,\cdot)$ is strictly monotone. \label{ass:univar_i}
     \item For all $x \in \supp(X)$, $\supp(\eta_X|X=x) = \supp(\eta_X)$. \label{ass:univar_iii}
 \end{enumerate}

Then, for all $x \in \supp(X)$, the interventional distribution $\Pint$ is uniquely determined from the observed data distribution $P_{(X,Y)|Z}$.
\end{proposition}

Assumption \ref{ass:univar_iii} is known as the common support assumption (see \citet{imbens2009noadditivity}), and requires the instrument $Z$ to affect the treatment $X$ and exhibit sufficient variation. This assumption aligns with  the relevance assumption \ref{ass:a1}, which requires that the instrument $Z$ is associated with the treatment $X$. Note that assumptions we make here are similar to those proposed by \citet{imbens2009noadditivity}. However, in Section~\ref{sec:identifiability}, we address a setting where both $X$ and $Y$ are multivariate, whereas their model class is limited to a univariate $X$. Moreover, we present novel identifiability results demonstrating that by adding more structural restrictions to the outcome model, the interventional distribution becomes identifiable under strictly weaker assumptions. More detailed remarks on the assumptions will be given in Section~\ref{sec:identifiability}.

Proposition~\ref{prop:univariate_identify} indicates that, under certain assumptions, the target interventional distribution is uniquely identifiable from the observed joint distribution of $(X,Y)$ given $Z$. This suggests that an estimation method for $\Pint$ should fit the distribution of $(X,Y) | Z$ from the observed data while ensuring consistency with the SCM \eqref{eq:IV_model}, thereby enabling identifiability (e.g.\ exogeneity of $Z$).  Before specifying the methodology, we would like to emphasise that matching the full distribution can be also `necessary' (in some cases) for identifying the above estimand. For example, classical IV regression, which estimates only conditional means, fails to achieve identification when the number of instruments is smaller than that of the treatment variables --- a situation typically referred to as an under-identified setting. In contrast, leveraging the full distribution allows for identification in cases where classical IV regression fails. Below is a simple example to illustrate the failure of 2SLS for identifying the causal effects. A formal identifiability result is given in Section~\ref{sec:identifiability} for settings with a single binary (or discrete) instrument $Z$ and multivariate treatment. 

Consider $Z \in \{0, 1\}$. Assume the data generating process follows the SCM

\begin{equation} \label{eq:SCM_underident}
\begin{aligned}
    X_1 &\coloneqq g_1(Z, \eta_{X_1}) \\
    X_2 &\coloneqq g_2(Z, \eta_{X_2}) \\
    Y &\coloneqq \beta_1 X_1 + \beta_2 X_2 + \eta_Y,
\end{aligned}
\end{equation}
where $X_1, X_2, Y, \eta_X, \eta_Y \in \bbR$. We first show that 2SLS procedure fails identifying the interventional mean.

\begin{example}[Failure of 2SLS] \label{ex:failure_2sls}
In the first stage, the treatments $X_1,X_2$ are regressed on $Z$. For a binary $Z$, the conditional mean of $X_i$ given $Z=z$ can always be written as a linear function of $z$: 
\begin{equation*}
    \bbE(X_i|Z=z) = (1 - z) \cdot \bbE[g_i(0,\eta_{X_i})] + z \cdot \bbE[g_i(1,\eta_{X_i})] = c_i + \alpha_i z,
\end{equation*}
where $c_i=\bbE(g_i(0,\eta_{X_i}))$ and $\alpha_i=(\bbE(g_i(1,\eta_{X_i})) - \bbE(g_i(0,\eta_{X_i})))$. Let $\hat{X}_1 \coloneqq \bbE(X_1|Z=z) = c_1 + \alpha_1 z$ and $\hat{X}_2 \coloneqq \bbE(X_2|Z=z) = c_2 + \alpha_2 z$.
    
    In the second stage, $Y$ is regressed on $\hat{X}_1$ and $\hat{X}_2$. Due to multicollinearity of $\hat{X}_1$ and $\hat{X}_2$, the parameter estimates are not well-defined, resulting in the non-identifiability of the causal effects $\beta_1$ and $\beta_2$.
\end{example}

In contrast, with the following proposition we demonstrate a novel result that with the distributional instrumental variable approach, the parameters $\beta_1$ and $\beta_2$, and therefore the interventional distribution $\Pint$ can still be identified under certain assumptions.

\begin{proposition}\label{prop:binary_instrument_ident}
    Assume for $j \in \{1, 2\}$, it holds for all $z \in \supp(Z)$ that $g_j(z,\cdot)$ is strictly monotone and differentiable almost everywhere, and for any constant $c$, it holds $(X_j|Z=0) \noneqdist (c + X_j|Z=1)$. Then $\beta_1$ and $\beta_2$ are uniquely determined from the observed data distribution $P_{(X_1,X_2,Y)|Z}$.
\end{proposition}
The assumption in this proposition means that the distributions of $X_j|Z=0$ and $X_j|Z=1$ are different in more than just a deterministic shift. The proposition indicates that if the instrument affects the treatments in more than just a mean shift,  then the full distribution $P_{(X_1,X_2,Y)|Z}$ of the observed data is sufficient to identify the causal effects, whereas 2SLS, which only exploits the conditional means of $X_j|Z=0$ and $X_j|Z=1$ in its first stage, cannot make use of the more diverse information even if it exists. In this sense, utilizing the full observed distribution allows us to identify the causal effects in this setting. If instead the two conditional distributions differ only by an additive shift, the assumption is violated and the causal effects are not identifiable in general.

\section{DIV method}\label{sec:DIV_method}

The previous section suggests the sufficiency and necessity of fitting the conditional distribution of $(X,Y)|Z$ for identifying and estimating the interventional distribution. In this section, we propose our DIV approach to realise this idea.

\subsection{Joint generative model}
Note that our SCM in \eqref{eq:IV_model} is a generative model for the underlying data distribution, where the noise variables associated with the treatment and outcome, $\eta_X$ and $\eta_Y$, are correlated. We propose to retain this generative form for our model class, while allowing the noise variables to be correlated, yielding the following joint generative model:
\begin{equation}\label{eq:joint_div_model}
	\begin{aligned}
		\eta_X &= h_X(\varepsilon_X,\varepsilon_H)\\
		\eta_Y &= h_Y(\varepsilon_Y,\varepsilon_H)\\
		X &= g(Z,\eta_X)\\
		Y &= f(X,\eta_Y)
	\end{aligned}
\end{equation}
where the correlated noise variables $\eta_X$ and $\eta_Y$ are parametrised as two functions (to be learned) of an independent noise term $\varepsilon_X\in\bbR^d$ and $\varepsilon_Y\in\bbR^p$, respectively, and a shared noise $\varepsilon_H\in\mathbb{R}^{\min\{d,p\}}$ to capture the correlation induced by latent confounders; all of them are assumed to follow the standard Gaussian distribution without loss of generality. Figure~\ref{fig:div_model} provides a graphical representation of the DIV model, illustrating the relationships between the instrumental variable  $Z$, treatment $X$, outcome $Y$, and the associated noise components. %

\begin{figure}[!ht]
\centering
\begin{tikzpicture}[node distance={18mm},main/.style = {draw, circle, minimum size=0.75cm}]
	\node[main,inner sep=1.5pt](x)  {$X$};
	\node[main,inner sep=1.5pt](y)[right of=x] {$Y$};
	\node[main,inner sep=1.5pt](z)[above left of=x] {$Z$};
	\node[main,inner sep=1pt, dashed](ex)[below of=x] {$\eta_X$};
	\node[main,inner sep=1pt, dashed](ey)[below of=y ] {$\eta_Y$};
	\node[main,inner sep=1pt, dashed](epsx)[below left of=ex] {$\varepsilon_X$};
	\node[main,inner sep=1pt, dashed](epsy)[below right of=ey] {$\varepsilon_Y$};
    \node[main,inner sep=1pt, dashed] (epsh) at ($(epsx)!0.5!(epsy)$) {$\varepsilon_H$};
	\draw[->] (x) -- (y);
	\draw[->] (z) -- (x);
	\draw[->] (ex) -- (x);
	\draw[->] (ey) -- (y);
	\draw[->] (epsx) -- (ex);
	\draw[->] (epsy) -- (ey);
	\draw[->] (epsh) -- (ex);
	\draw[->] (epsh) -- (ey);
\end{tikzpicture}
\caption{Graphical representation of the DIV model, depicting the generative structure used for estimating the joint distribution of $(X,Y) | Z$. Observed variables are depicted using solid circles, and dashed circles correspond to sampled/modelled noise components.}
\label{fig:div_model}
\end{figure}

\subsection{Distributional objective and DIV solution}
Our estimation approach uses the expected negative energy score \citep{gneiting2007scoring} as a loss function to train the conditional generative model. As mentioned in Section~\ref{sec:engression}, the energy score is a strictly proper scoring rule originally used for evaluation of multivariate distributional forecasts (for more details, see Appendix~\ref{app:energy_score}). For any two distributions $P$ and $P_0$, the energy loss is defined as 
\begin{equation}
    \mathcal{L}_e(P, P_0) \coloneqq \mathbb{E}_{Y \sim P_0}[-S(P, Y)] = \mathbb{E}_{Y \sim P_0, U \sim P} \|U - Y\| - \frac{1}{2} \mathbb{E}_{U, U^\prime \sim P}\|U - U^\prime\|,    
\end{equation}
where $Y \sim P_0$, $U$ and $U^\prime$ are two independent draws from $P$.

In principle, other proper scoring rules or distributional distances could be used to fit the generative model.
We use the energy score because it yields a simple sample-based objective that works well empirically for conditional generative modelling \citep{shen2024engression} and admits asymptotic efficiency guarantees \citep[Theorem~7]{shen2025reverse}.

Let $(\hat{X},\hat{Y})$ and $(\hat{X}',\hat{Y}')$ be two independent samples from the joint distribution of $(X,Y)|Z$ induced by the joint generative model \eqref{eq:joint_div_model}, obtained by
\begin{equation*}
	\begin{aligned}
		\hat{X} &\coloneqq g(Z,h_X(\varepsilon_X,\varepsilon_H))\\
		\hat{X}' &\coloneqq g(Z,h_X(\varepsilon_X',\varepsilon_H'))\\
		\hat{Y} &\coloneqq f(\hat{X},h_Y(\varepsilon_Y,\varepsilon_H))\\
		\hat{Y}' &\coloneqq f(\hat{X}',h_Y(\varepsilon_Y',\varepsilon_H'))
	\end{aligned}
\end{equation*}
with $\varepsilon_X,\varepsilon_Y,\varepsilon_H,\varepsilon_X',\varepsilon_Y',$ and $\varepsilon_H'$ being independently drawn from standard Gaussians. 
We then define the population version of the DIV solution as a minimiser of the energy loss
\begin{equation}\label{eq:joint_div_loss}
	(g^*,f^*,h^*_X,h^*_Y)\in\argmin_{f,g,h_X,h_Y} \bbE\left[\|(X,Y) - (\hat{X},\hat{Y})\| - \frac12\|(\hat{X},\hat{Y}) - (\hat{X}',\hat{Y}')\|\right].
\end{equation}
We parameterise $g,f,h_X,h_Y$ by neural networks. The population objective \eqref{eq:joint_div_loss} is approximated by its empirical counterpart using random draws of the latent noise variables, and the resulting differentiable loss is minimised by stochastic gradient descent, mirroring the optimisation setup used in engression \citep{shen2024engression}.

We show in the following proposition that the DIV solution induces the distribution over $(X,Y)|Z$ that matches the underlying distribution $P_{(X,Y)|Z}$. This result allows us to identify the interventional distribution from the DIV solution under suitable assumptions, as we present in Section~\ref{sec:identifiability}.

\begin{proposition}\label{prop:popul_sol_joint}
	The DIV solution defined in \eqref{eq:joint_div_loss} satisfies 
    \begin{center}$\Big( (g^*(Z,\eta^*_X),f^*(g^*(Z,\eta^*_X),\eta^*_Y))|Z=z \Big) \eqdist \Big((X,Y)|Z=z \Big)$\end{center} 
    for all $z\in\supp(Z)$, where $\eta^*_X=h^*_X(\varepsilon_X,\varepsilon_H)$ and $\eta^*_Y=h^*_Y(\varepsilon_Y,\varepsilon_H)$.
\end{proposition}

\subsection{Estimation of the interventional distribution and its functionals} \label{subsec:estimation}

Once a DIV model is fitted, we estimate the interventional distribution via sampling due to its generative model nature. Note that a do-intervention, $\mathrm{do}(X \! \coloneqq \! x)$, removes the dependency between $X$ and $\eta_Y$. That is, in the SCM \eqref{eq:IV_model}, $X$ is set to a fixed value $x$, while $\eta_Y$ follows its marginal distribution.

Thus, we propose the following sampling procedure that produces samples from the target interventional distribution: for any fixed $x$, we (i) sample $\varepsilon_Y, \varepsilon_H$ from standard Gaussians, (ii) compute the noise variable $\eta^*_Y=h^*_Y(\varepsilon_Y, \varepsilon_H)$, and (iii) obtain a sample $Y^*=f^*(x,\eta^*_Y)$. We will show below in Section~\ref{sec:identifiability} that, under suitable assumptions, the sample $Y^*$ obtained in this way indeed follows the interventional distribution $P_Y^{\mathrm{do}(X \coloneqq x)}$.

Based on samples from the interventional distribution, one can directly estimate its various characteristics, such as the interventional mean or quantiles. 
At the population level, the DIV estimator of \textit{interventional mean function} is derived from 
\begin{equation}
    \mu^*(x) := \mathbb{E}_{\varepsilon_H, \varepsilon_Y}[f^*(x,\varepsilon_H, \varepsilon_Y)].
\end{equation}
The DIV estimator of the \textit{interventional median function} is
\begin{equation}
    m^*(x) := Q_{0.5}[f^*(x,\varepsilon_H, \varepsilon_Y)],
\end{equation}
where the quantile is taken with respect to $(\varepsilon_H, \varepsilon_Y)$. More generally, for any $\alpha \in [0,1]$, the DIV estimator for the \textit{interventional quantile function} is 
\begin{equation}
    q^*_{\alpha}(x) := Q_{\alpha}[f^*(x,\varepsilon_H, \varepsilon_Y)].
\end{equation}

For a finite sample, based on the empirical solution $\hat{f}$, the corresponding estimators are constructed by sampling. To do so, for any $x$, we sample $(\varepsilon_{H,j}, \varepsilon_{Y,j})$, $j = 1,...,m$ where $m$ some positive constant, and then obtain $\hat{f}(x,\varepsilon_{H,j},\varepsilon_{Y,j})$, $j = 1,...,m$. These form an i.i.d.~sample from the estimated interventional distribution. All point estimates are then computed using the empirical versions of the estimators from this sample. This means, we compute the interventional mean by $\frac{1}{m}\sum_{j=1}^m \hat{f}(x,\varepsilon_{H,j},\varepsilon_{Y,j})$. Correspondingly, the interventional quantiles (in particular, the median) are estimated by the sample quantiles of $\hat{f}(x,\varepsilon_{H,j},\varepsilon_{Y,j})$, $j = 1,...,m$.

As an illustrative example, we consider an IV model as defined in \eqref{eq:IV_model}, with $g$ and $f$ both being nonlinear softplus functions. Figure~\ref{fig:interv_sample_quant} shows samples from the true interventional distribution $P_Y^{\mathrm{do}(X:=x)}$ and the estimated interventional distribution $\hat{P}_Y^{\mathrm{do}(X:=x)}$, which visually appear to closely match, along with the true and the estimated interventional quantile functions $q_\alpha^*(x)$ and $\hat{q}_\alpha^*(x)$ for $\alpha \in \{0.1, 0.5, 0.9\}$, which also show only marginal discrepancies. In Figure~\ref{fig:kernel_dens}, we present kernel density estimates based on samples from the true and the estimated interventional distributions at three distinct values of $x$. 

\begin{figure}[!ht]
    \begin{minipage}{0.48\textwidth}
        \centering
        \includegraphics[width=1.0\textwidth,height=0.8\textheight,keepaspectratio]{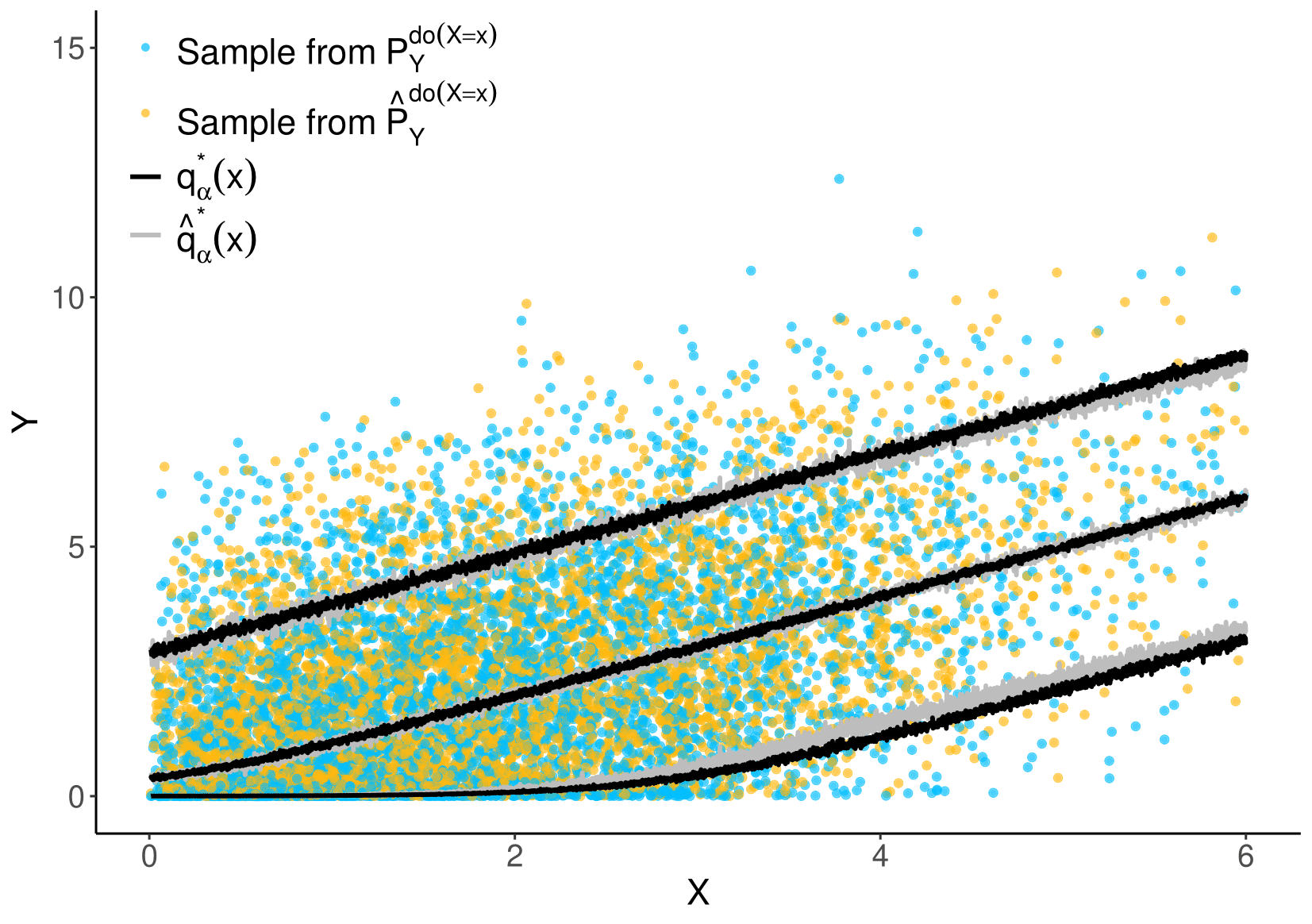}
		\caption{Samples from $P_Y^{\mathrm{do}(X\coloneqq x)}$ (blue) and $\hat{P}_Y^{\mathrm{do}(X\coloneqq x)}$ (yellow) along with interventional quantile functions $q^*_{\alpha}(x)$ and $\hat{q}^*_{\alpha}(x)$ for $\alpha\!\in\!\{0.1, 0.5, 0.9\}$}
        \label{fig:interv_sample_quant}
    \end{minipage}
    \hfill
    \begin{minipage}{0.48\textwidth}
        \centering
        \includegraphics[width=1.0\textwidth,height=0.8\textheight,keepaspectratio]{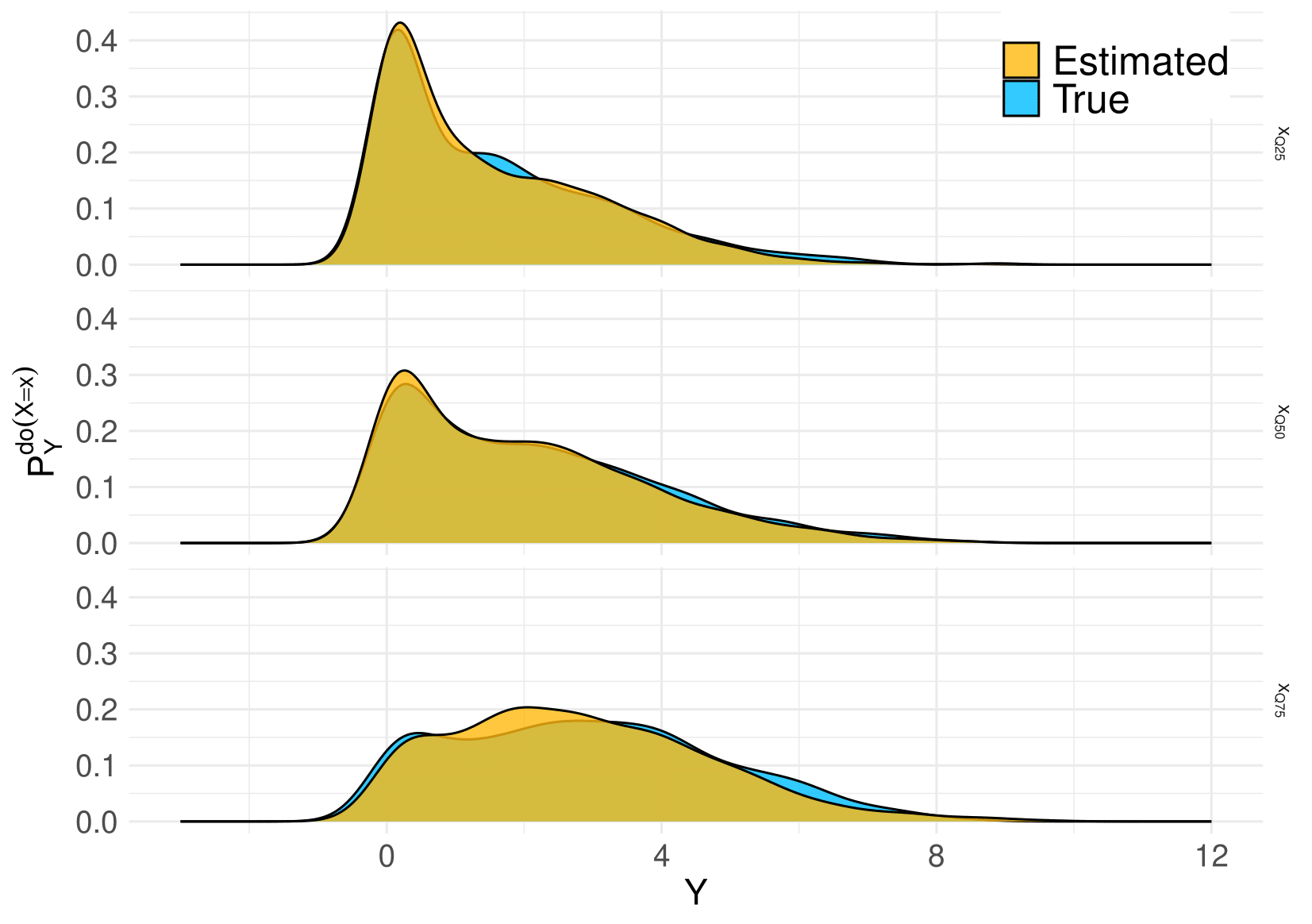}
		\caption{Kernel density estimates based on samples from $P_Y^{\mathrm{do}(X\coloneqq x)}$ (blue) and $\hat{P}_Y^{\mathrm{do}(X\coloneqq x)}$ (yellow) at training data quantiles $x\!\in\!\{x_{Q25}, x_{Q50}, x_{Q75}\}$ with $1000$ samples per $x$}
        \label{fig:kernel_dens}
    \end{minipage}
\end{figure}

Besides that, DIV not only estimates the interventional distribution $\Pint$ but also provides an estimation of the joint observational distribution $P_{(X,Y)}$ at no additional cost. This aspect is discussed in more detail, along with empirical results, in Appendix~\ref{app:Pobs_Pint}.

\subsection{Conditional interventional distribution} \label{subsec:conditional_interventional}

The DIV method can be directly adapted to incorporate additional exogenous covariates $W \in \bbR^l$ that affect both the treatment $X$ and the outcome $Y$, and the estimand becomes the conditional interventional distribution $P_{Y|W=w}^{\mathrm{do}(X \coloneqq x)}$, which can be used to obtain conventional estimands such as conditional average treatment effects or conditional quantile effects. Specifically, we augment the joint generative model \eqref{eq:joint_div_model} by adding $W$ into the treatment and outcome models, i.e.\ $X=g(Z,W,\eta_X)$ and $Y=f(X,W,\eta_Y)$. For estimation, we learn the DIV model to fit the joint distribution of $(X,Y)|Z,W$, which leads to the same objective function as in \eqref{eq:joint_div_loss} where samples $\hat{X},\hat{X}',\hat{Y},\hat{Y}'$ also depend on $W$ now. All the identification results developed in the next section can be readily extended to this setting, which guarantees the identification of the new estimand by the heterogenous adaption of DIV. Our \texttt{R} implementation also supports this scenario. 

Furthermore, in some cases where the instrument $Z$ is not exogenous, incorporating additional covariates $W$ could still render $Z$ (conditionally) exogenous and facilitate identifiability. This has been studied in the setting of conditional IV (see, e.g.\ \citep{conditional_iv}). While some existing methods may not be directly applicable in the conditional IV setting (see, e.g.\ Section~2.1 in \cite{saengkyongam2022exploiting}), our approach can be naturally extended to accommodate such cases.

\section{Identifiability and consistency} \label{sec:identifiability}

In this section, we present conditions under which a model from a certain model class, which induces a joint distribution $P_{(X,Y)|Z=z}$, is unique. This is referred to as the identifiability of the model class. Note that we are primarily interested in the identifiability of the interventional distribution $P_Y^{\mathrm{do}(X:=x)}$, which follows from the identifiability of the model class. 

In the following, we distinguish three model configurations and present the assumptions needed to ensure the identifiability of the interventional distribution. We first show the identifiability of the general model class $\mathcal{M}_{DIV}$ in Section~\ref{subsec:general_model_identifiability}, requiring the instrument~$Z$ to have a large support (see Remark~\ref{rem:2}). Further, in Section~\ref{subsec:pre-ANM_identifiability} we relax the large support condition by restricting the outcome model class to the pre-additive noise models $\mathcal{M}_{DIV}^{\pre}$, and present the results for two cases: the instrument $Z$ being continuous (but not requiring large support) and discrete (with binary instrument as a special case). Table~\ref{tab:ident_assumptions} below provides an overview of the identifiability results.

\begin{table}[t!]
\centering
\scriptsize
\begin{tabular}{|>{\raggedright\arraybackslash}p{2.3cm}|>
{\raggedright\arraybackslash}p{2.3cm}|>
{\raggedright\arraybackslash}p{2.6cm}|>
{\raggedright\arraybackslash}p{2.0cm}|>
{\raggedright\arraybackslash}p{2.3cm}|>
{\centering\arraybackslash}p{1.7cm}|}
\hline
\textbf{Condition on $\cF$} & \textbf{Condition on $\cG$} & \textbf{Instrument Type} & \textbf{\#Instruments} & \textbf{Assumptions} & \textbf{Theorem}\\ \hline
General class &  General class & Continuous \& large support & Require $q \geq d$ & \ref{ass:1}, \ref{ass:3}, \ref{ass:4} & \ref{thm:interv_ident_general} (Sec~\ref{subsec:general_model_identifiability}) \\ 
Pre-ANM class & General class & Continuous & Allow for $q < d$ & \ref{ass:c1}-\ref{ass:c5} & \ref{thm:interv_ident_cont_preadd} (Sec~\ref{subsec:pre-ANM_identifiability}) \\ 
Pre-ANM class & Strictly nonlinear & Discrete & Allow for $q < d$ & \ref{ass:d1}-\ref{ass:d4} & \ref{thm:interv_ident_bin_preadd} (Sec~\ref{subsec:pre-ANM_identifiability}) \\ \hline
\end{tabular}
\caption{Overview of some of the identifiability results for the interventional distribution $\Pint$, described in Section~\ref{sec:identifiability}.}
\label{tab:ident_assumptions}
\end{table}

\subsection{General model class} \label{subsec:general_model_identifiability}

Let $\mathcal{M}_{\rm DIV}$ be the class of structural causal models of the form:
\begin{equation} \label{eq:IV_model2}
\begin{aligned}
    X_j &\coloneqq g_j(Z, \eta_{X_j}), \forall \ j \in \{1,\dots, d\} \\
    Y_k &\coloneqq f_k(X, \eta_{Y_k}), \forall \ k \in \{1,\dots, p\},
\end{aligned}
\end{equation}
where $Z \sim Q_Z$ exogenous, $\eta_X \coloneqq (\eta_{X_1}, \dots,\eta_{X_d})$, $\eta_Y \coloneqq (\eta_{Y_1}, \dots,\eta_{Y_p})$ with $(\eta_X, \eta_Y) \sim Q_{(X,Y)}$, $Z \in \bbR^q$ and $(\eta_X, \eta_Y)$ being independent. Further, we define $X \coloneqq (X_1, \dots, X_d), Y \coloneqq (Y_1, \dots, Y_p)$,  for all $j \in \{1,\dots, d\}: g_j \in \cG$, for all $k \in \{1,\dots, p\}: f_k \in \cF$, and $\mathcal{G} \subseteq \{g: \bbR^{q+1} \rightarrow \bbR\}$, $\mathcal{F} \subseteq \{f: \bbR^{d + 1} \rightarrow \bbR\}$ are function classes.

Let a model in \eqref{eq:IV_model2} from $\mathcal{M}_{\rm DIV}$ satisfy the following conditions, which we will discuss after stating the results.

\begin{enumerate}[label=(\textit{B\arabic*})]
 	\item For all $g \in \cG$, it holds for all $z \in \supp(Z)$ that $g(z,\cdot)$ is strictly monotone on  $\supp(\eta_X)$. \label{ass:1}
    \item For all $f \in \cF$, it holds for all $x \in \supp(X)$ that $f(x,\cdot)$ is strictly monotone on  $\supp(\eta_Y)$. \label{ass:2}
    \item For all $j \in \{1,\dots d\}, k \in \{1,\dots, p\}$, the noise terms $\eta_{X_j}$ and $\eta_{Y_k}$ are absolutely continuous with respect to the Lebesgue measure. \label{ass:3}
    \item For all $x \in \supp(X)$, $\supp(\eta_X|X=x) = \supp(\eta_X)$. \label{ass:4}
\end{enumerate}

Since we are primarily interested in the distribution of the outcome $Y$ under an intervention on the treatment $X$, we first present the theorem showing identifiability of the interventional distribution $P_Y^{\mathrm{do}(X:=x)}$.

\begin{theorem} \label{thm:interv_ident_general}
    Consider the model in \eqref{eq:IV_model2} and suppose the assumptions \ref{ass:1}, \ref{ass:3} and \ref{ass:4} hold. For all $x \in \supp(X)$, the interventional distribution $P_Y^{\mathrm{do}(X:=x)}$ is then identifiable from the observed data distribution $P_{(X,Y)|Z}$. 
    \begin{proof}
        See Appendix~\ref{app:ident_gen_mod}.
    \end{proof}
\end{theorem}
Next, we present the main result on the identifiability of the treatment model, the response model, and also the confounding effect for the general model class \eqref{eq:IV_model2}. Note that the identifiability of the response model in (b), together with (c), implies the identifiability of the interventional distribution, which is concisely stated in Theorem~\ref{thm:interv_ident_general}. It also justifies the DIV approach. By learning the observed distribution of $(X,Y)|Z$, DIV is able to identify, under certain assumptions, the true outcome model, $f^*(x,\eta^*_Y)$, which then induces the interventional distribution $P_Y^{\mathrm{do}(X:=x)}$.

\begin{proposition} \label{prop:ident_general}

Consider the model in \eqref{eq:IV_model2}. Suppose the assumptions \ref{ass:1}-\ref{ass:4} hold. For any two models $(g_j,f_k,\eta_X,\eta_Y)$, $(\tilde{g}_j,\tilde{f}_k,\tilde{\eta}_X,\tilde{\eta}_Y) \in \mathcal{M}_{DIV}$ that induce the same conditional distribution of $(X,Y)$ given $Z=z$, it then holds 
    \begin{enumerate}[label=(\textit{\alph*})]
        \item for all $j \in \{1,\dots d\}$, $z \in \supp(Z)$, $e_{X} \in \supp(\eta_{X_j})$ we have $g_j(z,e_{X}) = \tilde{g}_j(z,e_{X})$,
        \item for all $k \in \{1,\dots, p\}$, $x \in \supp(X)$, $e_Y \in \supp(\eta_{Y_k})$ we have $f_k(x,e_{Y}) = \tilde{f}_k(x,e_{Y})$,
        \item $(\eta_X, \eta_Y) \eqdist (\tilde{\eta}_X, \tilde{\eta}_Y)$. 
    \end{enumerate}

    \begin{proof}
        See Appendix~\ref{app:ident_gen_mod}.
    \end{proof}  
\end{proposition}

\begin{remarks}
We now discuss the assumptions we made. 
    \begin{enumerate}
        
        \item \label{rem:1} Assumption \ref{ass:3}. For all $j \in \{1,\dots ,d\}$ and for all $k \in \{1,\dots, p\}$, correspondingly, we assume $\eta_{X_j}$ and $\eta_{Y_k}$ being absolutely continuous with respect to the Lebesgue measure. Without loss of generality, it can then be assumed that the marginals $\eta_{X_j}, \eta_{Y_k} \sim N(0,1)$. By applying the Sklar's theorem \citep{sklar1959fonctions} and using the invariance property of the copula with respect to strictly monotone transformations on the components of a continuous random vector (See, for example, Proposition 5.6. of \citet{mcneil2005quantitative}), we can express an arbitrary joint distribution of $(\eta_X, \eta_Y)$ using copula and the marginal standard Gaussians.
        \item \label{rem:2} Assumption \ref{ass:4}. For the common support assumption to be satisfied, the instrumental variable $Z$ must affect $X$ and also have a sufficiently large support. The assumption directly corresponds to the relevance assumption \ref{ass:a1}, which is one of three core assumptions made within the instrumental variable approach. Assuming the treatment model to be linear, say $X = M_0 Z + \eta_X$, with $M_0 \in \mathbb{R}^{d \times k}$ being the coefficient matrix, the common support assumption directly corresponds to $M_0$ being full row rank and $\supp(M_0 Z) = \mathbb{R}^d$.
    \end{enumerate}
\end{remarks}

\begin{remark}[Binary treatment]
    We can adapt Theorem~\ref{thm:interv_ident_general} to the case when the treatment $X$ is binary. Assumption~\ref{ass:3} has to be changed to:
    \begin{enumerate}[label=(\textit{B\arabic*}*)]
\setlength{\itemsep}{0pt}
\setcounter{enumi}{2}
    \item For all $j \in \{1,\dots, d\}$, we assume the noise terms \( \eta_{X_j} \sim \text{Bernoulli}(p) \), \(0 < p < 1\). For all $k \in \{1,\dots, p\}$, the noise terms $\eta_{Y_k}$ are absolutely continuous with respect to the Lebesgue measure. \label{ass:3*}
\end{enumerate}

\noindent The proof proceeds analogously to that of Theorem 1, with summation replacing integration for the discrete component. Assumption~\ref{ass:1} on strict monotonicity of $g(z,\cdot)$ now reduces to $g(z,1)\neq g(z,0)$ for all $z \in \supp(Z)$ fixed, ensuring that $X$ varies with the noise $\eta_X$.

\end{remark}

\paragraph{Relation to existing identification results.}
Our identification strategy builds on ideas from nonseparable triangular models (see \citet{imbens2009noadditivity}), but differs in two key respects. \emph{First}, rather than invoking a control-variable construction that hinges on rich continuous variation in the instrument, we impose strict monotonicity in the latent disturbances in both the first stage and the outcome, together with a common-support condition on \(P_{X\mid Z}\). This combination accommodates instruments with either continuous or discrete support. \emph{Second}, relative to \citet{torgovitsky2015}, who uses scalar latent monotonicity to obtain point identification with small-support instruments, our framework (i) allows \(X\) and \(Y\) to be multivariate, (ii) identifies the confounding effect (Proposition~\ref{prop:ident_general}), and (iii) links the identification argument directly to a generative estimation procedure for the full interventional distribution.

\subsection{Pre-additive noise model class} \label{subsec:pre-ANM_identifiability}
Theorem~\ref{thm:interv_ident_general} provides the identifiability for a general model class. The price to be paid is that we require a relatively strong assumption for the instruments (i.e.\ the common support assumption \ref{ass:4}). This section presents an identifiability result that relaxes this assumption by considering a more restricted outcome model class, namely pre-additive noise models (pre-ANMs). Pre-ANMs have been used in previous work to facilitate identifiability in other settings (e.g.\ \citet{shen2024engression}).

Let $\mathcal{M}_{\rm DIV}^{\pre}$ be the class of structural pre-additive noise IV (pre-ANM) causal models of the form:

\begin{equation} \label{eq:pre-ANM}
\begin{aligned}
    X_j &\coloneqq g_j(Z, \eta_{X_j}), \forall \ j \in \{1,\dots, d\} \\
    Y_k &\coloneqq f_k(X^\top \beta_k + \eta_{Y_k}), \forall \ k \in \{1,\dots, p\},
\end{aligned}
\end{equation}
where $Z \sim Q_Z$ exogenous, $\eta_X \coloneqq (\eta_{X_1}, \dots,\eta_{X_d})$, $\eta_Y \coloneqq (\eta_{Y_1}, \dots,\eta_{Y_p})$ with $(\eta_X, \eta_Y) \sim Q_{(X,Y)}$, with $Z$ and $(\eta_X, \eta_Y)$ being independent. Further, we define $X \coloneqq (X_1, \dots, X_d)$, $Y \coloneqq (Y_1, \dots, Y_p)$, $\beta_k = (1, \beta_{k,2}, \dots, \beta_{k,d})$, for all $j \in \{1,\dots, d\}: g_j \in \tilde{\cG}$, for all $k \in \{1,\dots, p\}: f_k \in \tilde{\cF}$, and $\tilde{\mathcal{G}} \subseteq \{g: \bbR \rightarrow \bbR\}$, $\tilde{\mathcal{F}} \subseteq \{f: \bbR \rightarrow \bbR\}$ are function classes.

\begin{remark}
    We assume that at least one of the treatments depends on $Z$ through a nonlinear function $g$ in the way that assumption \ref{ass:c4} holds true, then  without loss of generality we index this treatment as $j=1$ and absorb its coefficient $\beta_{k,1}$ into $f_k$.
\end{remark}

\subsubsection{Continuous instrument}
We need the following assumptions in the case of a continuous instrument.
\begin{enumerate}[label=(\textit{C\arabic*})]
	\item For all $g \in \tilde{\mathcal{G}}$, it holds for all $z \in \supp(Z)$ that $g(z, \cdot)$ is strictly monotone on $\supp(\eta_X)$. \label{ass:c1}
	\item For all $f \in \tilde{\mathcal{F}}$, $f$ is strictly monotone on $\supp(X^\top \beta + \eta_{Y})$ and differentiable almost everywhere. \label{ass:c2}
    \item For all $j \in \{1,\dots d\}, k \in \{1,\dots, p\}$, the noise terms $\eta_{X_j}$ and $\eta_{Y_k}$ are absolutely continuous with respect to the Lebesgue measure.  \label{ass:c3}
    \item For all $e_1 \in \supp(\eta_{X_1})$, there exists a subset $\cZ^{\diamond} \subseteq \supp(Z)$ with non-zero Lebesgue measure, such that for all $z_1, z_2 \in \cZ^{\diamond}$ we have $\pdv{g_1(z_1, e_1)}{e_1} \neq \pdv{g_1(z_2, e_1)}{e_1}$. \label{ass:c4}
    \item For $(z_1, \dots, z_q) \in \supp(Z)$ and $(e_2,\dots,e_d) \in \supp(\eta_{X_2},\dots, \eta_{X_d})$, we define the Jacobian matrix
    $\mathbf{J}_{g}(z, e) \coloneqq \begin{bmatrix}
        \pdv{g_2(z, e_2)}{z_1} & \dots & \pdv{g_d(z, e_d)}{z_1} \\
        \vdots & \ddots & \vdots \\
        \pdv{g_2(z, e_2)}{z_q} & \dots & \pdv{g_d(z, e_d)}{z_q}
    \end{bmatrix}$.
    There exists a subset $\cE^{\diamond} \subseteq \supp(\eta_{X_2},\dots, \eta_{X_d})$ with non-zero Lebesgue measure such that for all $e \in \cE^{\diamond}$, we have $\bigcap\limits_{z \in \supp(Z)} \ker(\mathbf{J}_{g}(z, e))$ = \{0\}. \label{ass:c5}
\end{enumerate}

We first present a theorem which shows the identifiability of the interventional distribution $P_Y^{\mathrm{do}(X:=x)}$ for the pre-additive model class and continuous instrument $Z$.

\begin{theorem} \label{thm:interv_ident_cont_preadd}
    Consider the model in \eqref{eq:pre-ANM} and suppose the assumptions \ref{ass:c1}-\ref{ass:c5} hold. For all $x \in \supp(X)$, the interventional distribution $P_Y^{\mathrm{do}(X:=x)}$ is then identifiable from the observed distribution $P_{(X,Y)|Z}$. 

    \begin{proof}
        See Appendix~\ref{app:ident_pre_ANM}.
    \end{proof}
\end{theorem}

\begin{proposition}\label{prop:ident_cont_preadd}
    Consider the model in \eqref{eq:pre-ANM}. Suppose the assumptions \ref{ass:c1}-\ref{ass:c5} hold. For any two models $(g_j,f_k,\beta_k,\eta_X,\eta_Y)$, $(\tilde{g}_j,\tilde{f}_k,\tilde{\beta}_k,\tilde{\eta}_X,\tilde{\eta}_Y) \in \mathcal{M}_{\rm DIV}^p$ that induce the same conditional distribution of $(X,Y)$ given $Z=z$, it then holds
    \begin{enumerate}[label=(\textit{\alph*})]
        \item for all $j \in \{1,\dots d\}$, $z \in \supp(Z)$, $e_{X} \in \supp(\eta_{X_j})$ we have $g_j(z,e_{X}) = \tilde{g}_j(z,e_{X})$,
        \item for all $k \in \{1,\dots p\}$, we have $\beta_k = \tilde{\beta}_k$, further for all $w \in \{x^\top \beta_k + e_Y \mid x \in \supp(X),  e_Y \in \supp(\eta_{Y_k})\}$ we have $f_k(w) = \tilde{f}_k(w)$,
        \item $(\eta_X, \eta_Y) \eqdist (\tilde{\eta}_X, \tilde{\eta}_Y)$. 
    \end{enumerate}  

\begin{proof}
    See Appendix~\ref{app:ident_pre_ANM}.
\end{proof} 
\end{proposition} 

\begin{remarks} 
We now discuss the technical assumptions we make to ensure the identifiability. 
    \begin{enumerate}
        \item Assumption \ref{ass:c4}. A necessary condition for this assumption to hold is that for all $e_1 \in \supp(\eta_{X_1})$, the function $g_1$ cannot be linear. Linearity would imply constant partial derivatives with respect to $e_1$ for all $z$, which violates the requirement that these derivatives vary across different values of $z$.
        
        \item Assumption \ref{ass:c5}. %
        If all treatment models $g_j$, $j \in \{2, ..., d\}$ are linear, for all $i \in \{1, \dots, q\}$ the partial derivatives $\pdv{g_j(z, e_j)}{z_i}$ are constant. In this case, $\mathbf{J}_{g}$ being of full column rank is equivalent to the full-rank condition in the classical 2SLS which also implies that we must have as many instruments as treatment variables (i.e.\ $q \geq d$). When at least one $g_j$ is nonlinear, the derivatives $\pdv{g_j(z, e_j)}{z_i}$ become functions of $z_i$. It is then no longer necessary to have $q \geq d$ (we can have less instruments than treatments) as long as $\pdv{g_j(z, e_j)}{z_i}$ vary sufficiently.
        \item Assumptions \ref{ass:c4} and \ref{ass:c5} refer to the relevance condition requiring the instrument $Z$ to be associated with the treatment variable $X$.
    \end{enumerate}
\end{remarks}

\subsubsection{Discrete instrument}

We present a new identification result for the case when the instrument $Z$ is discrete, showing that the traditional order condition, $\dim(Z) \geq \dim(X)$, is not necessary for identification under certain conditions. In Section \ref{sec:setting}, we provided an example where the conventional 2SLS method fails when the order condition is not satisfied. To the best of our knowledge, the most general sufficient conditions for point identification when $X$ is a vector and $Z$ is binary are given by \citet[Theorem S2 in Supplement]{torgovitsky2015}. For the general rectangular model class, he relies on the strong assumption that $(X, Z)$ has rectangular support. In contrast, by restricting the outcome model class to the pre-additive noise model, we are able to avoid this stringent assumption, achieving identification without needing the rectangular support condition.

Consider the pre-additive noise IV model class $\mathcal{M}_{\rm DIV}^{\pre}$ as defined in \eqref{eq:pre-ANM}, but assume $Z$ to be a discrete instrument, that is, $Z := (Z_1, \dots, Z_q)$ and for each $i \in \{1, \dots, q\}$, $\supp(Z_i) = \{z_{i1}, z_{i2}, \dots\}$, with $\supp(Z_i)$ being at most countable.

Let this model satisfy the following assumptions:

\begin{enumerate}[label=(\textit{D\arabic*})]
	\item For all $g \in \tilde{\cG}$, it holds for all $z \in \supp(Z)$ that $g(z, \cdot)$ is strictly monotone on $\supp(\eta_X)$ and differentiable almost everywhere. \label{ass:d1}
	\item For all $f \in \tilde{\cF}$, $f$ is strictly monotone on $\supp(X^\top \beta + \eta_{Y})$ and differentiable almost everywhere. \label{ass:d2}
    \item For all $j \in \{1,\dots d\}, k \in \{1,\dots, p\}$, the noise terms $\eta_{X_j}$ and $\eta_{Y_k}$ are absolutely continuous with respect to the Lebesgue measure.  \label{ass:d3}
    \item For all $e_1 \in \supp(\eta_{X_1})$, there exists a subset $\cZ^{\diamond} \subseteq \supp(Z)$ with $P_Z(\mathcal Z^\diamond)>0$ such that for all $z_1, z_2 \in \cZ^{\diamond}$ we have $\pdv{g_1(z_1, e_1)}{e_1} \neq \pdv{g_1(z_2, e_1)}{e_1}$. \label{ass:d4}
    \item For $z_1, z_2 \in \supp(Z)$ and $(e_2,\dots,e_d) \in \supp(\eta_{X_2},\dots, \eta_{X_d})$, we define the Jacobian matrix
    $\tilde{\mathbf{J}}_{g}(z_1, z_2, e_{\texttt{2:d}}) \coloneqq \begin{bmatrix}
        \pdv{(g_2(z_1, e_2) - g_2(z_2, e_2))}{e_2} & \dots & 0 \\
        \vdots & \ddots & \vdots \\
        0 & \dots & \pdv{(g_d(z_1, e_d) - g_d(z_2, e_d))}{e_d}
    \end{bmatrix}$.
    There exists a subset $\cE^{\diamond} \subseteq \supp(\eta_{X_2},\dots, \eta_{X_d})$ with non-zero Lebesgue measure and $z_1, z_2 \in \supp(Z)$ such that for all $e \in \cE^{\diamond}$, we have $\ker(\tilde{\mathbf{J}}_{g}(z_1, z_2, e))$ = \{0\}. \label{ass:d5}
\end{enumerate}

We now present a proposition of Theorem~\ref{thm:interv_ident_bin_preadd} which shows the identifiability of the interventional distribution $P_Y^{\mathrm{do}(X:=x)}$ for the pre-additive model class and discrete instrument~$Z$.

\begin{theorem} \label{thm:interv_ident_bin_preadd}
    Consider the model in \eqref{eq:pre-ANM} and suppose the assumptions \ref{ass:d1}-\ref{ass:d4} hold. For all $x \in \supp(X)$, the interventional distribution $P_Y^{\mathrm{do}(X:=x)}$ is then identifiable from the observed distribution $P_{(X,Y)|Z}$. 

    \begin{proof}
        See Appendix~\ref{app:ident_pre_ANM}.
    \end{proof}
\end{theorem}

\begin{proposition} \label{prop:ident_bin_preadd}
    Consider the model in \eqref{eq:pre-ANM}. Suppose the assumptions \ref{ass:d1}-\ref{ass:d4} hold. For any two models $(g_j,f_k,\beta_k,\eta_X,\eta_Y)$, $(\tilde{g}_j,\tilde{f}_k,\tilde{\beta}_k,\tilde{\eta}_X,\tilde{\eta}_Y) \in \mathcal{M}_{\rm DIV}^p$ that induce the same conditional distribution of $(X,Y)$ given $Z=z$ it then holds
    \begin{enumerate}[label=(\textit{\alph*})]
        \item for all $j \in \{1,\dots d\}$, $z \in \supp(Z)$, $e_{X} \in \supp(\eta_{X_j})$ we have $g_j(z,e_{X}) = \tilde{g}_j(z,e_{X})$,
        \item for all $k \in \{1,\dots p\}$, it holds $\beta_k = \tilde{\beta}_k$, further for all $w \in \{x^\top \beta_k + e_Y \mid x \in \supp(X),  e_Y \in \supp(\eta_{Y_k})\}$ we have $f_k(w) = \tilde{f}_k(w)$,
        \item $(\eta_X, \eta_Y) \eqdist (\tilde{\eta}_X, \tilde{\eta}_Y)$. 
    \end{enumerate}  
\end{proposition}

\begin{remark}
Note that assumption \ref{ass:d5} stipulates that $g_j$ cannot be linear for all $j \in \{1, \dots, d\}$, posing a crucial difference to the assumption \ref{ass:c5}, applicable when the instrument $Z$ is continuous.
\end{remark}

\subsection{Consistency}\label{subsec:consistency}

The identifiability results above are population statements. To connect them to estimation, we establish consistency of the DIV estimator in a simple parametric subclass.

Let $Z,X,Y$ be scalar and suppose the training law is induced by
\begin{equation}\label{eq:lin_iv_consistency}
X = a Z + g H + \varepsilon_X,
\qquad
Y = b X + d H + \varepsilon_Y,
\qquad
Z \ind (H,\varepsilon_X,\varepsilon_Y),
\end{equation}
with parameter $\theta \coloneqq (a,b,g,d)\in\Theta\subset\mathbb R^4$. For each $z$, write
$P^{\theta}_{(X,Y)\mid Z=z}$ for the model-implied conditional distribution of $(X,Y)$ given $Z=z$.

Let $\mathrm{ES}(P,v)$ denote the (bivariate) energy score of a distribution $P$ evaluated at
$v\in\mathbb R^2$. For $v=(x,y)^\top$, define the pointwise conditional score
\[
\ell_\theta(x,y,z) := \mathrm{ES}\!\bigl(P^{\theta}_{(X,Y)\mid Z=z},\, (x,y)^\top\bigr),
\]
and the population and empirical objectives
\[
L(\theta) := \mathbb E\!\left[\ell_\theta(X,Y,Z)\right],
\qquad
\hat L_n(\theta) := \frac1n\sum_{i=1}^n \ell_\theta(X_i,Y_i,Z_i),
\]
based on i.i.d.\ observations $\{(Z_i,X_i,Y_i)\}_{i=1}^n$. We define the DIV estimator in
this linear class by
\[
\hat\theta_n \in \arg\min_{\theta\in\Theta}\hat L_n(\theta).
\]

\begin{theorem}[Consistency in the linear model class]\label{thm:consistency_linear}
Assume \emph{well-specification}, i.e., there exists $\theta^*\in\Theta$ such that
$
P^{\theta^*}_{(X,Y)\mid Z=z}=P_{(X,Y)\mid Z=z} \ \text{for }P_Z\text{-a.e.\ }z$.
Further suppose that the parameter set $\Theta$ is compact and the exogenous variables
$(Z,H,\varepsilon_X,\varepsilon_Y)$ have bounded support.
Then $\hat\theta_n \xrightarrow{\mathbb P} \theta^*$.
\end{theorem}

\begin{proof}
    See Appendix~\ref{app:consistency}.
\end{proof}

\noindent We focus on the univariate linear subclass as a simple setting to illustrate how strict propriety of the energy score together with uniform convergence of the empirical objective over compact parameter set yields a standard consistency guarantee.

\section{Simulated experiments} \label{sec:experiments}
In this section, we empirically compare the performance of DIV with several benchmark methods on simulated data. We consider scenarios with binary and continuous instruments and treatments, an under-identified case where $\dim(Z) < \dim(X)$, and a setting where the treatment depends only weakly or not at all on the instrument through its conditional mean.

We benchmark the performance of DIV against several existing methods, including control functions (CF) \citep{heckman1976cf, newey1999control}, HSIC-X \citep{saengkyongam2022exploiting}, DeepIV \citep{hartford2017deepIV}, DeepGMM \citep{bennett2020deepGMM}, DIVE \citep{kook2024}, IVQR \citep{chernozhukov2005}, and engression \citep{shen2024engression}. HSIC-X, DeepIV, and DeepGMM are IV regression methods for estimating the interventional mean function. DIVE and IVQR estimate the interventional quantile function, while engression is a distributional regression approach aimed at estimating the conditional distribution rather than the interventional one. For a detailed description of these methods, and implementation details, see Appendix~\ref{app:benchmark}. Details on the software implementation of the DIV method, including model architecture, training configuration, and example code usage, can be found in Appendix~\ref{sec:software}.

In terms of computational efficiency, training a DIV model with $5000$ epochs on a data set of size $n \! = \! 10000$ takes approximately 8 minutes on a GPU (MPS) and 12 minutes on a CPU (MacBook Pro, M1 chip). We compare runtimes for methods that produce distributional estimates (the interventional distribution or a range of interventional quantiles). On a simulated data set with $n \! = \! 10000$, fitting $100$ quantile levels took about $62$ minutes for DIVE (with runtime increasing sharply in $n$; e.g., 4.4 minutes for $n \! = \! 5000$) and $2$ minutes for linear IVQR.

For all simulated experiments, we report average results over $10$ simulation runs. We additionally include experiments with $100$ replications and varying sample sizes to study finite-sample behaviour. The code for reproducing all experiments, including those with real-world data, is available at: \url{https://github.com/aholovchak/DIV}.

\subsection{Continuous treatment} \label{subsec:sim_cont_X}

We begin by evaluating the empirical performance of DIV in estimating interventional mean functions, comparing it against benchmark methods including CF, HSIC-X, DeepIV and DeepGMM. Our experiments encompass both linear and nonlinear settings to assess their robustness.

We use a training data set of size $10000$ across all experiments. The DIV models employ a 4-layer model architecture; the Adam optimiser is used with a learning rate of $10^{-3}$; we run $10000$ epochs for both models. We compare the performance both visually, by assessing the estimated interventional mean functions by different methods against the true interventional mean function, and in terms of the integrated mean squared error (MSE) $\mathbb{E}[(\hat{\mu}(x) - f^0(x))^2]$, where $\hat{\mu}(x)$ is the estimated interventional mean function, and $f^0(x)$ the true causal function, approximated using a test sample size of $5000$. 

In the first group of scenarios, we focus on the univariate case, meaning that $Z, X, Y \in \mathbb{R}$. 
We consider 6 data-generating processes, three with the instrument $Z$ following a continuous uniform distribution, $Z \sim \unif(0,3)$ (Settings 2, 4 and 6 below), and three with $Z$ following a Bernoulli distribution, $Z \sim \Bernoulli(0.5)$ (Settings 1, 3, and 5). In each setting, we assume  mutually independent $H,\varepsilon_X, \varepsilon_Y \sim N(0,1)$. For all settings, the treatment model is linear and defined as $g(Z,H,\varepsilon_X) \coloneqq Z + H + \varepsilon_X$. The outcome models are defined as follows: \\
\textbf{Scenario 1-2} (Linear function): $f(X, H, \varepsilon_Y) \coloneqq X - 3H + \varepsilon_Y$. \\
\textbf{Scenario 3-4} (Case distinction, linear \& log functions): $f(X, H, \varepsilon_Y) \coloneqq \mathbbm{1}_{\{X \leq 1\}}(\frac{1}{5}(5.5 + 2X + 3H + \varepsilon_Y)) + \mathbbm{1}_{\{X > 1\}}(\log((2X + H)^2 + \varepsilon_{Y}^2)).$ \\
\textbf{Scenario 5-6} (Nonlinear function): $f(X, H, \varepsilon_Y) \coloneqq 3\sin(2X) + 2X - 3H + \varepsilon_Y$. \\

\begin{figure}[b!]
        \centering        
        \includegraphics[width=1\textwidth,height=1\textheight,keepaspectratio]{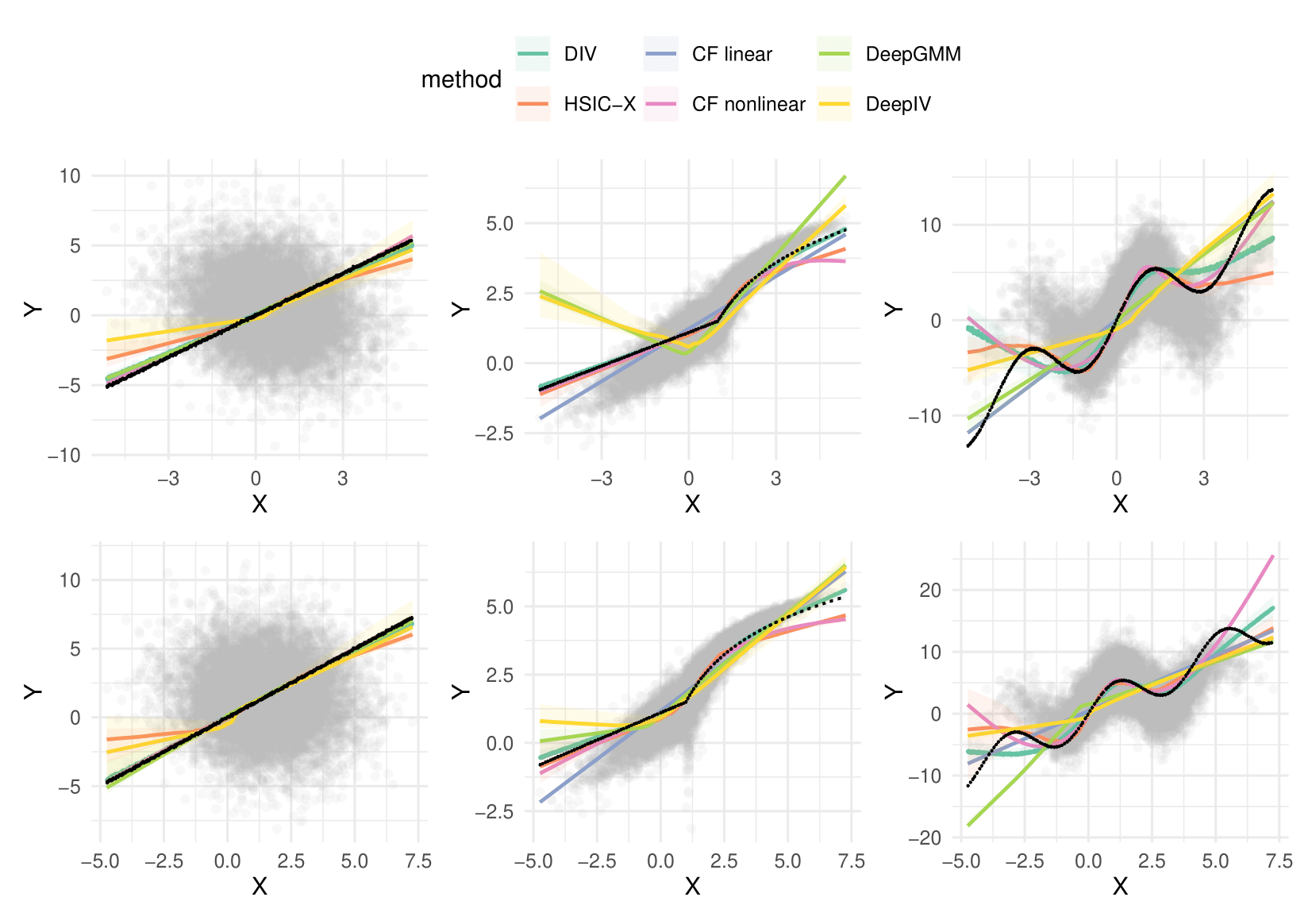}
        \caption{Estimated interventional mean functions $\hat{\mu}^*(x)$. True $\mu^*(x)$ represented as black dashed line (\dashed). Top row: binary instrument Z; bottom row: continuous instrument Z. Columns correspond to Scenarios 1–2, 3–4, and 5–6, respectively.}
        \label{fig:mean_fct}
\end{figure}

\begin{figure}[t!]
        \centering
\includegraphics[width=1\textwidth,height=1\textheight,keepaspectratio]{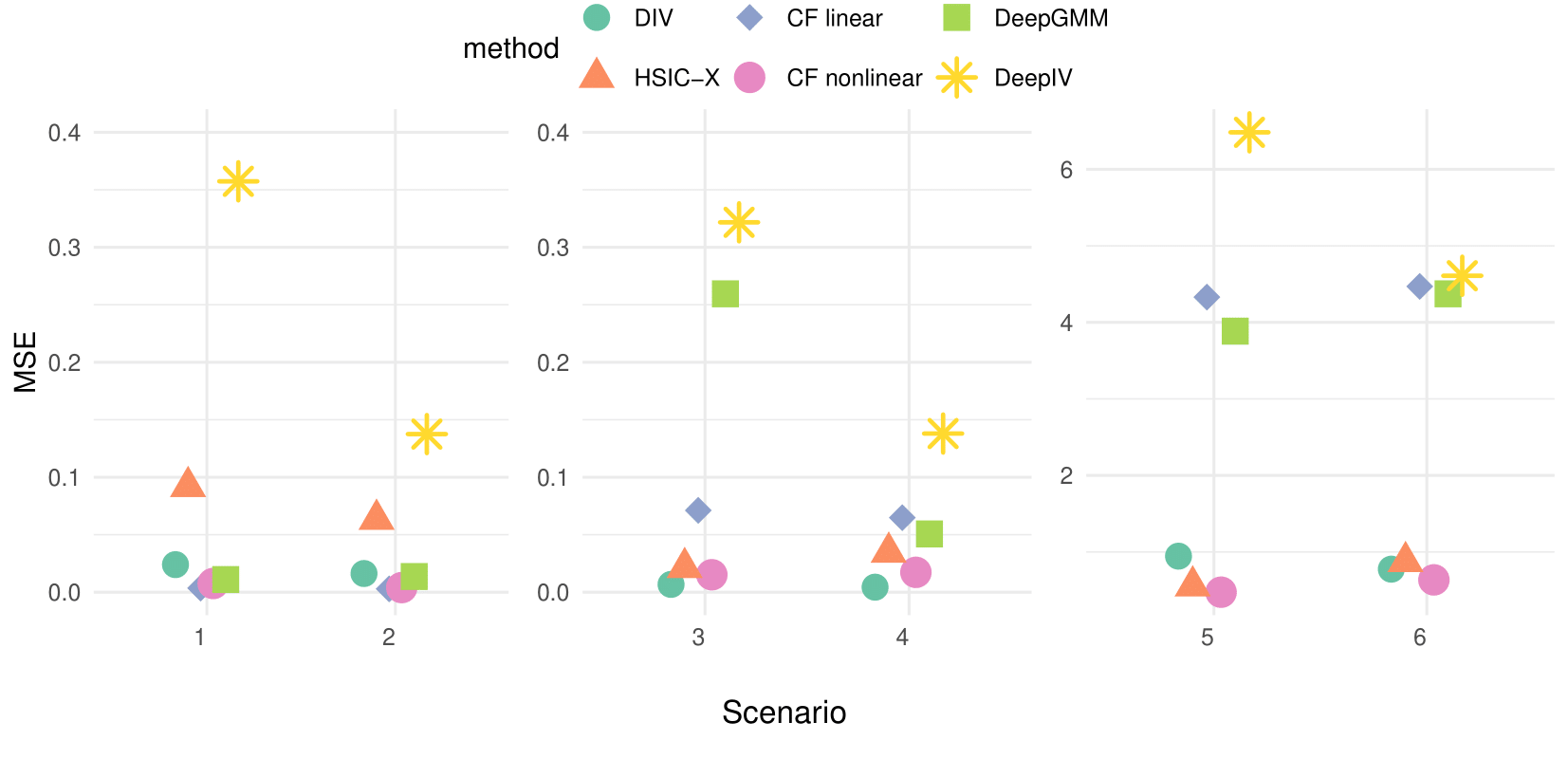}
        \caption{Mean MSE of the estimated interventional mean across 10 runs (training sample size $n = 10000$).}
        \label{fig:mean_fct_mse}
\end{figure}

Figure~\ref{fig:mean_fct} presents the estimation results for the interventional mean function across six different scenarios. Across all considered scenarios, the DIV method performs competitively, providing a close estimation of the interventional mean function. The DIV method performs well in the linear setting (scenarios 1-2) and outperforms benchmark methods in the presence of a pre-additive outcome noise model (scenarios 3-4). Furthermore, for the highly nonlinear post-additive noise outcome model (scenarios 5-6), DIV closely follows the sine curve form of the interventional mean, achieving performance comparable to HSIC-X. Figure~\ref{fig:mean_fct_mse} reinforces the conclusions drawn from the visual comparison of the estimated interventional mean functions, summarizing the performance of the evaluated IV methods in terms of the average MSE over $10$ simulation runs. It has to be stressed, though, that DIV goes beyond merely estimating the mean effect, as it estimates the full interventional distribution.

Finite-sample behaviour across increasing training sample sizes is reported in Appendix~\ref{app:finite-sample}, showing that the MSE of the estimated interventional mean decreases overall with $n$ across Scenarios~1--6 (and the energy distance improves correspondingly).

\subsection{Binary treatment} \label{subsec:sim_binary_X}

We consider two nonlinear scenarios with treatment $X$ being a binary variable. The scenarios we use are inspired by those described in \citet[Section~5.1]{kook2024}. The models are defined as follows: \\
\textbf{Scenario 1}: $Z, H, \varepsilon_X \sim \Logistic(0,1)$ mutually independent. $g(Z,H,\varepsilon_X) \coloneqq \mathbbm{1}(4Z + 4H > \varepsilon_X)$; $f(X, H, \varepsilon_Y) \coloneqq \log(1 + \exp(18 + 8X + 6H))$. \\
\textbf{Scenario 2}: $Z, H, \varepsilon_X, \varepsilon_Y \sim \Logistic(0,1)$ mutually independent. $g(Z,H,\varepsilon_X) \coloneqq \mathbbm{1}(4Z + 4H > \varepsilon_X)$; $f(X, H, \varepsilon_Y) \coloneqq 2 + (X + 1)^2 + 3(X + 1) + 2H + \varepsilon_Y$. 

The target quantity we are aiming for is the quantile treatment effect (QTE) \citep{doksum1974qte}, defined as $$\qte(\alpha) \coloneqq q^*_\alpha(1) - q^*_\alpha(0).$$ The QTE captures the difference between the quantiles of the interventional outcome distribution under treatment and control, providing insight into the heterogeneous impact of the treatment across different points of the outcome distribution. We consider a non-equidistant sequence of quantiles between $0.01$ and $0.99$. To evaluate the accuracy of the estimated QTE, we further compute the root mean squared error (RMSE) between the estimated and true QTE at each quantile; we perform 10 simulations runs. We use a training data set of size $10000$ for both experiments. The DIV models adopt a 4-layer architecture and are trained with the Adam optimiser at a learning rate of $10^{-4}$, running for $20000$ epochs. We compare the performance of DIV with linear IVQR and DIVE, using their respective implementations as described above. 

\begin{figure}[t!]
        \centering
\includegraphics[width=1\textwidth,height=1\textheight,keepaspectratio]{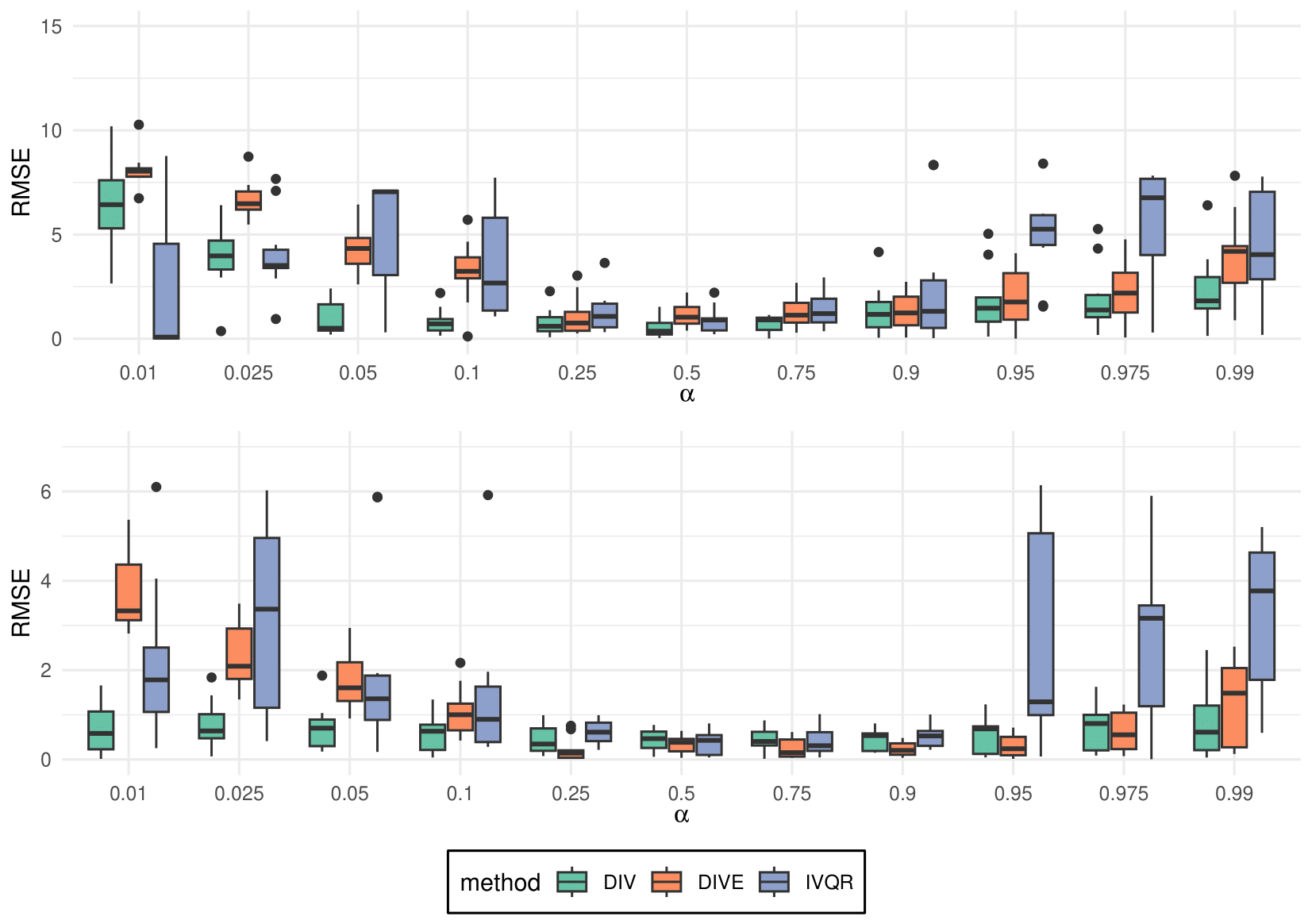}
        \caption{RMSE of the estimated QTEs. Top: scenario 1, bottom: scenario 2.}
        \label{fig:QTE_plot}
\end{figure}

The treatment function $g(z, \cdot)$ is not strictly monotone in either scenario considered. While the DIV method theoretically requires monotonicity of $g(z, \cdot)$, it demonstrates empirical robustness to violations of the monotonicity assumption and achieves better or at least comparable performance in terms of RMSE across the quantiles, as shown by the boxplots in Figure~\ref{fig:QTE_plot}. Additionally, in Figure~\ref{fig:binX-qte-true-estim} we examine the average estimated QTE curves together with their pointwise variability. As the sample size increases, the estimated QTE curves approach the true QTE and exhibit reduced variability, with the shaded bands shrinking as more data become available.

Appendix~\ref{app:finite-sample} reports finite-sample behaviour across increasing $n$, showing that the RMSE of the estimated QTE curves decreases overall with sample size (with improving energy-distance accuracy for the interventional outcome distributions).

\begin{figure}[t!]
        \centering
\includegraphics[width=1\textwidth,height=1\textheight,keepaspectratio]{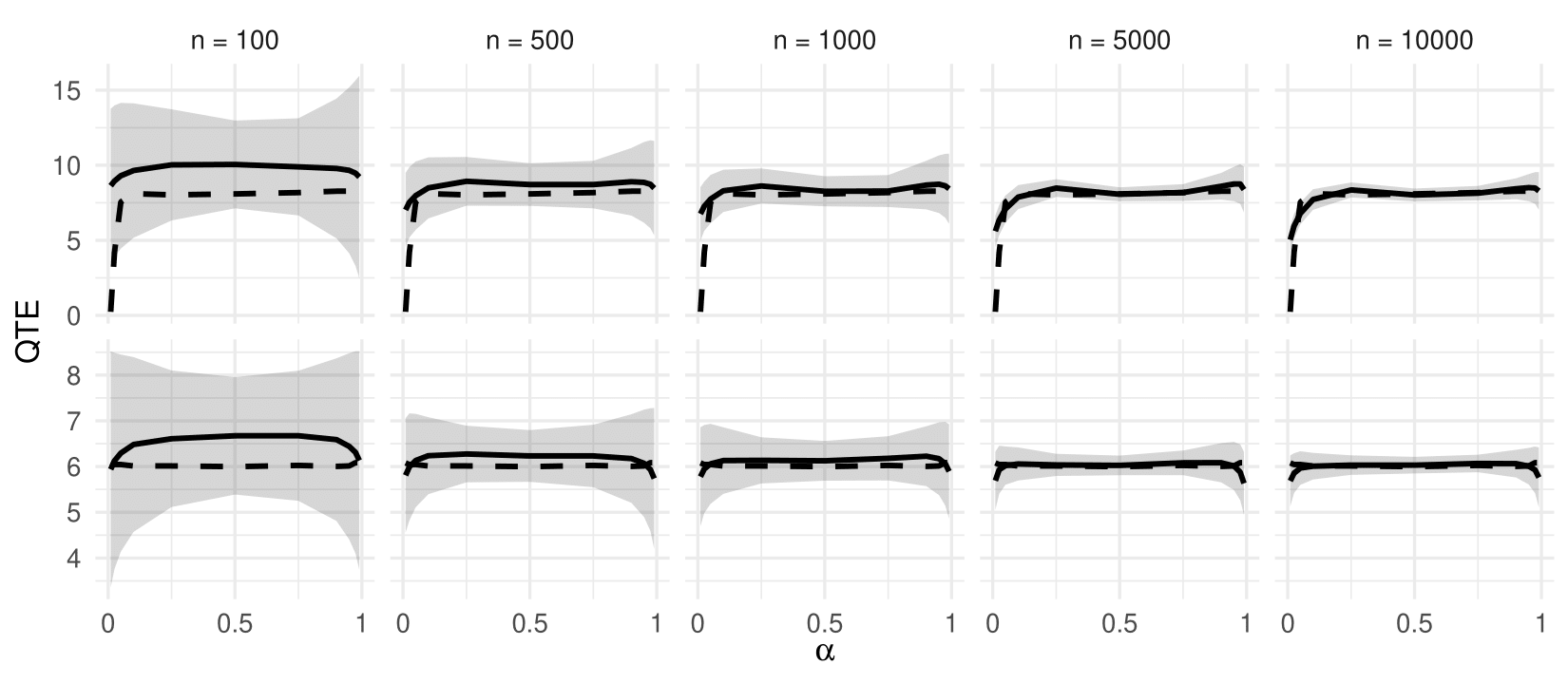}
        \caption{True (dashed) and DIV-based QTE curves (solid). The solid curves correspond to the mean across 100 simulation runs, with shaded bands indicating $\pm$ SD. Top: Scenario 1; bottom: Scenario 2.}
        \label{fig:binX-qte-true-estim}
\end{figure}

\subsection{`Under-identified' case} \label{subsec:under_ident}

Here, we focus on the `under-identified' case where $\dim(Z)<\dim(X)$ with a single binary instrument $Z$. We set $Z \sim \Bernoulli(0.5)$, $X \in \mathbb{R}^2$, and $Z, Y \in \mathbb{R}$. The noise terms and hidden confounder are assumed to be mutually independent and follow a standard normal distribution, $\varepsilon_{X_1}, \varepsilon_{X_2}, \varepsilon_Y, H \sim N(0,1)$. The treatment model functions are defined as $g_1(Z,H,\varepsilon_{X_1}) \coloneqq Z(2H - 0.5\varepsilon_{X_1})$ and $g_2(Z,H,\varepsilon_{X_2}) \coloneqq \log(7 + Z + H + \varepsilon_{X_2})$. The outcome model is linear and defined as $f(X_1, X_2, H, \varepsilon_Y) \coloneqq X_1 + 2X_2 + 2H + \varepsilon_Y$.
For the DIV method, we define the outcome model as a single linear layer without bias. The remaining model parameters follow those of the continuous treatment scenarios. The publicly available implementation of HSIC-X allows for selecting a linear outcome model, which we use as a  benchmark. We evaluate the $L_2$-norm of the estimation error for the vector of linear coefficients $\beta \coloneqq (1,2)^\top$. We consider sample sizes $n \in \{10^2, 10^3, 10^4\}$ and conduct $10$ simulation runs for each method.

\begin{table}[t]
\centering
\begin{tabular}{lccc}
\toprule
 & \multicolumn{3}{c}{$\|\hat{\beta} - \beta\|_2$} \\
\cmidrule(lr){2-4}
                & $n = 10^2$ & $n = 10^3$ & $n = 10^4$  \\
\midrule
DIV    & 0.311 & 0.188 & 0.037  \\
HSIC-X & 1.243 & 0.287 & 0.116 \\
\bottomrule
\end{tabular}
\caption{$L_2$-norm of the estimation error for the linear coefficient vector $\beta$.}
\label{tab:beta_norm}
\end{table}

The results in Table~\ref{tab:beta_norm} indicate that the DIV method estimates the linear coefficients with reasonable precision, achieving higher accuracy as the sample size increases. Furthermore, the DIV method outperforms HSIC-X, yielding lower $L_2$-norm across all considered sample sizes.

\subsection{Weak instrument relevance} \label{subsubsec:XZ_alpha}

Instrumental variable designs often face the challenge that the instrument has only a weak influence on the treatment. Weak relevance may arise for structurally different reasons, and it is therefore important that an estimation method remains robust across these settings. We consider two complementary designs and demonstrate that DIV performs well in both because it exploits the full conditional distribution of $(X,Y)\mid Z=z$.

\paragraph{Weak mean relevance.}
Let $Z \sim \unif(-3,3)$, \ $H,\varepsilon_X,\varepsilon_Y \sim \unif(-1,1)$ independent, and define $g(Z,H,\varepsilon_X)=0.5\,Z\,(\alpha + 3H + \varepsilon_X)$, and $f(X,H,\varepsilon_Y)=0.5\log\!\bigl(1+\exp(2X - 3H + \varepsilon_Y)\bigr)$, with tuning parameter $\alpha \in \mathbb{R}$. It holds $\mathbb{E}(X|Z)= 0.5 \alpha Z$, $\mathrm{Var}(X|Z)=\tfrac{5}{6}Z^2$, where $\alpha$ controls the strength of the conditional mean relationship between $X$ and $Z$. When $\alpha=0$, the usual moment identifiability condition (compare, e.g.\ \citet[Section~2]{saengkyongam2022exploiting}) breaks down.  
Since the outcome model is not of post-additive noise form, the independence restriction in \citet{saengkyongam2022exploiting} is also violated, reflected by the HSIC test being rejected in all epochs for all~$\alpha$. Table~\ref{tab:alpha_XZ}(a) shows that DIV continues to perform well even when mean relevance disappears, confirming that distributional information remains informative.

\paragraph{Weak mean-variance relevance.}
Next we consider a design in which $\alpha$ weakens both the conditional mean and variance: $g(Z,H,\varepsilon_X)=\alpha(Z+H) + 0.5\,\varepsilon_X$, $f(X,H,\varepsilon_Y)=\tanh(1.5X) - 2.5H + \varepsilon_Y$. Then $\mathbb{E}(X\mid Z)=\alpha Z$, $\mathrm{Var}(X\mid Z)=\alpha^2\bigl(\tfrac{1}{3}\bigr)+\tfrac{1}{12}$.
However, higher-order moments remain $Z$-dependent; for example, $\mathbb{E}[(Z+H)^3\mid Z] = Z(Z^2+1)$.
Thus, even when both first two moments vanish as identifying sources, the distribution of $X\! \mid \! Z$ continues to carry useful structural information. Table~\ref{tab:alpha_XZ}(b) displays the empirical MSE. The qualitative pattern parallels the first setting: most benchmark estimators deteriorate as $\alpha$ decreases, whereas DIV maintains consistently strong performance. In contrast to the first design, the outcome model here is of post-additive noise form, so the required independence restriction holds; consequently, HSIC-X remains stable as well.
\begin{table}[t!]
\centering

\begin{minipage}{0.47\textwidth}
\centering
(a) Weak mean relevance\\[2mm]
\begin{tabular}{lrrr}
\hline
& $\alpha=0$ & $\alpha=1$ & $\alpha=5$ \\
\hline
DIV & 0.002 & 0.002 & 0.002 \\ 
HSIC-X & 2.693 & 0.333 & 0.344 \\
CF linear & 141.941 & 0.476 & 1.625 \\ 
CF nonlinear & 2.762 & 0.243 & 0.057 \\ 
DeepGMM & 1.158 & 0.274 & 0.005 \\ 
DeepIV & 0.675 & 0.305 & 0.102 \\
\hline
\end{tabular}
\end{minipage}
\hfill
\begin{minipage}{0.47\textwidth}
\centering
(b) Weak mean-variance relevance\\[2mm]
\begin{tabular}{lrrr}
\hline
& $\alpha=0$ & $\alpha=1$ & $\alpha=5$ \\
\hline
DIV & 0.031 & 0.014 & 0.028 \\ 
HSIC-X & 0.029 & 0.011 & 0.018 \\
CF linear & 1.450 & 0.115 & 0.237 \\ 
CF nonlinear & 0.395 & 0.003 & 0.038 \\ 
DeepGMM & 0.207 & 0.033 & 0.015 \\ 
DeepIV & 0.236 & 0.076 & 0.179 \\
\hline
\end{tabular}
\end{minipage}

\caption{MSE of the interventional mean functions for weak instrument scenarios.}
\label{tab:alpha_XZ}
\end{table}

\section{Real-world applications}\label{sec:realdata}

\subsection{Colonial origins of comparative development data}

We investigate the performance of the DIV method in an application based on the real-world economic data set\footnote{The data are publicly available; see \url{https://ivmodels.readthedocs.io/en/latest/examples/acemoglu2001colonial.html} from \citet{acemoglu2001} for data description and access.}. The study examines the causal relationship between institutional quality and economic development, hypothesizing that historical institutions, shaped by European colonization strategies, have long-term effects on prosperity. Specifically, European settlers established different types of institutions depending on local conditions: in regions with high settler mortality rates, extractive institutions were set up to exploit resources, whereas in regions with low settler mortality, settlers built inclusive institutions that supported property rights and economic growth.

The validity of the settler mortality instrument used by \citet{acemoglu2001} has been debated in subsequent literature. \citet{albouy2012comment} argued that much of the original mortality data were based on conjectures or military campaign records rather than direct settler deaths, raising concerns about measurement error and potential weak instrument problems. In response, \citet{acemoglu2012reply} defended the historical comparability of their data and demonstrated that their main findings are robust to alternative codings and caps on extreme mortality values. In this paper, we follow the empirical design of \citet{acemoglu2001} to allow comparison with previous work, while acknowledging that the debate has highlighted the sensitivity of the results to data quality and model specification.

To measure institutional quality $X$, \citet{acemoglu2001} use the average protection against expropriation risk (1985-1995), an indicator of how secure property rights are in each country. The outcome variable $Y$ is defined as log GDP per capita in 1995, reflecting economic performance. The authors apply a linear two-stage least squares (2SLS) approach, using historical settler mortality $Z$ as an instrumental variable for institutional quality. The data set consists of $n = 64$ observations, leading to a particularly challenging scenario for DIV model estimation due to the small sample size.

In our analysis, we apply the DIV method to the same data set and compare its estimates of the interventional mean function to those obtained via 2SLS, which was used in the original analysis. Figure~\ref{fig:colonial_devel} displays the estimated interventional mean functions from DIV and 2SLS. For DIV, we additionally report the 5th and 95th interventional quantile functions, shown as a shaded region bounded by their lower and upper curves, correspondingly. The estimated effect of institutions on GDP per capita remains nearly linear, with a slope closely matching the 2SLS estimate. Although DIV is flexible enough to capture nonlinearities, it recovers an approximately linear relationship in this empirical setting, consistent with the key finding of the original study: institutions exert a positive and roughly linear causal effect on long-term economic prosperity. For $X>8$, the shaded region narrows; results in this range should be interpreted with caution, as only seven observations are available.

\begin{figure}[!b]
        \centering
        \includegraphics[width=1\textwidth,height=1\textheight,keepaspectratio]{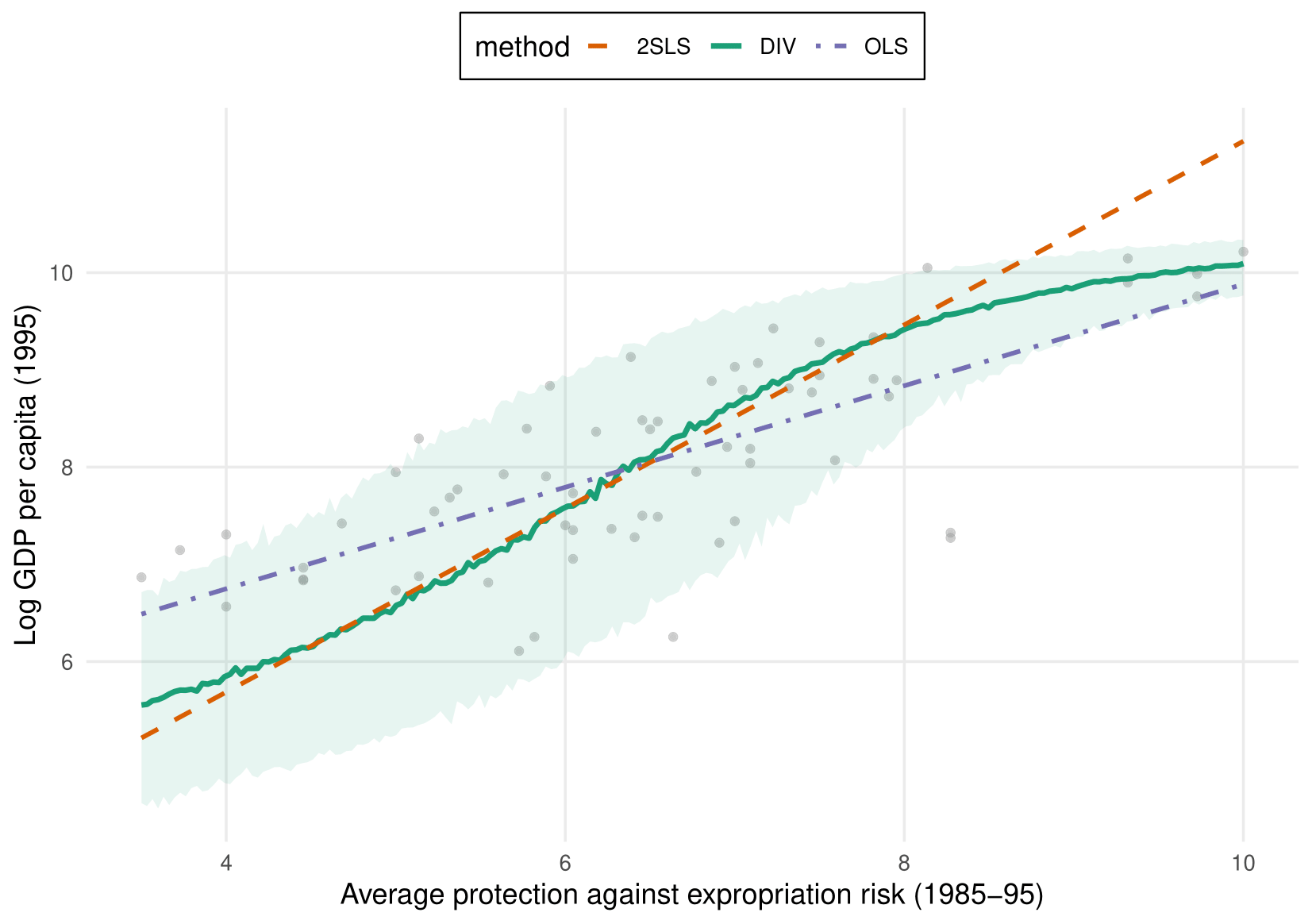}
        \caption{Estimated interventional mean functions for the effect of institutional quality on log GDP per capita. Shaded region represents the range between the 5th and 95th interventional quantile functions estimated by DIV.}
        \label{fig:colonial_devel}
\end{figure}

\subsection{Single-cell data} \label{subsec:singlecell}

So far, we have demonstrated that the DIV method performs well under controlled experimental settings, where assumptions hold, and in economic data, where background knowledge suggests a linear relationship. However, in complex biological settings, such assumptions are less clear to satisfy. We investigate the performance of DIV on a single-cell data set\footnote{The data are available in the accompanying GitHub repository.}, similar to \citet[Section~5]{shen2023causalityoriented}. The data consist of expression levels for ten genes, where one gene is designated as the response variable Y and the remaining nine serve as covariates X. Following \citet{shen2023causalityoriented}, the response gene was selected empirically as the one whose variance increases under interventions on other genes but has little effect on their variances, indicating that it behaves approximately as a leaf node in the underlying causal graph. The remaining nine genes are treated as potentially endogenous regressors, since their expression levels may be jointly influenced by unobserved cellular factors, thereby motivating the use of an instrumental variable approach.

The training data comprise ten distinct environments: one observational environment with $11{,}485$ samples, and nine interventional environments, each corresponding to a CRISPR perturbation targeting one of the nine covariate genes. The categorical variable representing the environment serves as the instrument $Z$, indicating whether an observation originates from the observational setting or from one of the single-gene interventions. Hence, the instrument takes ten categorical levels (nine intervention categories and one observational category) with highly unbalanced sample sizes: the observational environment dominates (around $94\%$ of all cells), while the nine interventional environments contain between $100$ and $550$ cells each.

The interventions are externally imposed by CRISPR perturbations, which are randomised with respect to cellular noise and independent of unobserved confounders affecting the response. Since each cell receives its perturbation at random through pooled viral transduction, the environment indicator satisfies the exogeneity condition required for a valid instrument. \citep{single_cell2022}.

Each CRISPR perturbation directly alters the expression distribution of its target gene and, consequently, the joint distribution of the covariates. We quantify this effect using pairwise energy distances \citep{rizzo2016energy_dist} between each interventional and the observational environment. All interventions yield clearly non-zero distances (ranging from approximately 50 to 550), indicating substantial distributional shifts and confirming that the instrument provides sufficient variation to satisfy the relevance assumption.

There are also data available from more than 400 environments involving interventions on hidden genes (meaning that we do not have access to the information about the instrument), which can be used for prediction evaluation. Among these environments, we select $50$ in which the distribution of the $10$ observed genes has the largest energy distance \citep{rizzo2016energy_dist} from their distribution in the pooled training data, indicating larger distribution shifts compared to the training data.

In the following two sections, we evaluate the models in terms of their generalizability to unseen interventions and their stability across different training subsets. DIV and engression are trained for $20000$ epochs each, while all other benchmark methods use their default tuning parameters. DeepIV results are excluded due to frequent occurrences of \texttt{NA} values, making the estimates unreliable.

Because the true causal effects are not directly observable, model performance is assessed indirectly through the connection between causality, generalizability, and stability. If the underlying causal mechanism linking $X$ and $Y$ remains unchanged across environments, then a correctly specified causal model should both generalise well to unseen interventions and yield consistent predictions across training subsets. In the single-cell application, this reasoning is plausible because CRISPR perturbations modify the distribution of upstream genes while leaving the intrinsic response process unchanged.

\subsubsection{Generalizability}

Generalizability refers to a model's ability to maintain predictive accuracy in unseen environments, especially when environments induce distribution shifts. Theoretical results on minimax solutions \citep[Proposition 3.1-3.3]{christiansen2022} suggest that the causal function is minimax optimal when environments induce distribution shifts that are sufficiently strong, ensuring better generalization across distribution shifts.

\begin{table}[b!]
\centering
\begin{tabular}{l|rrrrrrr}
  \hline
 & Q00 & Q05 & Q25 & Q50 & Q75 & Q95 & Q100 \\ 
  \hline
DIV & \textbf{0.1102} & \textbf{0.1184} & \textbf{0.1528} & \textbf{0.2447} & \textbf{0.3848} & \textbf{0.6827} & \textbf{0.6971} \\ 
  HSIC-X & \textbf{0.1168} & \textbf{0.1234} & \textbf{0.1574} & 0.2701 & 0.4204 & 0.7358 & 0.7544 \\ 
  DeepGMM & 0.1386 & 0.1625 & 0.1930 & \textbf{0.2469} & \textbf{0.3447} & \textbf{0.4959} & \textbf{0.5296} \\ 
  CF linear & 5.4931 & 5.7726 & 6.4438 & 7.4118 & 8.5316 & 9.9242 & 12.8477 \\ 
  CF nonlinear & 4.8406 & 5.4910 & 6.7146 & 7.7707 & 8.7012 & 9.8118 & 11.2271 \\ 
  Engression & 0.4811 & 0.4899 & 0.5251 & 0.6303 & 0.7710 & 1.0802 & 1.1011 \\ 
   \hline
\end{tabular}
\caption{MSE quantiles across 50 test environments with the strongest distribution shifts, averaged over 10 runs. Two best-performing methods per quantile are highlighted in bold.}
\label{tab:singlecell_generalizability}
\end{table}

We assess the model's generalizability by testing it on the 50 environments with the strongest distribution shifts. For each environment, we compute the MSE, summarise the results using quantiles, and report the average values over 10 runs. The results are presented in Table~\ref{tab:singlecell_generalizability}. Across all quantiles, DIV consistently ranks among the top two methods, demonstrating strong generalization to unseen environments. It performs particularly well in the lower and mid quantiles and remains competitive at higher quantiles. This suggests that DIV might effectively capture the causal function and can adapt well to varying intervention strengths.

\subsubsection{Stability}

One possible way to validate an estimator of a causal function is to compare its estimates on different environments. If the environments correspond to interventions on predictors, but not on the outcome itself, a reliable estimator should show minimal variability across all of the environments \citep{peters2016causal, meinshausen2018causality, rothenhausler2021anchor}. Conversely, if a method exhibits high instability across environments, it is likely not capturing the causal function but rather responding to spurious correlations.

To quantify the stability of an estimated model, we conduct a leave-one-environment-out analysis: we fit the model nine times, each time excluding one of the nine interventional environments from the training data. Predictions are then computed on the test data for each trained model, and stability is assessed using the following measure
$$
\mathcal{E}(\widehat{\mu}) = \widehat{\mathbb{E}}_X \left( \sum_{e, e^\prime} \left[ \widehat{\mu}_e(x) - \widehat{\mu}_{e^\prime}(x) \right]^2  \right),
$$
where $\widehat{\mu}_e(x)$ denotes the estimated interventional mean function, with the subscript $e$ indicating that data from environment $e$ was excluded during training. This metric quantifies the variance in predicted means across leave-one-out models, with lower values indicating greater stability.

\begin{table}[t!]
\centering
\begin{tabular}{l|r}
  \hline
& $\mathcal{E}(\widehat{\mu})$ \\ 
  \hline
  DIV & \textbf{0.496} \\ 
  HSIC-X & \textbf{0.985} \\ 
  DeepGMM & 5.309 \\ 
  DeepIV & 0.990 \\ 
  CF linear & 80.323 \\ 
  CF nonlinear & 154.992 \\ 
  Engression & 20.573 \\ 
   \hline
\end{tabular}
\caption{Stability measure $\mathcal{E}(\widehat{\mu})$ for different methods. Lower values indicate greater invariance across training subsets. Two best-performing methods are highlighted in bold.}
\label{tab:singlecell_stability}
\end{table}

Table \ref{tab:singlecell_stability} shows that DIV achieves the lowest stability error $\mathcal{E}(\widehat{\mu})$, exhibiting strong invariance across different training subsets. It is followed by HSIC-X and DeepIV, which are moderately stable. In contrast, methods such as DeepGMM and CF exhibit considerably higher stability errors, indicating that they are likely not capturing the causal effect reliably. For CF, this instability is likely due to the full-rank condition not being satisfied. Engression also shows highly unstable results as expected, since it aims to fit the conditional distribution of $Y|X = x$, which varies for different sets of training environments. 

While $\mathcal{E}(\widehat{\mu})$ captures stability at the level of the mean predictions, it does not account for possible changes in the shape or spread of the predicted distributions. To address this, we introduce a complementary notion of \textit{distributional stability}, which evaluates whether the entire predictive distribution remains consistent across environments. Among the considered methods, only DIV and Engression are generative and therefore allow sampling from the estimated target distribution. For each pair $e$ and $e^\prime$ of leave-one-environment-out models, we draw two samples from their respective estimated distributions $\widehat{P}_e$ and $\widehat{P}_{e^\prime}$ each, and compute the energy distance $\mathcal{D}_{\mathrm{E}}$ between them \citep{rizzo2016energy_dist}. Averaging these pairwise distances yields the distributional stability measure
$$\mathcal{E}(\widehat{P})
=\frac{2}{E(E-1)}
\sum_{e,e’}
\mathcal{D}_{\mathrm{E}}\big(\widehat{P}_e,\,\widehat{P}_{e^\prime}\big), $$
where $E$ denotes the total number of interventional environments.
Smaller values of $\mathcal{E}(\widehat{P})$ indicate that the models produce more consistent output distributions across environments.
In our experiments, DIV achieves a mean distributional stability of $0.016$ (SD$ =\!0.01$), substantially lower than $0.928$ (SD$ =\!0.03$) for engression. This indicates that DIV produces more consistent interventional distributions across environments, whereas the conditional distributions estimated by engression vary considerably with changes in the training interventions.

\section{Conclusion and future work}

In this paper, we propose a generative model-based approach, Distributional Instrumental Variable (DIV), for estimating the interventional distribution of the outcome in the presence of hidden confounding using instrumental variables. The flexibility of generative models enables the estimation of complex nonlinear causal effects without requiring the common additive noise assumption for either the treatment or response model. Furthermore, the distributional nature of our method allows for the estimation of the entire interventional distribution, rather than just the interventional mean, allowing for a richer understanding of causal effects.

We establish the identifiability of the interventional distribution for a general model class, accommodating both multivariate treatments and outcomes. Additionally, for the pre-additive noise outcome model class, we provide a novel identifiability result for the case of a binary instrument and multivariate continuous treatment---an `under-identified' setting where traditional methods such as 2SLS often fail.

The DIV method is computationally efficient, even for large-scale data. Our software implementation allows the estimation of interventional means and quantiles, while also allowing us to draw samples  from the estimated interventional distribution.

Note that our estimand is distributional, i.e.\ we model the full counterfactual distribution and thereby capture aleatoric uncertainty. Beyond the identifiability results and the consistency result provided for a simple parametric subclass, an important direction for future work is epistemic uncertainty quantification for flexible DIV implementations, for example via convergence-rate analysis and/or conformal methods, to provide uncertainty statements for causal functionals derived from the estimated interventional distribution.

Note that the current DIV approach relies on estimating the joint distribution of $(X,Y)|Z$. In many scenarios, the treatment $X$ often has significantly greater explanatory power for the outcome $Y$ than the instrument $Z$; for example, the conditional distribution of $Y|X,Z$ typically exhibits a much higher signal-to-noise ratio than that of $Y|Z$. 
A promising extension, which we explore in follow-up work, is a two-step estimation approach by first estimating the conditional distribution of $X|Z$, followed by estimating that of $Y|X,Z$.

Another potential direction is to apply the DIV methodology to distribution generalization, where the goal is not to estimate causal effects but to adapt to new environments (see, e.g.\ \citet{muandet2013generalization, christiansen2022, buhlmann2020invariance}). We believe this represents an exciting avenue for further research, with potential applications beyond causal inference.

\bibliography{biblio}
\bibliographystyle{abbrvnat}

\clearpage
\appendix

\section{Software} \label{sec:software}

The software implementation of the DIV method is available in \texttt{R} package \texttt{DistributionIV}, which is publicly available
at \url{http://CRAN.R-project.org/package=DistributionIV}. The generative neural network model is implemented using the \texttt{R}-package \texttt{torch} \citep{torch}, which also supports GPU acceleration. For moderate data sizes (up to several thousand observations and a few hundred variables), CPU-based training remains practical. For example, on a data set of size $10000$ trained for $5000$ epochs, our \texttt{R} implementation requires approximately 8 minutes on a GPU (MPS) and 12 minutes on a CPU (MacBook Pro, M1 chip).

The package provides functionality for point prediction of the interventional mean and interventional quantiles, as well as sampling from the estimated interventional distribution \(\Pint\) (see Section~\ref{subsec:estimation}). Additionally, the implementation allows for estimating conditional interventional distributions by incorporating exogenous variables into the treatment and outcome models, as described in Section~\ref{subsec:conditional_interventional}.

By default, the implementation uses high-dimensional noise with a dimension of 50 for both independent and shared noise terms, as empirical results indicate that this choice leads to estimation results that are both robust and flexible. Furthermore, we use multilayer perceptrons (MLPs) with four layers for both treatment and outcome models, with each layer containing $100$ neurons. The models are trained for $10000$ epochs with Adam optimiser and learning rate of $10^{-3}$, though empirical results suggest that for simpler scenarios, $5000$ or even fewer epochs are sufficient.

The example below demonstrates how the software can be used on a simple instrumental variable model. Note that for this particular example, the observational and the interventional distributions differ, and the DIV method aims to learn the interventional distribution.

\lstset{basicstyle=\ttfamily\footnotesize}
\begin{lstlisting}
# 1000 training & test samples
n_tr <- n_test <- 1000

# true underlying data generating process
g_lin <- function(Z, H, eps_X) return(Z + H + 0.1 * eps_X)
f_softplus <- function(X, H, eps_Y) return(log(1 + exp(X + 2 * H + eps_Y)))

# simulate observational data
eps_Xobs <- rnorm(n_tr); eps_Yobs <- rnorm(n_tr)
Zobs <- runif(n_tr, -3, 3); Hobs <- rnorm(n_tr, mean = 2)

Xobs <- g_lin(Z = Zobs, H = Hobs, eps_X = eps_Xobs)
Yobs <- f_softplus(X = Xobs, H = Hobs, eps_Y = eps_Yobs)

# simulate interventional data
eps_Xint <- rnorm(n_test); eps_Yint <- rnorm(n_test)
Zint <- runif(n_test, -3, 3)
# for generating interventional data, different H are used for X and Y
H1int <- rnorm(n_test, mean = 2); H2int <- rnorm(n_test, mean = 2)

Xint <- g_lin(Z = Zint, H = H1int, eps_X = eps_Xint)
Yint <- f_softplus(X = Xint, H = H2int, eps_Y = eps_Yint)

# fit DIV model
div_mod <- div(X = Xobs, Z = Zobs, Y = Yobs)
# predict interventional mean
predict(div_mod, Xtest = Xint, type = "mean")
predict(div_mod, Xtest = Xint, type = "quantile", quantiles = c(0.1, 0.5, 0.9))
# draw 10 samples from interventional distribution
predict(div_mod, Xtest = Xint, type = "sample", nsample = 10)
\end{lstlisting}

\section{Energy score} \label{app:energy_score}

The DIV method uses the expected negative energy score \citep{gneiting2007scoring} as a loss function to train the conditional generative model. The energy score is a scoring rule, originally used for evaluation of multivariate distributional forecasts.

The energy score is a strictly proper scoring rule, so it holds
$$\mathbb{E}_{U\sim P}[S(P,U)] \geq \mathbb{E}_{U\sim P}[S(P^\prime,U)],$$ with equality if and only if $P = P'$, ensuring that the expected score is maximised when $U$ is drawn from the true data-generating distribution $P$.

Minimization of the energy loss, which is defined as 
$$\mathcal{L}_e(P, P_0) = %
\underbrace{\mathbb{E}_{Y \sim P_0, U \sim P} \|U - Y\|}_{=:s_1} - \frac{1}{2} \underbrace{\mathbb{E}_{U, U^\prime \sim P}\|U - U^\prime\|}_{:=s_2},$$ therefore guarantees that the generated conditional distribution matches the observed conditional distribution.
The first term, $s_1$, corresponds to the prediction loss, while $s_2$ is the variation loss term, ensuring that samples from the generated distribution exhibit enough variability. The equality of both terms, $s_1 = s_2$, is a necessary condition for the energy loss to be minimised, which means that $P = P_0$. This correspondence can therefore be used as a sanity check during the model training process.

\section{Proofs of Section~\ref{sec:setting}} \label{app:setting}

\subsection{Proof of Proposition~\ref{prop:univariate_identify}}

\begin{proof}[Proof of Proposition~\ref{prop:univariate_identify}] 
The proof follows directly from Proposition~\ref{prop:ident_general} and Theorem~\ref{thm:interv_ident_general}.
\end{proof}

\subsection{Proof of Proposition~\ref{prop:binary_instrument_ident}}

\begin{proof}[Proof of Proposition~\ref{prop:binary_instrument_ident}] 
Proof relies on Proposition~\ref{prop:ident_bin_preadd}, using that $(X_j|Z=0) \noneqdist (c + X_j|Z=1)$ is sufficient for \ref{ass:d4}-\ref{ass:d5} to hold true, since for $j \in \{1, 2\}$
$$ (X_j|Z=0) \noneqdist (c + X_j|Z=1) \Rightarrow g_j(0,e_j) \neq c + g_j(1,e_j) \Rightarrow \pdv{g_j(0, e_j)}{e_j} \neq \pdv{g_j(1, e_j)}{e_j}. 
$$
\end{proof}

\section{Proofs of Section~\ref{sec:DIV_method}} \label{app:DIV_method}

\begin{proof}[Proof of Proposition~\ref{prop:popul_sol_joint}] 
Assume $P_{(X,Y)|Z}$ is induced by the SCM~\eqref{eq:IV_model}. Let $z \in \supp(Z)$ and $\varepsilon_H, \varepsilon_X, \varepsilon_Y$ standard Gaussians. We define a 4-tuple $(g^*, f^*, h_X^*, h_Y^*)$ with $$(g^*(z,h^*(\varepsilon_X, \varepsilon_H)), f^*(g^*(z,h^*(\varepsilon_X, \varepsilon_H)),h^*(\varepsilon_Y, \varepsilon_H))) \sim P^*_{(X,Y)|Z=z},$$ such that $P^*_{(X,Y)|Z=z} = P_{(X,Y)|Z=z}$ almost everywhere. Given any 4-tuple $(g^\diamond,f^\diamond,h^\diamond_X,h^\diamond_Y)$ with $$(g^\diamond(z,h^\diamond(\varepsilon_X, \varepsilon_H)), f^\diamond(g^\diamond(z,h^\diamond(\varepsilon_X, \varepsilon_H)),h^\diamond(\varepsilon_Y, \varepsilon_H))) \sim P^\diamond_{(X,Y)|Z=z},$$ following the generative model~\eqref{eq:joint_div_model}, assume there exists a subset $\mathcal{Z}^\prime \subseteq \supp(Z)$ with a non-zero base measure such that for all $z \in \mathcal{Z}^\prime$, $P^\diamond_{(X,Y)|Z=z} \neq P_{(X,Y)|Z=z}$. Then, according to the strict properness of the energy score, we have for all $z \in \mathcal{Z}^\prime$ $$\bbE_{(X,Y) \sim P_{(X,Y)|Z=z}}\bigl[ S(P^*_{(X,Y)|Z=z}, (X,Y)) \bigr] > \bbE_{(X,Y) \sim P_{(X,Y)|Z=z}}\bigl[ S(P^\diamond_{(X,Y)|Z=z}, (X,Y))\bigr].$$
Taking the expectation with respect to $P_{Z}$ then yields
$$\bbE_{P_{(X,Y,Z)}}\bigl[ S(P^*_{(X,Y)|Z}, (X,Y)) \bigr] > \bbE_{P_{(X,Y,Z)}}\bigl[ S(P^\diamond_{(X,Y)|Z}, (X,Y))\bigr].$$
Thus it holds for the expected negative energy score (which we define being the loss function):
$$\bbE_{P_{(X,Y,Z)}}\bigl[- S(P^*_{(X,Y)|Z}, (X,Y)) \bigr] < \bbE_{P_{(X,Y,Z)}}\bigl[ - S(P^\diamond_{(X,Y)|Z}, (X,Y))\bigr],$$
which concludes the proof.
    
\end{proof}

\section{Proofs of Section~\ref{sec:identifiability}} \label{app:ident_gen_mod}

\subsection{Proofs of results for general model class}

Recall  the class of structural causal models $\mathcal{M}_{\rm DIV}$ defined in the main text:
\begin{equation*}
\begin{cases}
    X_j \coloneqq g_j(Z, \eta_{X_j}), \forall \ j \in \{1,\dots, d\} \\
    Y_k \coloneqq f_k(X, \eta_{Y_k}), \forall \ k \in \{1,\dots, p\},
\end{cases}
\end{equation*}
where $Z \sim Q_Z$, $\eta_X \coloneqq (\eta_{X_1}, \dots,\eta_{X_d})$, $\eta_Y \coloneqq (\eta_{Y_1}, \dots,\eta_{Y_p})$ with $(\eta_X, \eta_Y) \sim Q_{(X,Y)}$, with $Z \in \bbR^q$ and $(\eta_X, \eta_Y)$ being independent. Further, we define $X \coloneqq (X_1, \dots, X_d), Y \coloneqq (Y_1, \dots, Y_p)$,  for all $j \in \{1,\dots, d\}: g_j \in \cG$, for all $k \in \{1,\dots, p\}: f_k \in \cF$, and $\mathcal{G} \subseteq \{g: \bbR^{q+1} \rightarrow \bbR\}$, $\mathcal{F} \subseteq \{f: \bbR^{d + 1} \rightarrow \bbR\}$ are function classes.

The identifiability refers to the uniqueness of a model that induces a single joint (conditional) distribution of $(X,Y)$ given $Z=z$. Informally, the theorem says that if two models from the class $\mathcal{M}_{\rm DIV}$ induce the same distribution of $(X,Y)$ given $Z=z$, then for all $j \in \{1, \dots, d\}$ their treatment models $g_j$, and for all $k \in \{1, \dots, p\}$ the outcome models $f_k$ are also the same for the given observed data support. Furthermore, we show distributional equality of the confounding effect $(\eta_X, \eta_Y)$.

\begin{proof}[Proof of Proposition~\ref{prop:ident_general}]
    The proof proceeds in 3 steps.
    \begin{itemize}
        \item In step I, we show identifiability of the treatment models $g_1, \dots, g_d$ along with $\eta_X$.
        \item In step II, we show identifiability of the outcome models $f_1, \dots, f_p$ along with $\eta_Y$.
        \item In step III, we combine the results from the previous two steps to conclude identifiability of the confounding effect $(\eta_X, \eta_Y)$.
    \end{itemize}

    \textit{Step I.} If for all $z \in \supp(Z)$ two models $\mathcal{M}$ and $\mathcal{M^\prime}$ from $\mathcal{M}_{DIV}$ induce the same joint (conditional) distribution of $(X,Y)$ given $Z=z$, then also the same marginal (conditional) distribution of $X$ given $Z=z$. 
    
    Fix $j \in \{1,\dots d\}$. The distributional equality implies \begin{equation}\label{eq:distr_eq_stepI_gen}
        g_j(z, \eta_{X_j}) \eqdist \tilde{g}_j(z, \tilde{\eta}_{X_j}).
    \end{equation}
    Then, for all $x \in \supp(X)$ it directly follows
    $$ P(g_j(z, \eta_{X_j}) \leq x) = P(\tilde{g}_j(z, \tilde{\eta}_{X_j}) \leq x). $$
    Since $g_j(z,\cdot), \tilde{g}_j(z,\cdot)$ are strictly monotone, this is equivalent to
    $$ P(\eta_{X_j} \leq g^{-1}_j(z, x)) = P(\tilde{\eta}_{X_j} \leq \tilde{g}^{-1}_j(z, x)). $$
    Without loss of generality assume $\eta_{X_j}, \tilde{\eta}_{X_j} \sim N(0,1)$, for all $x \in \supp(X_j)$ this results in
    $$g^{-1}_j(z,x) = \tilde{g}^{-1}_j(z,x)$$
     and thus for all $e_X \in \mathbb{R}$
    \begin{equation}\label{eq:g_equal}
        g_j(z,e_X) = \tilde{g}_j(z,e_X).
    \end{equation} 
    Next, from the distributional equality as stated in \eqref{eq:distr_eq_stepI_gen}, and \eqref{eq:g_equal}, we also have that
    $$(g_1(z, \eta_{X_1}), \dots, g_d(z, \eta_{X_d})) \eqdist (g_1(z, \tilde{\eta}_{X_1}), \dots, g_d(z, \tilde{\eta}_{X_d})).$$ Since $g_1(z,\cdot), \dots, g_d(z,\cdot)$ are strictly monotone, this is equivalent to
    \begin{equation}\label{eq:eta_equal}
        (\eta_{X_1}, \dots, \eta_{X_d}) \eqdist (\tilde{\eta}_{X_1}, \dots, \tilde{\eta}_{X_d}).
    \end{equation}

    \textit{Step II.} If two models $\mathcal{M}$ and $\mathcal{M^\prime}$ from $\mathcal{M}_{DIV}$ induce the same joint (conditional) distribution of $(X,Y)$ given $Z=z$ for all $z \in \supp(Z)$, then also the same conditional distribution of $Y$ given $X=x$, $Z=z$ for all $z \in \supp(Z)$ and $x \in \supp(X)$.
    Since $g_1, \dots, g_d$ are strictly monotone, the event $$X = x, Z = z$$ is equivalent to the event $$(X_1, \dots, X_d) = (x_1, \dots,x_d), (\eta_{X_1}, \dots, \eta_{X_d})  = (g^{-1}_1(z,x_1), \dots, g^{-1}_d(z,x_d)) := v.$$
    Fix $k \in \{1,\dots, p\}$. We now consider $f_k(x, \eta_{k, v})$ and $\tilde{f}_k(x, \tilde{\eta}_{k, v})$, where $\eta_{k, v} \eqdist (\eta_{Y_k}|\eta_X=v$), $\tilde{\eta}_{k,v} \eqdist (\tilde{\eta}_{Y_k}|\eta_X=v)$. For all $x \in \supp(X)$ and $v \in \{(g^{-1}_1(z,x_1), \dots, g^{-1}_d(z,x_d))|z \in \supp(Z)\}$, the distributional equality above implies $$f_k(x, \eta_{k, v}) \eqdist \tilde{f}_k(x, \tilde{\eta}_{k,v}),$$
    from which for all $y \in \supp(Y_k)$ it follows 
    \begin{equation}\label{eq:cdf_equal}
        F_{Y_k|X,\eta_X}(y|x,v) = F_{\tilde{Y}_k|X,\eta_X}(y|x,v),
    \end{equation}
    with $F$ being the corresponding conditional CDF.

    The conditional CDF of $Y_k$ given $X$ and $\eta_X$ can be written as
    $$F_{Y_k|X,\eta_X}(y|x,v) = \int \mathds{1}(f_k(x,e) \leq y) p_{\eta_{Y_k}|\eta_X}(e|v)de$$
    based on the conditional independence statement $\eta_Y \ind X|\eta_{X_k}$. This follows from the assumption of the instrument $Z$ being jointly independent of the noise $(\eta_X, \eta_Y)$ (e.g.\ \citet{saengkyongam2024identifying}, Lemma 8).

    The left-hand side is only defined in $x \in \supp(X)$, $v \in \{(g^{-1}_1(z,x_1), \dots, g^{-1}_d(z,x_d))|z \in \supp(Z)\}$. Using assumption \ref{ass:4}, we have $\{(g^{-1}_1(z,x_1), \dots, g^{-1}_d(z,x_d))|z \in \supp(Z)\} = \supp(\eta_X)$, and thus we can integrate out $v$ with respect to the marginal distribution of $\eta_X$ as follows:
    
    $$\int F_{Y_k|X,\eta_X}(y|x,v)p_{\eta_X}(v)dv = \int \int \mathds{1}(f_k(x,e) \leq y) p_{\eta_{Y_k}|\eta_X}(e|v)p_{\eta_X}(v)dedv = \int \mathds{1}(f_k(x,e) \leq y)p_{\eta_{Y_k}}(e)de.$$
    From this, combined with \eqref{eq:cdf_equal}, for all $x \in \supp(X)$ and $y \in \supp(Y_k)$ it follows
    $$\int \mathds{1}(f_k(x,e) \leq y)p_{\eta_{Y_k}}(e)de = \int \mathds{1}(\tilde{f}_k(x,e) \leq y)p_{\tilde{\eta}_{Y_k}}(e)de,$$
    which can be written as
    $$P(f_k(x,\eta_{Y_k}) \leq y) = P(\tilde{f}_k(x, \tilde{\eta}_{Y_k}) \leq y).$$
    Since $f_k(x,\cdot)$ and $\tilde{f}_k(x,\cdot)$ are both strictly monotone (by assumption \ref{ass:2}), this is equivalent to
    $$P(\eta_{Y_k} \leq f^{-1}_k(x, y)) = P(\tilde{\eta}_{Y_k} \leq \tilde{f}^{-1}_k(x, y)).$$
    Without loss of generality assume $\eta_{Y_k}, \tilde{\eta}_{Y_k} \sim N(0,1)$, this results in
    $$f^{-1}_k(x,y) = \tilde{f}^{-1}_k(x,y)$$
    and thus for all $x \in \supp(X)$ and $e_Y \in \mathbb{R}^p$
    \begin{equation} \label{eq:f_equal}
        f_k(x,e_Y) = \tilde{f}_k(x,e_Y).
    \end{equation}

    \textit{Step III.} Fix $k \in \{1,\dots, p\}.$ Using that the two models $\mathcal{M}$ and $\mathcal{M^\prime}$ induce the same conditional distribution of $Y$ given $X$ and $Z$ and \eqref{eq:f_equal}, it follows for all $x \in \supp(X)$ and $v \in \supp(\eta_X)$ (by assumption \ref{ass:4}) that 
    $$(f_1(x,\eta_{1, v}), \dots, f_p(x, \eta_{p, v})) \eqdist (f_1(x, \tilde{\eta}_{1, v}), \dots, f_p(x, \tilde{\eta}_{p, v})).$$
    Since $f_1(x, \cdot), \dots, f_p(x, \cdot)$ are strictly monotone, this leads to
    $$(\eta_{1, v}, \dots, \eta_{p, v}) \eqdist (\tilde{\eta}_{1, v}, \dots, \tilde{\eta}_{p, v}).$$
    With this distributional equality and \eqref{eq:eta_equal}, it follows $(\eta_X, \eta_Y) \eqdist (\tilde{\eta}_X, \tilde{\eta}_Y)$, and thus we conclude the identifiability of the confounding effect.
\end{proof}

Next, we present the proof for Theorem~\ref{thm:interv_ident_general} showing the identifiability of the interventional distribution $\Pint$.

\begin{proof}[Proof of Theorem~\ref{thm:interv_ident_general}]
    From Proposition~\ref{prop:ident_general} (using assumptions \ref{ass:1}, \ref{ass:3} and \ref{ass:4}), for all $k \in \{1, \dots, p\}$, $x \in \supp(X)$ and $y_k \in \supp(Y_k)$, the probability $P(f_k(x,\eta_{Y_k}) \leq y_k)$ is identifiable from the observed distribution $P_{(X,Y)|Z}$.
    Therefore, for any $(y_1,\dots,y_p) \in \supp(Y)$,
    \[
    P_Y^{\do(X:=x)}(Y_1\le y_1,\dots,Y_p\le y_p)
    =P\!\big(f_1(x,\eta_{Y_1})\le y_1,\dots,f_p(x,\eta_{Y_p})\le y_p\big),
    \]
    which is identifiable from $P_{(X,Y)\mid Z}$ and equals the CDF of $P_Y^{\mathrm{do}(X:=x)}$.
\end{proof}

\subsection{Proofs of results for pre-ANM model class} \label{app:ident_pre_ANM}

We now recall the class of pre-additive noise models $\mathcal{M}_{\rm DIV}^{\pre}$ as defined in the main text:

\begin{equation*}
\begin{cases}
    X_j \coloneqq g_j(Z, \eta_{X_j}), \forall \ j \in \{1,\dots, d\} \\
    Y_k \coloneqq f_k(X^\top \beta_k + \eta_{Y_k}), \forall \ k \in \{1,\dots, p\},
\end{cases}
\end{equation*}
where $Z \sim Q_Z$, $\eta_X \coloneqq (\eta_{X_1}, \dots,\eta_{X_d})$, $\eta_Y \coloneqq (\eta_{Y_1}, \dots,\eta_{Y_p})$ with $(\eta_X, \eta_Y) \sim Q_{(X,Y)}$, with $Z$ and $(\eta_X, \eta_Y)$ being independent. Further, we define $X \coloneqq (X_1, \dots, X_d), \beta_k = (1, \beta_{k,2}, \dots, \beta_{k,d}), Y \coloneqq (Y_1, \dots, Y_p)$,  for all $j \in \{1,\dots, d\}: g_j \in \tilde{\cG}$, for all $k \in \{1,\dots, p\}: f_k \in \tilde{\cF}$, and $\tilde{\mathcal{G}} \subseteq \{g: \bbR \rightarrow \bbR\}$, $\tilde{\mathcal{F}} \subseteq \{f: \bbR \rightarrow \bbR\}$ are function classes.

\subsubsection{Continuous instrument}

\begin{proof}[Proof of Proposition~\ref{prop:ident_cont_preadd}]
   Analogously to the proof of Proposition~\ref{prop:ident_general}, the proof proceeds in 3 steps.
    
    \textit{Step I.} As in Step I in the proof of Proposition~\ref{prop:ident_general}, we have for all $j \in \{1,\dots ,d\}$, $z \in \supp(Z)$ and $e_X \in \bbR^d$ that $$g_j(z, e_X) = \tilde{g}_j(z,e_X).$$ 
    
    Furthermore, we have that
    \begin{equation}\label{eq:eta_equal2}
        (\eta_{X_1}, \dots, \eta_{X_d}) \eqdist (\tilde{\eta}_{X_1}, \dots, \tilde{\eta}_{X_d}).
    \end{equation}

    \textit{Step II}. If two models $\mathcal{M}$ and $\mathcal{M^\prime}$ from $\mathcal{M}_{\rm DIV}$ induce the same joint (conditional) distribution of $(X,Y)$ given $Z=z$ for all $z \in \supp(Z)$, then also the same conditional distribution of $Y$ given $X_j = Q_{\gamma^j}(X_j|Z=z)$ for all $j \in \{1, \dots, d\}$, $Z=z$ for all $z \in \supp(Z)$ and a fix $\gamma \in [0,1]$.
    Since for all $j \in \{1,\dots, d\}$, $g_j$ is strictly monotone, for a fix $j$ the event 
    \begin{align*}
        X_j = Q_{\gamma^j}(X_j|Z=z) &= Q_{\gamma^j}(g_j(Z, \eta_{X_j})|Z=z) \overset{Z \ind \eta_{X_j}}{=} Q_{\gamma^j}(g_j(z, \eta_{X_j})) \\ 
        &=g_j(z, Q_{\gamma^j}(\eta_{X_j}))
    \end{align*}
    is equivalent to the event 
    $$\eta_{X_j} = Q_{\gamma^j}(\eta_{X_j}) =: e_j.$$
    Now, we define $e_{\texttt{2:d}} \coloneqq (e_2,\dots,e_d)$ and $\bar{g}: (z, e) \mapsto (g_2(z, e_2), \dots, g_d(z, e_d))$. Further, for all $k \in \{1,\dots,p\}$, define $\beta_{k,\texttt{2:d}} \coloneqq (\beta_{k,2}, \dots, \beta_{k,d})$ and $\tilde{\beta}_{k,\texttt{2:d}} \coloneqq (\tilde{\beta}_{k,2}, \dots, \tilde{\beta}_{k,d})$.
    Fix $k \in \{1,\dots,p\}$, we now consider $f_k(g_1(z, e_1) + \bar{g}(z,e_{2\texttt{+}})^\top \beta_{k,\texttt{2:d}} + \eta_{k,e}^*)$ and 
    $\tilde{f}_k(g_1(z, e_1) + \bar{g}(z,e_{2\texttt{+}})^\top \tilde{\beta}_{k,\texttt{2:d}} + \tilde{\eta}_{k,e}^*)$, where $\eta_{k,e}^* \eqdist (\eta_{Y_k}|\eta_X=e)$, $\tilde{\eta}_{k,e}^* \eqdist (\tilde{\eta}_{Y_k}|\eta_X=e)$. The distributional equality above implies
    $$
    f_k(g_1(z, e_1) + \bar{g}(z,e_{\texttt{2:d}})^\top \beta_{k,\texttt{2:d}} + \eta_{k,e}^*)
    \eqdist
    \tilde{f}_k(g_1(z, e_1) + \bar{g}(z,e_{\texttt{2:d}})^\top \tilde{\beta}_{k,\texttt{2:d}} + \tilde{\eta}_{k,e}^*)
    $$
    We assume $(\eta_{X_j},\eta_Y)$ being jointly independent of $Z$ and absolutely continuous with respect to the Lebesgue measure. From this, it directly follows $\eta_{k,e}^*, \tilde{\eta}_{k,e}^*$ are independent of $Z$ (e.g.\ \citet{saengkyongam2024identifying}, Lemma 8) and absolutely continuous with respect to the Lebesgue measure, so that there exist strictly monotone functions $h_{k,e}$, $\tilde{h}_{k,e}$ such that $\eta_{k,e}^* \eqdist h_{k,e}(\varepsilon_Y)$ and $\tilde{\eta}_{k,e}^* \eqdist \tilde{h}_{k,e}(\varepsilon_Y)$ with $\varepsilon_Y \sim \unif[0,1]$.

    Since distributional equality induces the equality of all quantiles, and due to strict monotonicity of $f_k$, $\tilde{f}_k$, for all $e_Y \in [0,1]$ it then holds: 
    \begin{equation}\label{eq:step2_pre-ANM_cont}
        f_k(g_1(z, e_1) + \bar{g}(z,e_{\texttt{2:d}})^\top \beta_{k,\texttt{2:d}} + h_{k,e}(e_Y))
    =
    \tilde{f}_k(g_1(z, e_1) + \bar{g}(z,e_{\texttt{2:d}})^\top \tilde{\beta}_{k,\texttt{2:d}} + \tilde{h}_{k,e}(e_Y))
    \end{equation}
    Since $f_k$ strictly monotone, we take the inverse of it on both sides of \eqref{eq:step2_pre-ANM_cont}:
        \begin{equation} \label{eq:step2_inverse_f}
        g_1(z, e_1) + \bar{g}(z,e_{\texttt{2:d}})^\top \beta_{k,\texttt{2:d}} + h_{k,e}(e_Y) = \underbrace{f_k^{-1}\tilde{f_k}}_{:=\phi_k}(g_1(z, e_1) + \bar{g}(z,e_{\texttt{2:d}})^\top \tilde{\beta}_{k,\texttt{2:d}} + \tilde{h}_{k,e}(e_Y)).
    \end{equation}
    Taking partial derivative on both sides with respect to $e_1$ and $e_Y$, yields
    \begin{align}
        \pdv{g_1(z, e_1)}{e_1} + \pdv{h_{k,e}(e_Y)}{e_1} &= \phi_k^\prime(g_1(z, e_1) + \bar{g}(z,e_{\texttt{2:d}})^\top \tilde{\beta}_{k,\texttt{2:d}} + \tilde{h}_{k,e}(e_Y))(\pdv{g_1(z, e_1)}{e_1} + \pdv{\tilde{h}_{k,e}(e_Y)}{e_1}) \label{eq:step2_pdv_e1} \\
        \pdv{h_{k,e}(e_Y)}{e_Y} &= \phi_k^\prime(g_1(z, e_1) + \bar{g}(z,e_{\texttt{2:d}})^\top \tilde{\beta}_{k,\texttt{2:d}} + \tilde{h}_{k,e}(e_Y))\pdv{\tilde{h}_{k,e}(e_Y)}{e_Y} \label{eq:step2_pdv_eY}
    \end{align}
    Since $h_{k,e}$ is strictly monotone, 
    we substitute \eqref{eq:step2_pdv_e1} with \eqref{eq:step2_pdv_eY} and rearrange the terms:
    \begin{equation}\label{eq:step2_g1_nonlinear}
        \pdv{g_1(z, e_1)}{e_1} \bigg(\pdv{\tilde{h}_{k,e}(e_Y)}{e_Y} - \pdv{h_{k,e}(e_Y)}{e_Y}\bigg) = \pdv{h_{k,e}(e_Y)}{e_Y}\pdv{\tilde{h}_{k,e}(e_Y)}{e_1} - \pdv{h_{k,e}(e_Y)}{e_1} \pdv{\tilde{h}_{k,e}(e_Y)}{e_Y}.
    \end{equation}
    Since the right-hand side of \eqref{eq:step2_g1_nonlinear} does not depend on $z$, and relying on the assumption \ref{ass:c4}, we have that
    \begin{equation}\label{eq:step2_equal_h}
        \pdv{\tilde{h}_{k,e}}{e_Y} = \pdv{h_{k,e}}{e_Y}.
    \end{equation}
    Plugging \eqref{eq:step2_equal_h} in \eqref{eq:step2_pdv_eY} and using that $\tilde{h}_{k,e}$ is strictly monotone yields
    \begin{equation}\label{eq:step2_equal_phi}
        1 = \phi_k^\prime(g_1(z, e_1) + \bar{g}(z,e_{\texttt{2:d}})^\top \tilde{\beta}_{k,\texttt{2:d}} + \tilde{h}_{k,e}(e_Y)).
    \end{equation}
    Then, from \eqref{eq:step2_equal_phi}, we have for all $w \in \{x^\top \beta_k + e_Y \mid x \in \supp(X),  e_Y \in \supp(\eta_{Y_k})\}$
    \begin{equation} \label{eq:f_equality}
        f_k(w + \tilde{c}) = \tilde{f}_k(w),
    \end{equation}
    where $\tilde{c} \in \mathbb{R}$ is a constant.
    Next, we plug \eqref{eq:f_equality} in \eqref{eq:step2_pre-ANM_cont} and get
    \begin{equation*}
        f_k(g_1(z, e_1) + \bar{g}(z,e_{\texttt{2:d}})^\top \beta_{k,\texttt{2:d}} + h_{k,e}(e_Y)) = f_k(g_1(z, e_1) + \bar{g}(z,e_{\texttt{2:d}})^\top \tilde{\beta}_{k,\texttt{2:d}} + \tilde{h}_{k,e}(e_Y) + \tilde{c}).
    \end{equation*}
    Using that $f_k$ is strictly monotone, we then have
    \begin{equation*}
        \bar{g}(z,e_{\texttt{2:d}})^\top \beta_{k,\texttt{2:d}} + h_{k,e}(e_Y) = \bar{g}(z,e_{\texttt{2:d}})^\top \tilde{\beta}_{k,\texttt{2:d}} + \tilde{h}_{k,e}(e_Y) + \tilde{c}.
    \end{equation*}
    We take the derivative on both sides with respect to $z$, yielding
    \begin{equation}\label{eq:step2:Jg_bk}
        \mathbf{J}_{g}(z, e_{\texttt{2:d}}) (\beta_{k,\texttt{2:d}} - \tilde{\beta}_{k,\texttt{2:d}}) = 0,
    \end{equation}
    where $\mathbf{J}_{g}(z, e_{\texttt{2:d}}) \coloneqq \begin{bmatrix}
        \pdv{g_2(z, e_2)}{z_1} & \dots & \pdv{g_d(z, e_d)}{z_1} \\
        \vdots & \ddots & \vdots \\
        \pdv{g_2(z, e_2)}{z_q} & \dots & \pdv{g_d(z, e_d)}{z_q}
    \end{bmatrix}$.
    Using assumption~\ref{ass:c5}, we can then conclude from \eqref{eq:step2:Jg_bk} that $\beta_{k,\texttt{2:d}} = \tilde{\beta}_{k,\texttt{2:d}}$, and thus
    \begin{equation}\label{eq:step2_beta_equal}
         \beta_k = \tilde{\beta}_k.
    \end{equation}
    
    Next, define $\eta_{X_{\texttt{2:d}}} \coloneqq (\eta_{X_2}, \dots ,\eta_{X_d})$. For all $z \in \supp(Z)$ (and a fix $k$), the distributional equality of $Y$ given $Z=z$ is induced by
    $$f_k(g_1(z, \eta_{X_1}) + \bar{g}(z,\eta_{X_{\texttt{2:d}}})^\top \beta_{k,\texttt{2:d}} + \eta_{Y_k}) \eqdist \tilde{f}_k(g_1(z, \eta_{X_1}) + \bar{g}(z,\eta_{X_{\texttt{2:d}}})^\top \tilde{\beta}_{k,\texttt{2:d}} + \tilde{\eta}_{Y_k}).$$
    Combining this with \eqref{eq:f_equality} and \eqref{eq:step2_beta_equal}, we get
    $$
        f_k(g_1(z, \eta_{X_1}) + \bar{g}(z,\eta_{X_{\texttt{2:d}}})^\top \beta_{k,\texttt{2:d}} + \eta_{Y_k}) \eqdist f_k(g_1(z, \eta_{X_1}) + \bar{g}(z,\eta_{X_{\texttt{2:d}}})^\top \beta_{k,\texttt{2:d}} + \tilde{\eta}_{Y_k} + \tilde{c}).
    $$
    From this, since $f_k$ is strictly monotone, it holds 
    $$g_1(z, \eta_{X_1}) + \bar{g}(z,\eta_{X_{\texttt{2:d}}})^\top \beta_{k,\texttt{2:d}} + \eta_{Y_k} \eqdist g_1(z, \eta_{X_1}) + \bar{g}(z,\eta_{X_{\texttt{2:d}}})^\top \beta_{k,\texttt{2:d}} + \tilde{\eta}_{Y_k} + \tilde{c}$$ 
    and therefore $$\eta_{Y_k} \eqdist \tilde{\eta}_{Y_k} + \tilde{c}.$$
    
    Without loss of generality assume $\eta_{Y_k}, \tilde{\eta}_{Y_k}$ having median zero (assumption \ref{ass:c3}), it follows $\tilde{c} = 0$, and thus for all $w \in \{x^\top \beta_k + e_Y \mid x \in \supp(X),  e_Y \in \supp(\eta_{Y_k})\}$ it holds 
    \begin{equation} \label{eq:f_k_equal_pre_ANM}
        f_k(w) = \tilde{f}_k(w).
    \end{equation}

    \textit{Step III.} Using that two models $\mathcal{M}$ and $\mathcal{M^\prime}$ from $\mathcal{M}_{\rm DIV}^{\pre}$ induce the same conditional distribution of $Y$ given $X$ and $Z$ and \eqref{eq:f_k_equal_pre_ANM}, it follows for all $z \in \supp(Z)$ and $e \in \supp(\eta_X)$ that
    \begin{align*}
        &(f_1(g_1(z, e_1) + \bar{g}(z,e_{\texttt{2:d}})^\top \beta_{k,\texttt{2:d}} + \eta_{1,e}^*), \dots, f_p(g_1(z, e_1) + \bar{g}(z,e_{\texttt{2:d}})^\top \beta_{k,\texttt{2:d}} + \eta_{p,e}^*) \\ 
        \eqdist \ &(f_1(g_1(z, e_1) + \bar{g}(z,e_{\texttt{2:d}})^\top \beta_{k,\texttt{2:d}} + \tilde{\eta}_{1,e}^*), \dots, f_p(g_1(z, e_1) + \bar{g}(z,e_{\texttt{2:d}})^\top \beta_{k,\texttt{2:d}} + \tilde{\eta}_{p,e}^*).
    \end{align*}
    Since $f_1, \dots, f_p$ are strictly monotone, it follows
    $$(\eta_{1,e}^*, \dots, \eta_{p,e}^*) \eqdist (\tilde{\eta}_{1, e}^*, \dots, \tilde{\eta}_{p, e}^*)$$
    for any fix $e \in \supp(\eta_X)$, and with this we recover the marginal (conditional) distribution of $\eta_Y|\eta_X$.
    With this and \eqref{eq:eta_equal2}, it follows $(\eta_X, \eta_Y) \eqdist (\tilde{\eta}_X, \tilde{\eta}_Y)$, and thus we conclude the identifiability of the confounding effect.
    
\end{proof}

Next, we present the proof for Theorem~\ref{thm:interv_ident_cont_preadd} showing the identifiability of the interventional distribution $\Pint$.

\begin{proof}[Proof of Theorem~\ref{thm:interv_ident_cont_preadd}]
        From Proposition~\ref{prop:ident_cont_preadd}, the functions $f_1, \dots, f_p$ and the distribution of the noise $(\eta_{Y_1}, \dots, \eta_{Y_p})$ are identifiable from the observed distribution $P_{(X,Y)|Z}$. We can then identify the following: 
        $$P^{\mathrm{do}(X:=x)}(Y_1 \leq y_1, \dots, Y_p \leq y_p) = P(f_1(x^\top \beta_1 + \eta_{Y_1}) \leq y_1, \dots, f_p(x^\top \beta_p + \eta_{Y_p}) \leq y_p),$$
        which is the CDF of the required interventional distribution $P_Y^{\mathrm{do}(X:=x)}$.
    \end{proof}

\subsubsection{Discrete instrument}

\begin{proof}[Proof of Proposition~\ref{prop:ident_bin_preadd}]
        The proof proceeds in 3 steps.
    \begin{itemize}
        \item In step I, we show identifiability of the treatment models $g_1, \dots, g_d$.
        \item In step II, we show identifiability of the outcome models $f_1, \dots, f_p$.
        \item In step III, we combine the results from the previous two steps to conclude identifiability of the confounding effect $(\eta_X, \eta_Y)$.
    \end{itemize}

    \textit{Step I.} 
        As in Step I in the proof of Proposition~\ref{prop:ident_cont_preadd}, we have for all $j \in \{1,\dots ,d\}$, $z \in \supp(Z)$ and $e_X \in \bbR^d$ that $$g_j(z, e_X) = \tilde{g}_j(z,e_X).$$ 
    
    Furthermore, we have that
    \begin{equation}
        (\eta_{X_1}, \dots, \eta_{X_d}) \eqdist (\tilde{\eta}_{X_1}, \dots, \tilde{\eta}_{X_d}).
    \end{equation}

    \textit{Step II.}
    As in Step II in the proof of Proposition~\ref{prop:ident_cont_preadd} (relying on assumption \ref{ass:d4}), we have for all $w \in \{x^\top \beta_k + e_Y \mid x \in \supp(X),  e_Y \in \supp(\eta_{Y_k})\}$ (see \eqref{eq:f_equality})
    \begin{equation} \label{eq:f_equality_Z_cat}
        f_k(w + c) = \tilde{f}_k(w),
    \end{equation}
     with $c \in \mathbb{R}$ being a constant, and using that $f_k$ is strictly monotone, we then get
    \begin{equation*}
        \bar{g}(z,e_{\texttt{2:d}})^\top \beta_{k,\texttt{2:d}} + h_{k,e}(e_Y) = \bar{g}(z,e_{\texttt{2:d}})^\top \tilde{\beta}_{k,\texttt{2:d}} + \tilde{h}_{k,e}(e_Y) + c.
    \end{equation*}
    For any $z_1, z_2 \in \supp(Z)$, we therefore have
    \begin{equation*}
        (\bar{g}(z_1,e_{\texttt{2:d}}) - \bar{g}(z_2,e_{\texttt{2:d}}))^\top (\beta_{k,\texttt{2:d}} - \tilde{\beta}_{k,\texttt{2:d}}) = 0.
    \end{equation*}
    We now take the derivative on both sides with respect to $e_{\texttt{2:d}}$, yielding \\
    \begin{equation}\label{eq:step2:Jg_bk_e2}
        \tilde{\mathbf{J}}_{g}(z, e_{\texttt{2:d}}) (\beta_{k,\texttt{2:d}} - \tilde{\beta}_{k,\texttt{2:d}}) = 0,
    \end{equation}
    where $\tilde{\mathbf{J}}_{g}(z, e_{\texttt{2:d}}) \coloneqq \begin{bmatrix}
        \pdv{(g_2(z_1, e_2) - g_2(z_2, e_2))}{e_2} & \dots & 0 \\
        \vdots & \ddots & \vdots \\
        0 & \dots & \pdv{(g_d(z_1, e_d) - g_d(z_2, e_d))}{e_d}
    \end{bmatrix}$. \\

    Using assumption~\ref{ass:d5}, we can argue that the left-hand side of \eqref{eq:step2:Jg_bk_e2} depends on $z$, while the right-hand side does not depend on $z$. From this, it follows that the only way for \eqref{eq:step2:Jg_bk_e2} to hold for all $z \in \supp(Z)$ is if $\beta_{k,\diamond} = \tilde{\beta}_{k,\diamond}$, and thus we conclude
    \begin{equation}\label{eq:step2_beta_equal_Z_cat}
         \beta_k = \tilde{\beta}_k.
    \end{equation}
    The remaining part of Step II (compare argumentation following \eqref{eq:step2_beta_equal}) proceeds with the same reasoning as previously established in the proof of Proposition~\ref{prop:ident_cont_preadd}. Thus for all $w \in \{x^\top \beta_k + e_Y \mid x \in \supp(X),  e_Y \in \supp(\eta_{Y_k})\}$ it holds 
    \begin{equation} \label{eq:f_k_equal_pre_ANM_Z_cat}
        f_k(w) = \tilde{f}_k(w).
    \end{equation}
    \textit{Step III.} 
        Analogously to Step III in the proof of Proposition~\ref{prop:ident_cont_preadd}, we have $(\eta_X, \eta_Y) \eqdist (\tilde{\eta}_X, \tilde{\eta}_Y)$, and thus we conclude the identifiability of the confounding effect.      
    \end{proof}

The proof for Theorem~\ref{thm:interv_ident_bin_preadd} showing the identifiability of the interventional distribution $\Pint$ if the instrument $Z$ is discrete is exactly the same as for Theorem~\ref{thm:interv_ident_cont_preadd}.

\subsection{Consistency} \label{app:consistency}

We provide a self-contained proof of Theorem~\ref{thm:consistency_linear} in a simple univariate linear
parametric class by combining (i) strict propriety of the energy score and (ii) uniform convergence of the
empirical objective over a compact parameter set.

We consider the univariate linear model with a latent confounder,
\begin{equation*}
X = aZ + gH + \varepsilon_X,
\qquad
Y = bX + dH + \varepsilon_Y,
\qquad
Z \ind (H,\varepsilon_X,\varepsilon_Y),
\end{equation*}
with parameter $\theta=(a,b,g,d)\in\Theta\subset\mathbb R^4$.
Define
\[
W := (Z,H,\varepsilon_X,\varepsilon_Y)^\top\in\mathbb R^4,
\qquad
U := (X,Y)^\top\in\mathbb R^2.
\]
Then $U=A_\theta W$ with
\begin{equation}\label{eq:U=AW_cons}
A_\theta :=
\begin{pmatrix}
a & g & 1 & 0 \\
ba & bg+d & b & 1
\end{pmatrix}\in\mathbb R^{2\times 4}.
\end{equation}
For fixed $z$, let $W(z):=(z,H,\varepsilon_X,\varepsilon_Y)^\top$ and $U_\theta(z):=A_\theta W(z)$, and denote the
model-implied conditional law by $P_{U\mid Z=z}^\theta:=\mathcal L(U_\theta(z))$.

Given an observation $v:=(x,y)^\top\in\mathbb R^2$, the energy score is
\[
\mathrm{ES}(P,v):=\mathbb E\|U-v\|-\frac12\,\mathbb E\|U-U'\|,
\qquad U,U'\stackrel{i.i.d.}{\sim}P.
\]
Define the pointwise expected score and the population / empirical objectives
\begin{align}
\ell_\theta(x,y,z)
&:= \mathbb E\|U_\theta(z)-v\|-\frac12\,\mathbb E\|U_\theta(z)-U_\theta'(z)\|,
\qquad v=(x,y)^\top, \label{eq:ell_cons}\\
L(\theta) &:= \mathbb E[\ell_\theta(X,Y,Z)],
\qquad
\hat L_n(\theta):=\frac1n\sum_{i=1}^n \ell_\theta(X_i,Y_i,Z_i). \label{eq:Lhat_cons}
\end{align}
The estimator is $\hat\theta_n\in\arg\min_{\theta\in\Theta}\hat L_n(\theta)$.

We verify the continuity and uniform convergence assumptions in Theorem~\ref{thm:consistency_linear} under
simple boundedness and compactness conditions.

Let the linear model in \eqref{eq:lin_iv_consistency} satisfy the following boundedness and compactness conditions.
There exist constants $M\ge 1$ and $R>0$ such that:
\begin{enumerate}[label=(UB\arabic*)]
\item\label{ass:Theta_compact}
(Compact parameter set.) $\Theta=[-M,M]^4$.
\item\label{ass:bounded_supports}
(Bounded supports of exogenous variables.)
\[
\supp(Z)\subset[-R,R],\quad
\supp(H)\subset[-R,R],\quad
\supp(\varepsilon_X)\subset[-R,R],\quad
\supp(\varepsilon_Y)\subset[-R,R].
\]
\end{enumerate}

\begin{lemma}[Uniform bounds]\label{lem:bounds_cons}
Under Assumptions~\ref{ass:Theta_compact}--\ref{ass:bounded_supports}, for all $z\in\supp(Z)$,
\[
\|W(z)\|_2\le 2R.
\]
Moreover, there exists $B_U<\infty$ such that for all $\theta\in\Theta$ and all $z\in\supp(Z)$,
\[
\|U_\theta(z)\|_2\le B_U.
\]
\end{lemma}

\begin{proof}
Since $W(z)\in[-R,R]^4$, we have $\|W(z)\|_2\le \sqrt{4R^2}=2R$.
Also $\|U_\theta(z)\|_2=\|A_\theta W(z)\|_2\le \|A_\theta\|_{\mathrm{op}}\|W(z)\|_2\le \|A_\theta\|_F\|W(z)\|_2$.
Because $\Theta$ is compact and $\theta\mapsto \|A_\theta\|_F$ is continuous, $\sup_{\theta\in\Theta}\|A_\theta\|_F<\infty$,
so we may take $B_U:=2R\sup_{\theta\in\Theta}\|A_\theta\|_F$.
\end{proof}

\begin{lemma}[Lipschitzness of $\theta\mapsto A_\theta$ in Frobenius norm]\label{lem:A_lip_cons}
Under Assumptions~\ref{ass:Theta_compact}--\ref{ass:bounded_supports}, there exists $C_A<\infty$ such that for all $\theta,\theta'\in\Theta$,
\[
\|A_\theta-A_{\theta'}\|_F \le C_A\,\|\theta-\theta'\|_2.
\]
One valid choice is $C_A:=\sqrt{3+13M^2}$.
\end{lemma}

\begin{proof}
Write $\theta=(a,b,g,d)$ and $\theta'=(a',b',g',d')$. From \eqref{eq:U=AW_cons},
\[
A_\theta-A_{\theta'}=
\begin{pmatrix}
a-a' & g-g' & 0 & 0\\
ba-b'a' & (bg+d)-(b'g'+d') & b-b' & 0
\end{pmatrix}.
\]
For the linear entries, $|a-a'|,|g-g'|,|b-b'|\le \|\theta-\theta'\|_2$.
For the bilinear term,
\[
|ba-b'a'|=|b(a-a')+a'(b-b')|
\le |b||a-a'|+|a'||b-b'|
\le 2M\|\theta-\theta'\|_2.
\]
For $(bg+d)-(b'g'+d')$,
\[
|(bg+d)-(b'g'+d')|
\le |b(g-g')|+|g'(b-b')|+|d-d'|
\le (M+M+1)\|\theta-\theta'\|_2
\le 3M\|\theta-\theta'\|_2,
\]
using $M\ge1$.
Thus the five possibly nonzero entries are bounded by
$1,1,2M,3M,1$ times $\|\theta-\theta'\|_2$, so
\[
\|A_\theta-A_{\theta'}\|_F^2
\le (1^2+1^2+(2M)^2+(3M)^2+1^2)\,\|\theta-\theta'\|_2^2
=(3+13M^2)\,\|\theta-\theta'\|_2^2.
\]
\end{proof}

\begin{lemma}[Uniform Lipschitzness of $U_\theta(z)$ in $\theta$]\label{lem:U_lip_cons}
Under Assumptions~\ref{ass:Theta_compact}--\ref{ass:bounded_supports}, for all $\theta,\theta'\in\Theta$ and all $z\in\supp(Z)$ (with coupled noises),
\[
\|U_\theta(z)-U_{\theta'}(z)\|_2
\le C_U\,\|\theta-\theta'\|_2,
\qquad
C_U:=2R\,C_A,
\]
where $C_A$ is as in Lemma~\ref{lem:A_lip_cons}.
\end{lemma}

\begin{proof}
Couple the draws by using the same $W(z)$ for $\theta$ and $\theta'$. Then
$U_\theta(z)-U_{\theta'}(z)=(A_\theta-A_{\theta'})W(z)$, hence
\[
\|U_\theta(z)-U_{\theta'}(z)\|_2
\le \|A_\theta-A_{\theta'}\|_{\mathrm{op}}\|W(z)\|_2
\le \|A_\theta-A_{\theta'}\|_F\|W(z)\|_2.
\]
Use Lemma~\ref{lem:bounds_cons} and Lemma~\ref{lem:A_lip_cons}.
\end{proof}

\begin{lemma}[Uniform boundedness and Lipschitzness of $\ell_\theta$ and $L(\theta)$]\label{lem:ell_lip_cons}
Under Assumptions~\ref{ass:Theta_compact}--\ref{ass:bounded_supports}, for all $(x,y,z)$ in the support of $(X,Y,Z)$ and all $\theta,\theta'\in\Theta$,
\[
|\ell_\theta(x,y,z)-\ell_{\theta'}(x,y,z)|
\le L_\ell\,\|\theta-\theta'\|_2,
\qquad
L_\ell:=2C_U.
\]
Moreover, $|\ell_\theta(x,y,z)|\le B_\ell$ uniformly over $\theta\in\Theta$ and $(x,y,z)$ in the support, where one may take
$B_\ell:=3B_U$ with $B_U$ from Lemma~\ref{lem:bounds_cons}. Consequently, $L$ is (globally) Lipschitz on $\Theta$, hence continuous.
\end{lemma}

\begin{proof}
Fix $(x,y,z)$ and set $v=(x,y)^\top$. Couple model draws across $\theta,\theta'$ using the same base noises
to form $(U_\theta(z),U_{\theta'}(z))$ and an independent copy $(U_\theta'(z),U_{\theta'}'(z))$.

\emph{Boundedness.}
By Lemma~\ref{lem:bounds_cons}, $\|U_\theta(z)\|_2\le B_U$ uniformly.
Since $(X,Y)$ arise from the same bounded-support model class, also $\|v\|_2\le B_U$ on the support.
Thus $\|U_\theta(z)-v\|\le 2B_U$ and $\|U_\theta(z)-U_\theta'(z)\|\le 2B_U$.
Hence from \eqref{eq:ell_cons},
\[
|\ell_\theta(x,y,z)|
\le \mathbb E\|U_\theta(z)-v\|+\frac12\mathbb E\|U_\theta(z)-U_\theta'(z)\|
\le 2B_U + B_U = 3B_U.
\]

\emph{Lipschitzness.}
Using the reverse triangle inequality,
\[
\big|\|u-v\|-\|u'-v\|\big|\le \|u-u'\|,
\]
we obtain
\[
\left|\mathbb E\|U_\theta(z)-v\|-\mathbb E\|U_{\theta'}(z)-v\|\right|
\le \mathbb E\|U_\theta(z)-U_{\theta'}(z)\|.
\]
Similarly,
\[
\left|\mathbb E\|U_\theta(z)-U_\theta'(z)\|-\mathbb E\|U_{\theta'}(z)-U_{\theta'}'(z)\|\right|
\le 2\,\mathbb E\|U_\theta(z)-U_{\theta'}(z)\|.
\]
Combine with \eqref{eq:ell_cons} to get
\[
|\ell_\theta(x,y,z)-\ell_{\theta'}(x,y,z)|
\le 2\,\mathbb E\|U_\theta(z)-U_{\theta'}(z)\|.
\]
Apply Lemma~\ref{lem:U_lip_cons} to bound the RHS by $2C_U\|\theta-\theta'\|_2$.
Finally,
\[
|L(\theta)-L(\theta')|
=\left|\mathbb E[\ell_\theta(X,Y,Z)-\ell_{\theta'}(X,Y,Z)]\right|
\le \mathbb E\big[|\ell_\theta-\ell_{\theta'}|\big]
\le L_\ell\|\theta-\theta'\|_2,
\]
so $L$ is Lipschitz, hence continuous.
\end{proof}

\begin{lemma}[Uniform convergence of $\hat L_n$]\label{lem:unifconv_cons}
Under Assumptions~\ref{ass:Theta_compact}--\ref{ass:bounded_supports},
\[
\sup_{\theta\in\Theta}\big|\hat L_n(\theta)-L(\theta)\big| \xrightarrow{\mathbb P} 0.
\]
\end{lemma}

\begin{proof}
We give a standard covering-number argument.
Let $r>0$ and let $\{\theta^1,\dots,\theta^N\}$ be an $r$-net of $\Theta$ in $\|\cdot\|_2$,
i.e.\ for every $\theta\in\Theta$ there exists $k$ with $\|\theta-\theta^k\|_2\le r$, where $N=N(\Theta,\|\cdot\|_2,r)$.
For each fixed $k$, the variables $\ell_{\theta^k}(X_i,Y_i,Z_i)$ are i.i.d.\ and bounded in $[-B_\ell,B_\ell]$
(Lemma~\ref{lem:ell_lip_cons}), so Hoeffding's inequality \citep{hoeffding1963}, yields for any $t>0$,
\[
\mathbb P\!\left(|\hat L_n(\theta^k)-L(\theta^k)|>t\right)
\le 2\exp\!\left(-\frac{nt^2}{2B_\ell^2}\right).
\]
A union bound over $k=1,\dots,N$ gives
\[
\mathbb P\!\left(\max_{1\le k\le N}|\hat L_n(\theta^k)-L(\theta^k)|>t\right)
\le 2N\exp\!\left(-\frac{nt^2}{2B_\ell^2}\right).
\]
Now fix any $\theta\in\Theta$ and pick $\theta^k$ with $\|\theta-\theta^k\|_2\le r$.
By Lemma~\ref{lem:ell_lip_cons}, both $\theta\mapsto \hat L_n(\theta)$ and $\theta\mapsto L(\theta)$ are $L_\ell$-Lipschitz:
\[
|\hat L_n(\theta)-\hat L_n(\theta^k)|\le L_\ell r,\qquad |L(\theta)-L(\theta^k)|\le L_\ell r.
\]
Therefore,
\[
|\hat L_n(\theta)-L(\theta)|
\le |\hat L_n(\theta^k)-L(\theta^k)| + |\hat L_n(\theta)-\hat L_n(\theta^k)| + |L(\theta)-L(\theta^k)|
\le \max_{j\le N}|\hat L_n(\theta^j)-L(\theta^j)| + 2L_\ell r.
\]
Taking the supremum over $\theta$ yields
\[
\sup_{\theta\in\Theta}|\hat L_n(\theta)-L(\theta)|
\le \max_{j\le N}|\hat L_n(\theta^j)-L(\theta^j)| + 2L_\ell r.
\]
Since $\Theta=[-M,M]^4$ has finite covering numbers (e.g.\ $N(\Theta,\|\cdot\|_2,r)\le (1+4M/r)^4$),
choose $r=r_n\downarrow0$ slowly so that the Hoeffding--union bound term vanishes, which implies the claim.
\end{proof}

The previous lemmas imply $\Delta_n=\sup_{\theta\in\Theta}|\hat L_n(\theta)-L(\theta)|=o_p(1)$.
We now conclude Theorem~\ref{thm:consistency_linear} by combining this uniform convergence with strict properness and identifiability.

\begin{proof}[Proof of Theorem~\ref{thm:consistency_linear}.]

By strict properness of the energy score (cf.\ the argument used for Proposition~\ref{prop:popul_sol_joint}),
for each fixed $z$ the conditional expected score
\[
\mathbb E\!\left[\ell_\theta(X,Y,z)\mid Z=z\right]
\]
is minimised (within the class) only when the model-implied conditional law matches the true conditional law,
i.e.
\[
\mathbb E\!\left[\ell_\theta(X,Y,z)\mid Z=z\right]
=\inf_{\vartheta\in\Theta}\mathbb E\!\left[\ell_\vartheta(X,Y,z)\mid Z=z\right]
\ \Longrightarrow\
P^\theta_{(X,Y)\mid Z=z}=P_{(X,Y)\mid Z=z}.
\]
Taking expectations over $Z$ yields that
\[
L(\theta)=\inf_{\vartheta\in\Theta}L(\vartheta)
\ \Longrightarrow\
P^\theta_{(X,Y)\mid Z=z}=P_{(X,Y)\mid Z=z}\ \text{for }P_Z\text{-a.e.\ }z.
\]
Under well-specification, there exists $\theta^*\in\Theta$ with
$P^{\theta^*}_{(X,Y)\mid Z=z}=P_{(X,Y)\mid Z=z}$ for $P_Z$-a.e.\ $z$, hence $\theta^*$ attains the minimum of $L$.

By identifiability within the class, this minimiser is unique.
Moreover, by Lemma~\ref{lem:ell_lip_cons}, $L$ is continuous on compact $\Theta$.

Fix $\varepsilon>0$ and define the closed subset
\[
A_\varepsilon := \{\theta\in\Theta:\|\theta-\theta^*\|_2\ge \varepsilon\}.
\]
Then $A_\varepsilon$ is compact and $L$ attains its minimum on $A_\varepsilon$.
Uniqueness of the minimiser over $\Theta$ implies
\[
\delta(\varepsilon):=\inf_{\theta\in A_\varepsilon}\big(L(\theta)-L(\theta^*)\big)>0,
\]
and therefore
\[
\|\theta-\theta^*\|_2\ge\varepsilon
\quad\Longrightarrow\quad
L(\theta)\ge L(\theta^*)+\delta(\varepsilon).
\]

Let
\[
\Delta_n := \sup_{\theta\in\Theta}\big|\hat L_n(\theta)-L(\theta)\big|.
\]
Fix $\varepsilon>0$ and define the event
\[
\mathcal B_n(\varepsilon):= \left\{\Delta_n \ge \tfrac{\delta(\varepsilon)}{2}\right\}.
\]
We show $\{\Delta_n<\delta(\varepsilon)/2\}\subseteq \{\|\hat\theta_n-\theta^*\|_2<\varepsilon\}$.
Indeed, suppose $\Delta_n<\delta(\varepsilon)/2$. Then for any $\theta$ with $\|\theta-\theta^*\|_2\ge\varepsilon$,
\[
\hat L_n(\theta)\ge L(\theta)-\Delta_n
\ge L(\theta^*)+\delta(\varepsilon)-\Delta_n
> L(\theta^*)+\tfrac{\delta(\varepsilon)}{2},
\]
while
\[
\hat L_n(\theta^*)\le L(\theta^*)+\Delta_n
< L(\theta^*)+\tfrac{\delta(\varepsilon)}{2}.
\]
Hence $\hat L_n(\theta)>\hat L_n(\theta^*)$ for all $\theta\in A_\varepsilon$.
Since $\hat\theta_n$ minimises $\hat L_n$, it follows that $\|\hat\theta_n-\theta^*\|_2<\varepsilon$.
Equivalently,
\[
\{\|\hat\theta_n-\theta^*\|_2\ge \varepsilon\}\subseteq \mathcal B_n(\varepsilon).
\]

Finally, by Lemma~\ref{lem:unifconv_cons}, $\Delta_n=o_p(1)$, so $\mathbb P(\mathcal B_n(\varepsilon))\to 0$.
Therefore
\[
\mathbb P(\|\hat\theta_n-\theta^*\|_2\ge\varepsilon)
\le \mathbb P(\mathcal B_n(\varepsilon))\to 0,
\]
which implies $\hat\theta_n\xrightarrow{p}\theta^*$.
\end{proof}

\section{Benchmark methods} \label{app:benchmark}

This section provides an overview of the benchmarked methods, their estimation capabilities, and details about their publicly available implementations: 

\begin{itemize}
    \item Control functions (CF), an IV method developed in econometrics \citep{heckman1976cf, newey1999control} for interventional mean estimation based on decomposing the hidden confounder into a treatment-correlated and an independent part. In the nonlinear version we use natural cubic splines for basis expansion. The algorithm is described by \citet{guo2016}.
	\item DeepIV \citep{hartford2017deepIV}:
    A flexible IV approach that employs neural networks, relying on the moment restriction (a `deep variant' of 2SLS). The method estimates the conditional density of $X$ given $Z=z$ in the first stage, and is limited to interventional mean estimation. Python implementation: \url{https://github.com/jhartford/DeepIV}.
	\item DeepGMM \citep{bennett2020deepGMM}: An IV method based on the generalised method of moments (GMM), which uses neural networks to learn a structural function that satisfies the moment restriction. The method aims to estimate the interventional mean. Python implementation: \url{https://github.com/CausalML/DeepGMM}.
    \item HSIC-X \citep{saengkyongam2022exploiting}: A neural network-based IV method that leverages the Hilbert-Schmidt Independence Criterion (HSIC) to enforce independence restrictions for valid instruments.  It is designed to estimate the interventional mean. Python implementation: \url{https://github.com/sorawitj/HSIC-X}.
	\item DIVE \citep{kook2024}: A distributional IV-based approach using independence restrictions, designed for estimating distributional causal effects with binary treatment and an absolutely continuous response \citep{kook2024}. Implemented in \texttt{R}, available at \url{https://github.com/LucasKook/dive}.
	\item IVQR \citep{chernozhukov2005}: An IV quantile regression framework for estimation of quantile treatment effects (QTEs) for binary treatment and absolutely continuous response. Linear IVQR is implemented using the \texttt{IVQR} \texttt{R} package, accessible at \url{https://github.com/yuchang0321/IVQR}.
	\item Engression \citep{shen2024engression}: A generative-modelling-based distributional regression method that minimises the energy loss to learn the conditional distribution of $Y|X=x$ \citep{shen2024engression}, which is not an IV method. Implementation in \texttt{R}, \texttt{engression} package: \url{https://cran.r-project.org/web/packages/engression}.
\end{itemize}

\section{Additional experiments}

\subsection{Finite-sample behaviour}\label{app:finite-sample}

To complement the fixed-sample comparisons in Sections~\ref{subsec:sim_cont_X}--\ref{subsec:sim_binary_X}, we study how DIV behaves as the training sample size increases. Throughout, we keep the model architecture and optimization settings fixed as in the corresponding main experiments, and repeat each configuration over $100$ independent simulation runs.

\subsubsection{Continuous treatment (Scenarios 1--6)}

For each of the six data-generating processes, we fit DIV on training samples of size
$n \in \{100, 500, 1000, 5000, 10000\}$.
Figure~\ref{fig:scen1-6-mse-samplesize} reports the mean MSE ($\pm$ SD) of the estimated interventional mean.
Across scenarios, the MSE decreases overall with increasing $n$, with only minor non-monotonic fluctuations in a few settings.

\begin{figure}[h!]
    \centering
    \includegraphics[width=1\textwidth,height=1\textheight,keepaspectratio]{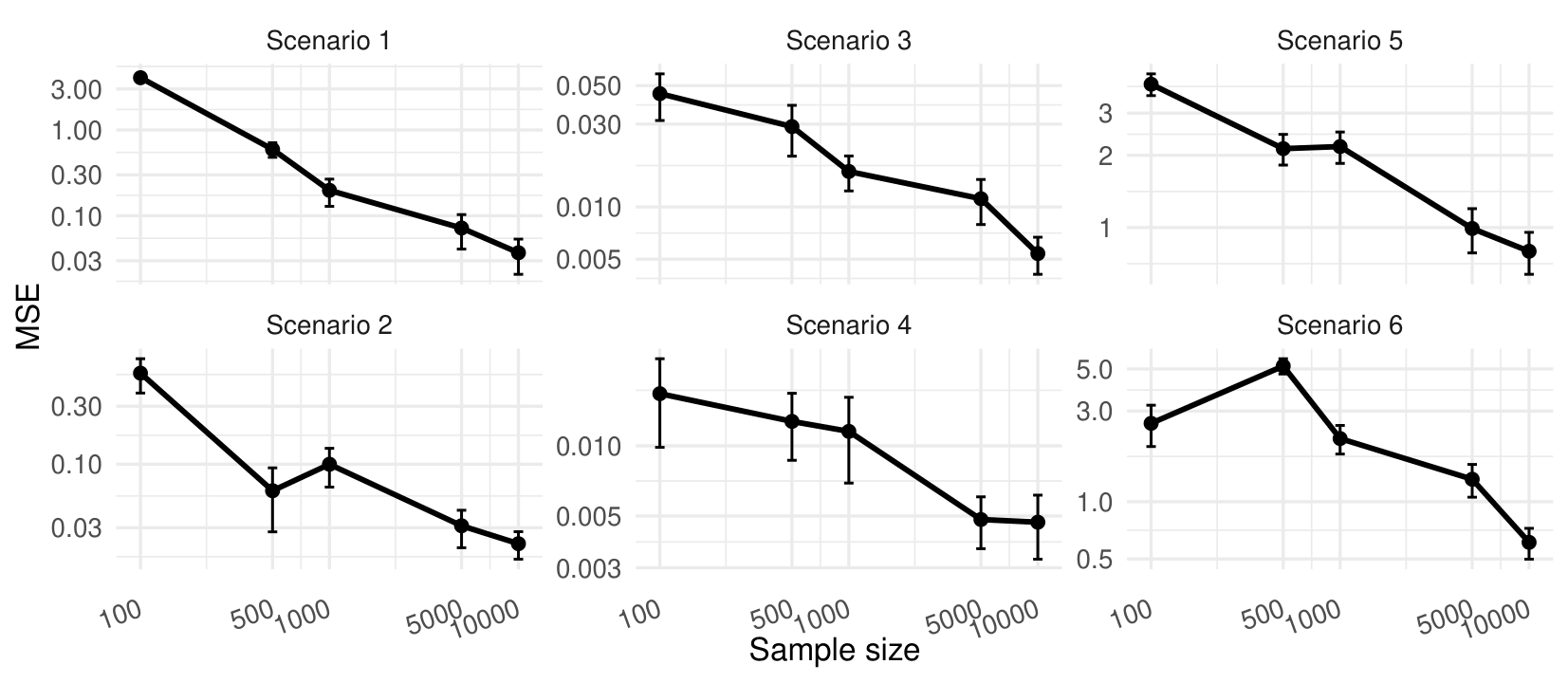}
    \caption{Mean MSE ($\pm$ SD) of the DIV estimator for the interventional mean across $100$ simulation runs as a function of the sample size. Top row: binary instrument $Z$; bottom row: continuous instrument $Z$. Columns correspond to Scenarios 1--2, 3--4, and 5--6, respectively.}
    \label{fig:scen1-6-mse-samplesize}
\end{figure}

We additionally assess distributional accuracy via the energy distance between the estimated and true interventional outcome distributions.
Figure~\ref{fig:scen1-6-energy-samplesize} reports the mean energy distance ($\pm$ SD) across the same sample sizes.

\begin{figure}[h!]
    \centering
    \includegraphics[width=1\textwidth,height=1\textheight,keepaspectratio]{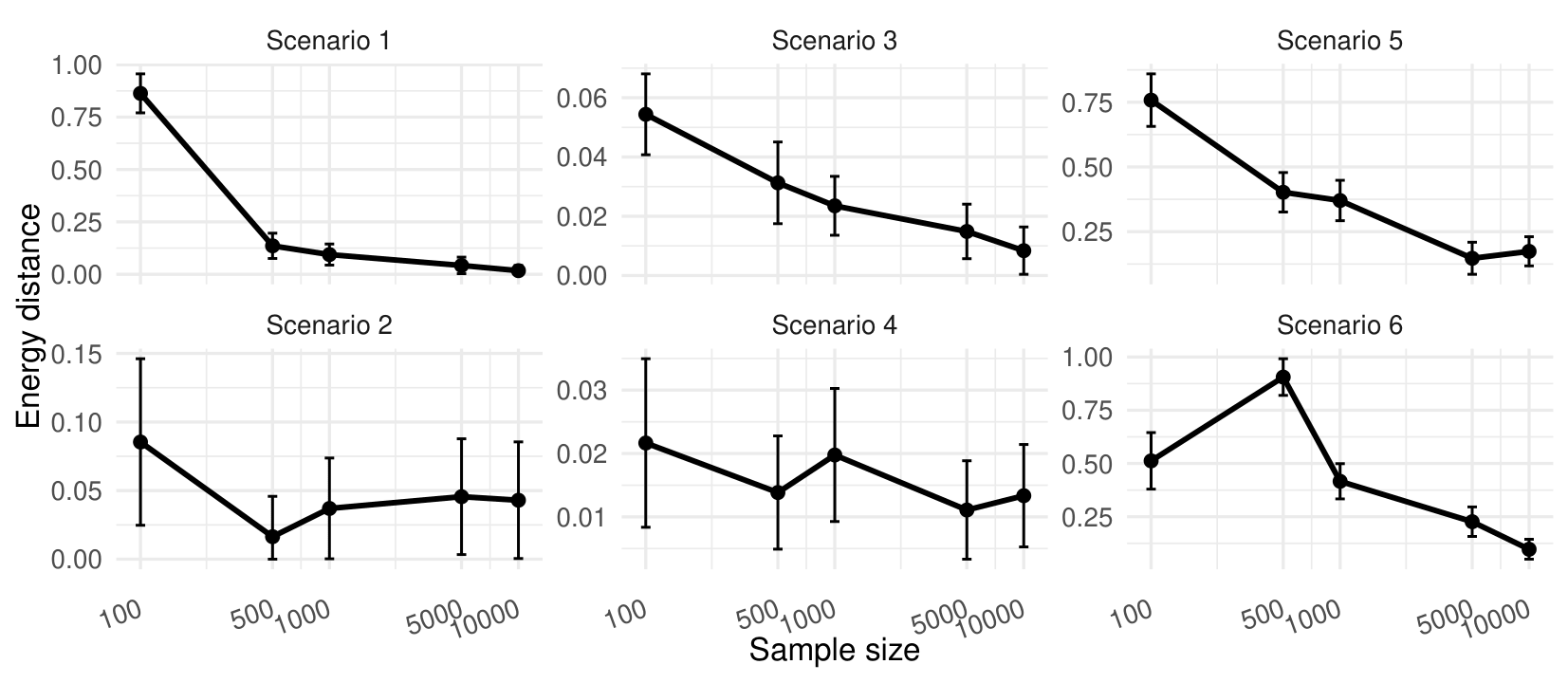}
    \caption{Mean energy distance ($\pm$ SD) between the estimated and true interventional outcome distributions across $100$ simulation runs as a function of the sample size. Top row: binary instrument $Z$; bottom row: continuous instrument $Z$. Columns correspond to Scenarios 1--2, 3--4, and 5--6, respectively.}
    \label{fig:scen1-6-energy-samplesize}
\end{figure}

\subsubsection{Binary treatment (QTE scenarios)}

For the two binary-treatment scenarios in Section~\ref{subsec:sim_binary_X}, we evaluate finite-sample behaviour using the RMSE of the estimated QTE curve across quantile levels.
Figure~\ref{fig:binX-rmse-samplesize} shows that the RMSE decreases with increasing $n$.

\begin{figure}[t!]
    \centering
    \includegraphics[width=1\textwidth,height=1\textheight,keepaspectratio]{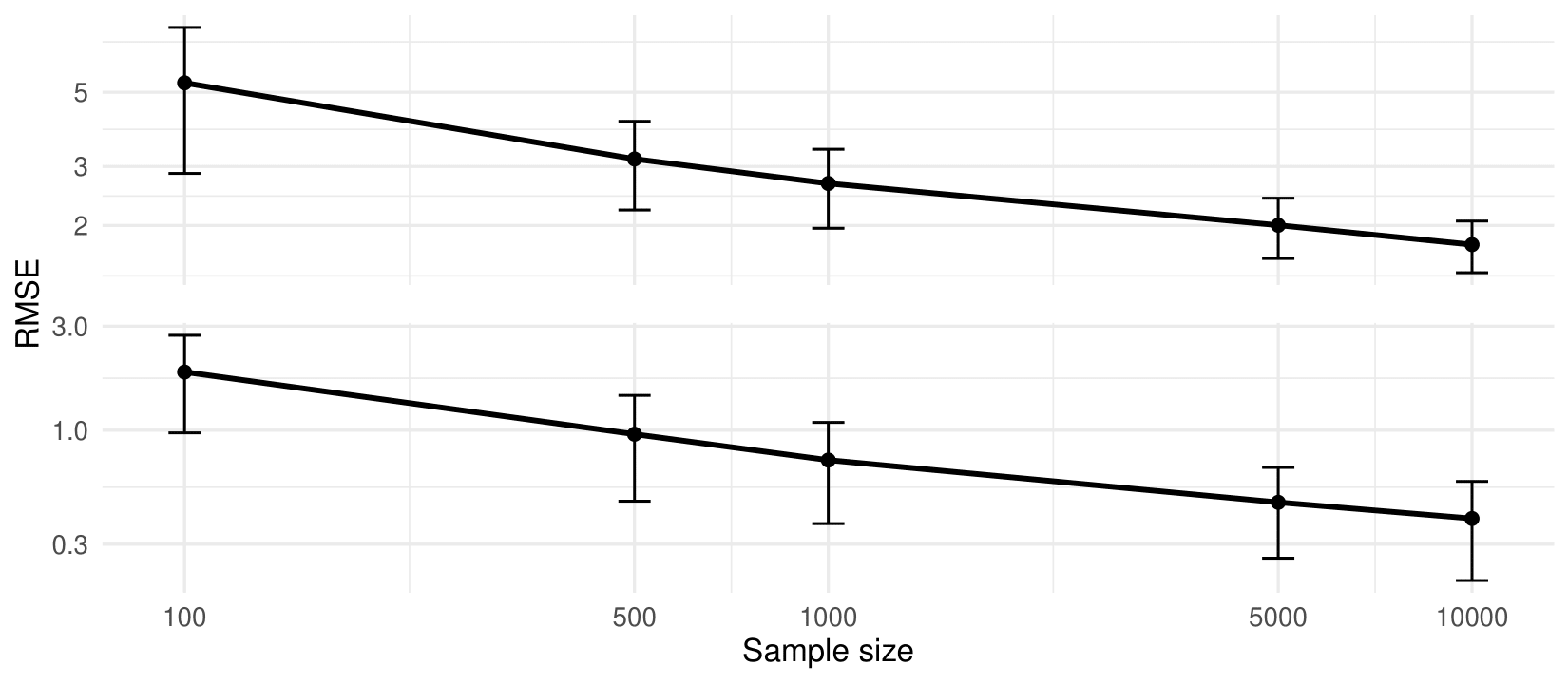}
    \caption{Mean RMSE ($\pm$ SD) of the DIV estimator for the QTE across $100$ simulation runs as a function of the sample size. Top: Scenario 1; bottom: Scenario 2.}
    \label{fig:binX-rmse-samplesize}
\end{figure}

Finally, we quantify distributional accuracy for both treatment states via the energy distance.
Figure~\ref{fig:binX-energy-samplesize} reports the mean energy distance ($\pm$ SD) between samples from the estimated and true interventional outcome distributions under $\mathrm{do}(X\! \coloneqq \! 0)$ and $\mathrm{do}(X \! \coloneqq \! 1)$.

\begin{figure}[h!]
    \centering
    \includegraphics[width=1\textwidth,height=1\textheight,keepaspectratio]{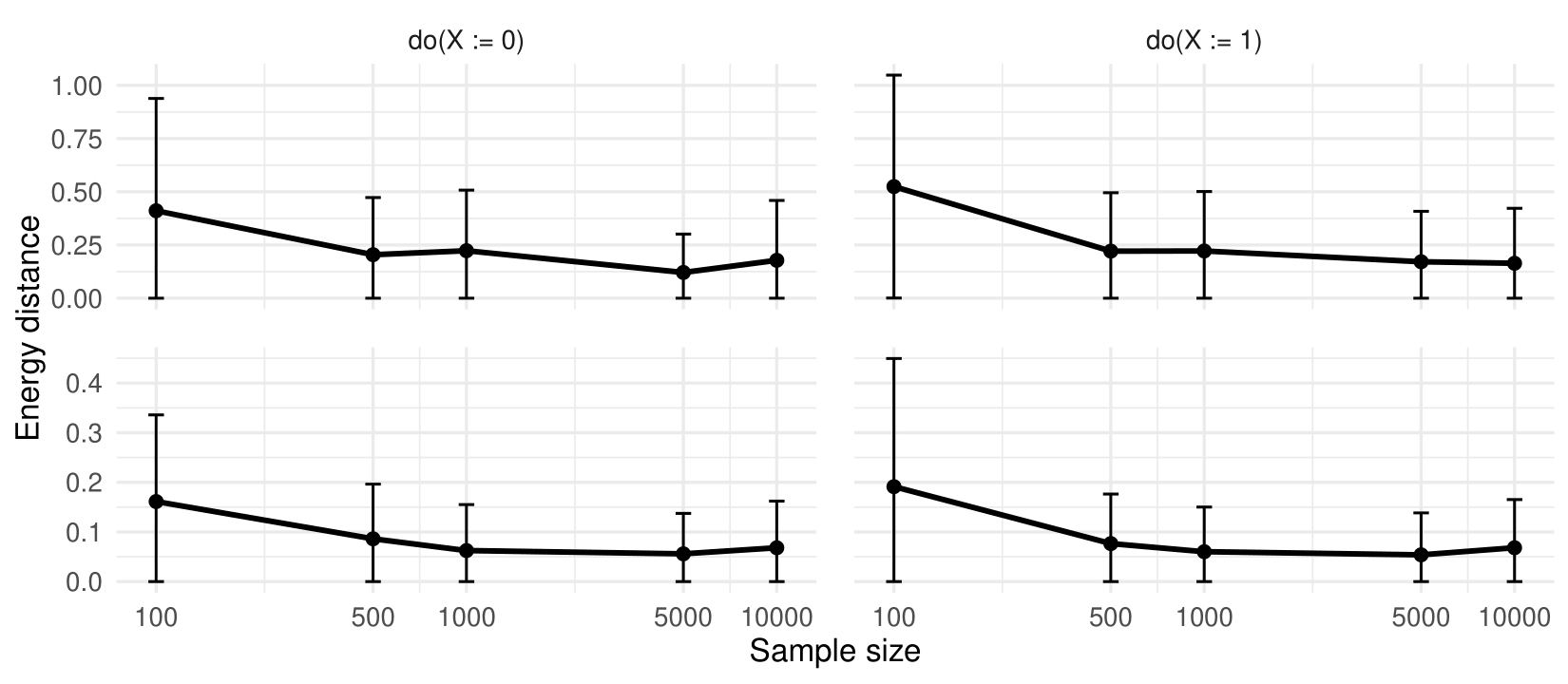}
    \caption{Mean energy distance ($\pm$ SD) between the estimated and true interventional outcome distributions obtained with DIV across $100$ simulation runs, as a function of the sample size. Left: $\mathrm{do}(X=0)$; right: $\mathrm{do}(X=1)$. Top: Scenario 1; bottom: Scenario 2.}
    \label{fig:binX-energy-samplesize}
\end{figure}

\subsection{Categorical Instrument with Unbalanced Support}
\label{subsec:knockout_unbalanced}

Categorical instruments with highly uneven support across levels are common in applications such as genetic perturbation experiments, where some intervention environments are observed frequently while others are rare. Such imbalance reduces effective first-stage information for certain treatments and can make causal estimation challenging. We study the behaviour of DIV in this setting using a controlled simulation design motivated by CRISPR knockout experiments. The design mirrors the single-cell CRISPR setting introduced in Section~\ref{subsec:singlecell}, where a categorical environment indicator represents gene-specific perturbations with highly unbalanced support and serves as an exogenous instrument inducing distributional shifts in $X$.

We observe a scalar outcome $Y$, a treatment vector $X=(X_1,\dots,X_4)$, and a categorical instrument $Z\in\{1,2,3,4\}$ with unbalanced probabilities $\mathbb P(Z=k) = (0.60,\,0.25,\,0.10,\,0.05)$ for $k\in\{1,2,3,4\}$, correspondingly.  
A latent confounder $H\sim N(0,1)$ affects both treatment and outcome. Conditional on $(H,Z=k)$, we generate a latent Gaussian vector $X^\ast \mid (H,Z=k) \sim N(\mu_0 + \Gamma H + \delta_k,\;\Sigma)$,
where $\Gamma\in\mathbb R^4$ controls the strength of confounding and $\delta_k\in\mathbb R^4$ encodes environment-specific perturbations. The covariance matrix $\Sigma$ has constant pairwise correlations to induce dependence across treatment components. The observed treatment $X$ is obtained by clipping each coordinate of $X^\ast$ to a bounded range, mimicking realistic gene expression levels. The outcome follows the linear structural equation $Y = \beta^\top X + H + \varepsilon_Y$, with $\varepsilon_Y \sim \mathcal N(0,1)$.

We consider knockout interventions of the form $\mathrm{do}(X_j \! \coloneqq \! 0)$ for $j=1,\dots,d$. DIV is trained on data of sample sizes $n\in\{100,500,1000,5000,10000\}$. We evaluate performance using the MSE of the estimated interventional mean function, averaged over $100$ simulation runs.

\begin{figure}[h!]
        \centering
\includegraphics[width=1\textwidth,height=1\textheight,keepaspectratio]{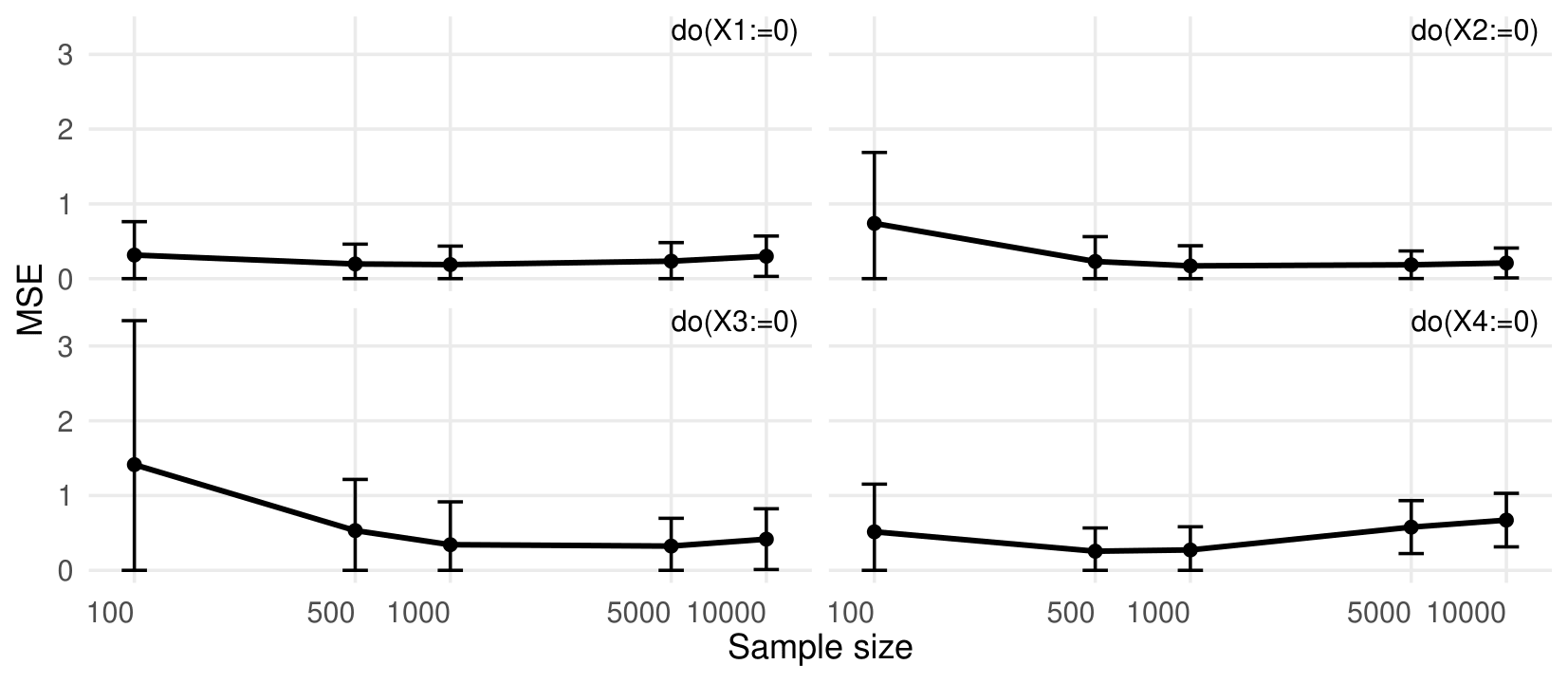}
        \caption{Mean MSE ($\pm$ SD; lower SD bound clipped at 0) of DIV of the estimated interventional mean across 100 simulation runs for varying sample sizes.}
        \label{fig:knockout_mse}
\end{figure}

Figure~\ref{fig:knockout_mse} shows that DIV yields stable interventional mean estimates across all knockout interventions, including those associated with rare instrument levels. While the MSE does not decrease monotonically with sample size for all interventions, it remains bounded and of comparable magnitude across $n$, reflecting the difficulty of extrapolative knockouts under imbalanced categorical instruments. Overall, the results indicate that DIV remains robust in this setting and can recover meaningful interventional effects even when identification relies on severely unbalanced categorical instruments.

\subsection{Observational and interventional distributions} \label{app:Pobs_Pint}

In Section~\ref{subsec:estimation}, we argue that DIV enables the estimation of the interventional distribution $\Pint$ along with its functionals. However, it is important to emphasize that DIV also provides an estimation of the joint observational distribution  $P_{(X,Y)}$  at no additional cost. To demonstrate empirical results, we now consider a setting where the treatment model is defined as $g(Z,H,\varepsilon_X) \coloneqq Z + H + 0.5\varepsilon_X$, and the outcome model is $f(X, H, \varepsilon_Y) \coloneqq X - 3H + 0.5\varepsilon_Y$, with $Z \sim \unif(0,3)$ and $H, \varepsilon_X, \varepsilon_Y \sim N(0,1)$. Figure~\ref{fig:obs_int} shows that DIV model manages to estimate both the observational and the interventional distributions well. Technically, for drawing a sample from the observational distribution $P_{(X,Y)}$, one needs to (i) sample the noise $\varepsilon_{H,i}, \varepsilon_{X,i}, \varepsilon_{Y,i}$ from standard Gaussians, (ii) obtain a sample $\hat{x}_i = \hat{g}(z_i, \varepsilon_{H,i}, \varepsilon_{X,i})$, and then (iii) obtain a sample $\hat{y}_i = \hat{f}(\hat{x}_i, \varepsilon_{H,i}, \varepsilon_{Y,i})$. The resulting set of pairs $(x_i,y_i), \ i = 1, \dots, n$, is an i.i.d. sample from the observational distribution $P_{(X,Y)}$.

\begin{figure}[h!]
        \centering
        \includegraphics[width=0.69\textwidth,height=0.69\textheight,keepaspectratio]{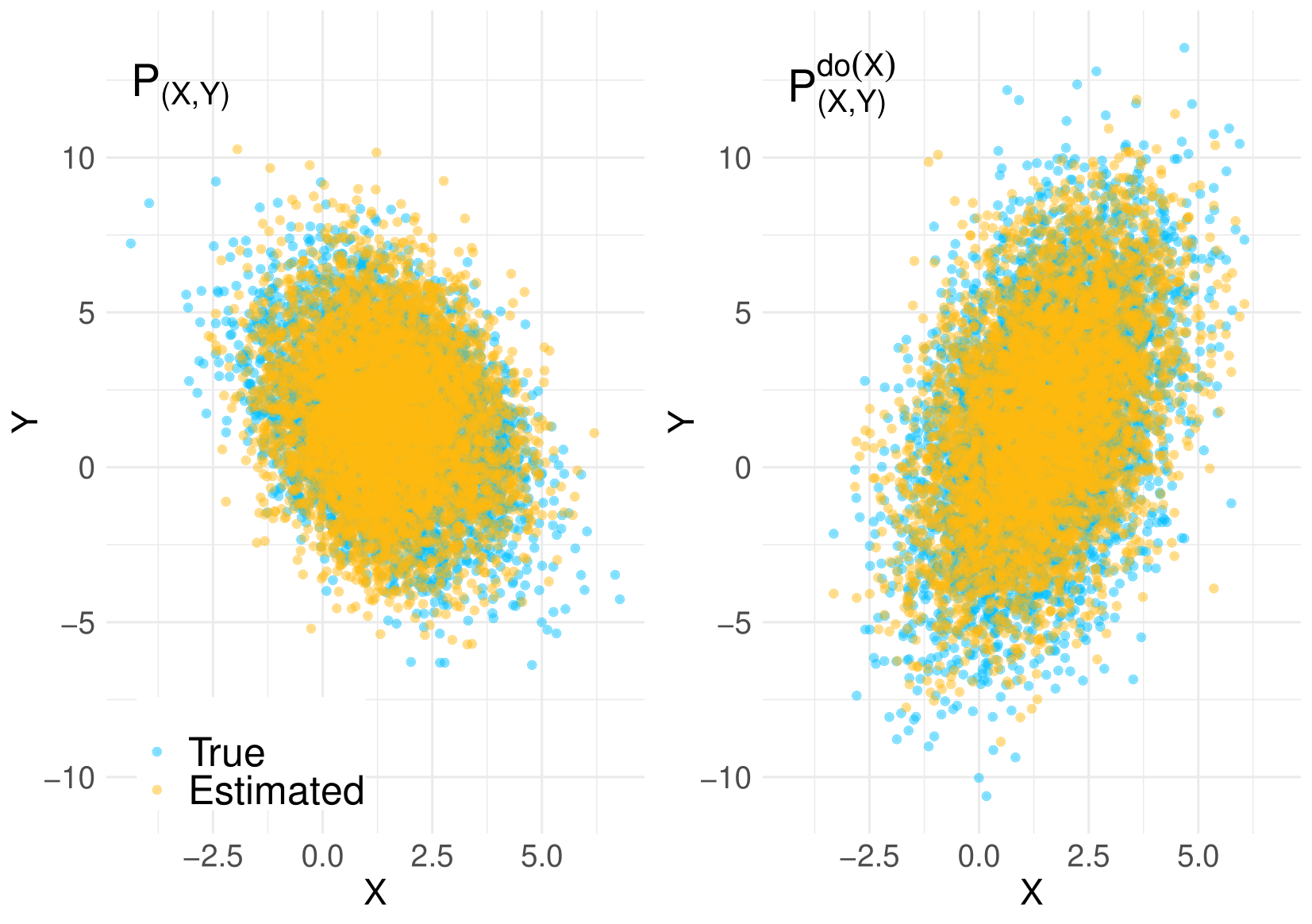}
        \caption{Samples from the observational distribution $P_{(X,Y)}$ (left) and the interventional distribution  $P_{(X,Y)}^{\mathrm{do}(X)}$ which is defined as $P_{(X,Y)}^{\mathrm{do}(X\coloneqq \tilde{X})}$ (right), with $\tilde{X} \eqdist X$.}
        \label{fig:obs_int}
\end{figure}

\end{document}

%% file: commands.tex
\DeclareMathOperator*{\argmin}{arg\,min}

\newcommand\eqdist{\stackrel{\mathclap{\normalfont\mbox{\small\textit{d}}}}{=}}
\newcommand\noneqdist{\stackrel{\mathclap{\normalfont\mbox{\small\textit{d}}}}{\neq}}

\def\bbE{\mathbb{E}}
\def\bbR{\mathbb{R}}

\def\cF{\mathcal{F}}
\def\cG{\mathcal{G}}
\def\cZ{\mathcal{Z}}

\def\cE{\mathcal{E}}

\def\supp{\mathrm{supp}}
\def\Logistic{\mathrm{Logistic}}
\def\qte{\mathrm{QTE}}
\def\Bernoulli{\mathrm{Bernoulli}}

\def\Pint{P_Y^{\mathrm{do}(X \coloneqq x)}}
\def\dim{\mathrm{dim}}